\documentclass[twocolumn,
english,
superscriptaddress,prl,amsmath,amssymb,floatfix,longbibliography
]{revtex4-1}
\usepackage[T1]{fontenc}
\usepackage[utf8]{inputenc}
\usepackage{color}
\usepackage[english]{babel}
\usepackage{amsmath}
\usepackage{graphicx}
\usepackage{esint}
\usepackage{float}
\usepackage[unicode=true,pdfusetitle,bookmarks=true,bookmarksnumbered=false,bookmarksopen=false,breaklinks=false,pdfborder={0 0 0},pdfborderstyle={},backref=false,colorlinks=true]{hyperref}

\makeatletter

\usepackage{dcolumn}
\usepackage{bm}
\usepackage{color}
\usepackage{soul}
\usepackage{amsfonts}
\usepackage{slashed}
\usepackage{enumerate}
\usepackage{mathtools}
\usepackage{physics}
\usepackage{amsthm}
\usepackage{inputenc}
\usepackage[normalem]{ulem}

\let\old@makecaption=\@makecaption
\usepackage{subcaption}
\let\@makecaption=\old@makecaption

\usepackage[figurename=FIG.]{caption}

 \newcommand{\ad}{\text{ad}}
 \newcommand{\hc}{\text{h.c.}}
 \newtheorem{theorem}{Theorem}[section]
\usepackage{epsfig}
\usepackage{dsfont}
\usepackage{xcolor}

\makeatother
\begin{document}

\title{Quantum state engineering by steering in the presence of errors}
\author{E. Medina-Guerra}
\affiliation{Department of Condensed Matter Physics, Weizmann Institute of Science, Rehovot 7610001, Israel}
\author{Parveen Kumar}
\affiliation{Department of Condensed Matter Physics, Weizmann Institute of Science, Rehovot 7610001, Israel}
\affiliation{Department of Physics, Indian Institute of Technology Jammu, Jammu 181221, India}
\author{I. V. Gornyi}
\affiliation{\mbox{Institute for Quantum Materials and Technologies, Karlsruhe Institute of Technology, 76021 Karlsruhe, Germany}}
\affiliation{\mbox{Institut für Theorie der Kondensierten Materie, Karlsruhe Institute of Technology, 76128 Karlsruhe, Germany}}

\author{Yuval Gefen}
\affiliation{Department of Condensed Matter Physics, Weizmann Institute of Science, Rehovot 7610001, Israel}
\date{\today}

\begin{abstract}
Quantum state engineering plays a vital role in various applications in the field of quantum information. Different strategies, including drive-and-dissipation, adiabatic cooling, and measurement-based steering, have been proposed for state generation and manipulation, each with \sout{its} upsides and downsides. Here, we address a class of measurement-based state engineering protocols where a sequence of generalized measurements is employed to steer a quantum system toward a desired (pure or mixed) target state. Previously studied measurement-based protocols relied on idealized procedures and avoided exploration of the effects of various errors stemming from imperfections of experimental realizations and external noise. We employ the quantum trajectory formalism to provide a detailed analysis of the robustness of these steering protocols against multiple classes of errors. We study a set of realistic errors that can be classified as dynamic or static, depending on whether they do or do not remain unchanged while running the protocol.
More specifically, we investigate the impact of the erroneous choice of detector-system coupling, erroneous re-initialization of the detector state following a measurement step, fluctuating steering directions, and environmentally induced errors in the detector-system interaction. We show that the protocol remains fully robust against the erroneous choice of detector-system coupling parameters and presents reasonable robustness against other types of errors. Our analysis employs various quantifiers such as fidelity, trace distance, and linear entropy to characterize the protocol's robustness and provide analytical results for these quantifiers against various errors. We introduce averaging hierarchies of stochastic equations describing individual quantum trajectories associated with detector readouts. Subsequently, we demonstrate the commutation between the classical expectation value and the time-ordering operator of the exponential of a Hamiltonian with multiplicative white noise, as well as the commutation of the expectation value and the partial trace with respect to detector outcomes. Our ideas are implemented and demonstrated for a specific class of steering platforms, addressing a single qubit.

\end{abstract}

\maketitle

\section{Introduction}
The field of quantum information processing and quantum simulations has progressed rapidly in the past few years, from theoretical studies to experimental setups where toy systems perform simple but practical tasks. Preparing a quantum system in a specific state is essential in most such tasks. Broad schemes for quantum state engineering include the application of feedback following a projective measurement, a thermalization process where the quantum system is cooled down to its ground state by coupling it to a cold reservoir, and the so-called quantum annealing. These approaches have challenges in the following sense: The first involves a feedback step that increases the circuit complexity. The second scheme requires a reservoir placed at nearly zero temperature to reduce thermal excitations, which is challenging, especially for larger systems. Quantum annealing calls for modifying a trivial Hamiltonian (for which the ground state is trivial and easy to prepare) adiabatically to a non-trivial Hamiltonian whose ground state could be used as a resource for quantum information protocols. In this context, we note that the performance of this approach is determined by the smallest gap encountered during the evolution. Indeed, some of the avoiding crossing gaps in its many-body spectrum were shown to be exponentially small \cite{yuv1PhysRevA.80.062326,yuv2altshuler_anderson_2010,yuv3PhysRevLett.101.147204,yuv4knysh_zero-temperature_2016,yuv5PhysRevLett.101.170503}.

Recently introduced measurement-based quantum-state engineering protocols \cite{roy2020measurement,kumar_engineering_2020-1,kumar2022optimized,herasimenko,hoke2023quantum,volya2023state}  overcome the challenges mentioned above (see also Refs. \cite{pechen2006quantum,roa2006measurement,roa_quantum_2007,ashhab2010control} for early related ideas). These steering \footnote{Throughout this paper, we use the word \emph{steering} as the name of a process that leads the quantum system starting from an arbitrary initial state to the predesignated target state. This should not be confused with \emph{steering} from the theory of quantum measurements and quantum information, which defines a special kind of non-local correlations (see, e.g., \cite{wiseman_steering_2007,uola2020quantum})
} protocols employ a sequence of generalized measurements to ``steer'' a quantum system toward the predesignated target state. A generalized measurement comprises two steps
\cite{von2018mathematical, wiseman2009quantum,wheeler2014quantum}: (i) coupling the quantum system to an ancillary quantum system (also referred to as a \emph{detector}) by means of an interaction Hamiltonian, resulting in the unitary evolution of the joint state of the system and the detector, and (ii) a projective measurement of the detector, which disentangles the joint state and induces a measurement backaction to the system state. Usually, measurement-induced backaction is considered an undesired effect, as the primary purpose of quantum measurement is to extract information about the quantum state. Following a contrarian paradigm, the measurement-induced state engineering protocols utilize this measurement backaction in a controlled manner to guide the system toward the desired target state. Note that the convergence to the target state in the measurement-induced steering protocols is achieved by the measurement only, as compared to the drive-and-dissipation protocols
\cite{kraus2008preparation,roncaglia2010pfaffian,diehl2011topology,pechen2011engineering,murch2012cavity,leghtas2013stabilizing,shankar2013autonomously,lin2013dissipative,liu2016comparing,goldman2016topological,lu2017universal,huang2018universal,horn2018quantum}, where the relaxation is due to an uncontrolled dissipative environment.

Apart from the built-in challenges of circuit complexity and controlled dissipative environments, there are external noises arising from the imperfect isolation of a quantum system from its surroundings: quantum information processors are susceptible to such noisy environments. As long as these sources of imperfection remain uncontrolled, as is the current state of technological development, it is vital to understand the dynamics of quantum systems in the presence of external noises. This brought to life the concept of the \emph{Noisy Intermediate-Scale Quantum} (NISQ) era \cite{preskill2018quantum,bharti2022noisy}, which refers to quantum technologies based on devices composed of hundreds of noisy qubits. Moreover, within the NISQ paradigm, inevitable noise can be employed to achieve particular goals in quantum device functionality. Thus, analyzing the effects of noise and associated errors in quantum state engineering is one of the most pressing issues of the NISQ era.

Looking beyond the perspective of the NISQ era, one may design an appropriate quantum error correction scheme to protect a quantum state against environmentally induced errors. The quantum error correction strategy is to encode quantum information in a larger Hilbert space redundantly. This will ensure that the logical qubits experience a significantly lower error rate than what the physical qubits do
\cite{lidar2013quantum,nielsen2002quantum}. Various attempts in different experimental setups to implement error correction procedures have been discussed in the literature, e.g., liquid \cite{cory1998experimental,knill2001benchmarking,boulant2005experimental} and solid-state NMR \cite{moussa2011demonstration} trapped ions \cite{chiaverini2004realization,schindler2011experimental}, photon modes \cite{pittman2005demonstration}, superconducting qubits \cite{google2023,reed2012realization,kelly2015state}, and NV centers in diamond \cite{waldherr2014quantum,taminiau2014universal}. In general, it is instructive to understand the error mechanism as much as possible for practical applications before designing any error correction scheme.

    Our work investigates the impact of the different classes of errors that may affect a measurement-based state steering protocol by various experimental imperfections. We distinguish between two types of error, static and dynamic, depending on whether they change during the run time of the protocol. For most of this paper, we base our analysis on the quantum trajectory formalism \cite{barchielli1991measurements,wiseman2009quantum,patel2017weak} instead of describing the influence of the errors on the steered states via completely positive trace-preserving (CPTP) maps \cite{wiseman2009quantum}. The motivation behind implementing stochastic quantum trajectories lies in the amount of information about the system, which is lost once averages are performed.  This information is important, for example, for observables that are non-linear in the system density matrix. In particular, measurement-induced entanglement entropy transitions are captured by keeping the individual quantum trajectories and are invisible upon an ensemble average \cite{zerba,turkeshi2021measurement}. Furthermore, the trajectory-resolved evolution is the cornerstone for active-decision steering protocols \cite{herasimenko}, where the measurement outcomes are used to navigate the system toward the target state. Thus, although we eventually average over quantum trajectories, our formalism paves the way for the above-mentioned applications. 
    
    With this at hand, we derive novel stochastic differential equations that govern the system dynamics in the presence of errors. We employ several quantifiers: fidelity, trace distance, and linear entropy, to characterize the robustness of the state engineering protocols against errors.  In order to demonstrate the applicability of this approach, we address here the simplest example: steering a single qubit to a predesignated target state. We provide analytical results for the robustness quantifiers and show excellent agreement between the analytical expressions and the numerical results obtained from simulating stochastic quantum trajectories.

Specifically, we analyze the robustness of a steering protocol against errors due to erroneously chosen detector-system coupling, wrongly prepared detectors, fluctuating steering directions, environmentally induced errors in the detector-system interaction Hamiltonian, and fluctuating measurement directions. Our analysis of environmentally-induced fluctuations in the detector-system interaction Hamiltonian can mimic errors due to a fluctuating background field in an experimental setup and generalizes the Langevin stochastic Schrödinger equation [cf. Eq.~\eqref{eq:esto1}]. Consequently, we demonstrate the non-obvious fact that the average detector outcomes and stochastic white noise commute. This argument is based on the commonly used yet not formally proven fact of the commutativity of the time-ordering operator with the average over stochastic noise (or average over noise realizations). These observations (proven in Appendices \ref{sec:commutation_expectation} and \ref{sec:lastproof}) facilitate the derivation of three novel stochastic master equations, alluding to different hierarchies of averaging over stochastic processes of different origin [cf. Eqs.~\eqref{eq:sde1}, \eqref{eq:ssde5} and \eqref{eq:sde11}].

We thus derive stochastic differential equations governing the system's dynamics in specific scenarios.
From this perspective, this part of our analysis generalizes the model of repeated interactions of a system with a set of detectors \cite{attal_repeated_2006,attal2010stochastic} by including error and noise-induced stochasticity.
We show further that the steering protocol remains fully robust against errors due to erroneously chosen detector-system coupling parameters and erroneously chosen measurement directions. Compared to these errors, robustness against other errors is more moderate.

The paper is organized as follows. In Sec.~\ref{sec:ideal_protocol_and_setup}, we revisit the theme of state engineering  protocol introduced in Ref.~\cite{roy2020measurement}, ignoring any errors (``ideal steering''). We provide a detailed description of the measurement model and derive the stochastic master equation that describes the individual quantum trajectories followed by the steered system. Next, starting with the paradigm of blind measurement \footnote{The term \emph{blind measurement} refers to a measurement where the final state of the detector (the readout) is traced out, discarding the acquirement of information through the measurement process but accounting for measurement backaction.} procedure, we derive the Lindblad master equation describing the evolution of the steered state in the continuum time limit. We introduce the errors and their quantifiers to be studied throughout this paper in Sec.~\ref{sec:definition_of_the_errors}. The analysis of static errors is presented in Sec.~\ref{sec:static_errors}, where we discuss two types of such errors: one is caused by an incorrect choice of the interaction strength between the detector-system Hamiltonian (Sec.~\ref{ssec:erroneous_coupling_parameter}), and the other is caused by an error in the preparation of detector states  (Sec.~\ref{ssec:errors_in_the_detector}). A discussion of dynamic errors is presented in Sec.~\ref{sec:dynamic_errors}. We deal with four types of such errors. The first concerns fluctuating steering directions in the sense that the direction in which the protocol steers the system is stochastically altered following each protocol step (Sec.~\ref{ssec:errors_steering_direction}); the second refers to temporally fluctuating detector-system coupling strength (Sec.~\ref{ssec:errors_coupling_constant}); the third concerns errors in the steering Hamiltonian itself (Sec.~\ref{ssec:steering_hamiltonian}); and the last dynamic error involves errors in the direction the detector is projected (Sec.~\ref{ssec:errros_in_the_measurement_direction}). We provide conclusions and prospects in Sec.~\ref{sec:conclusions}.

 \begin{figure*}[t!]
     \centering
    \includegraphics[width=\linewidth]{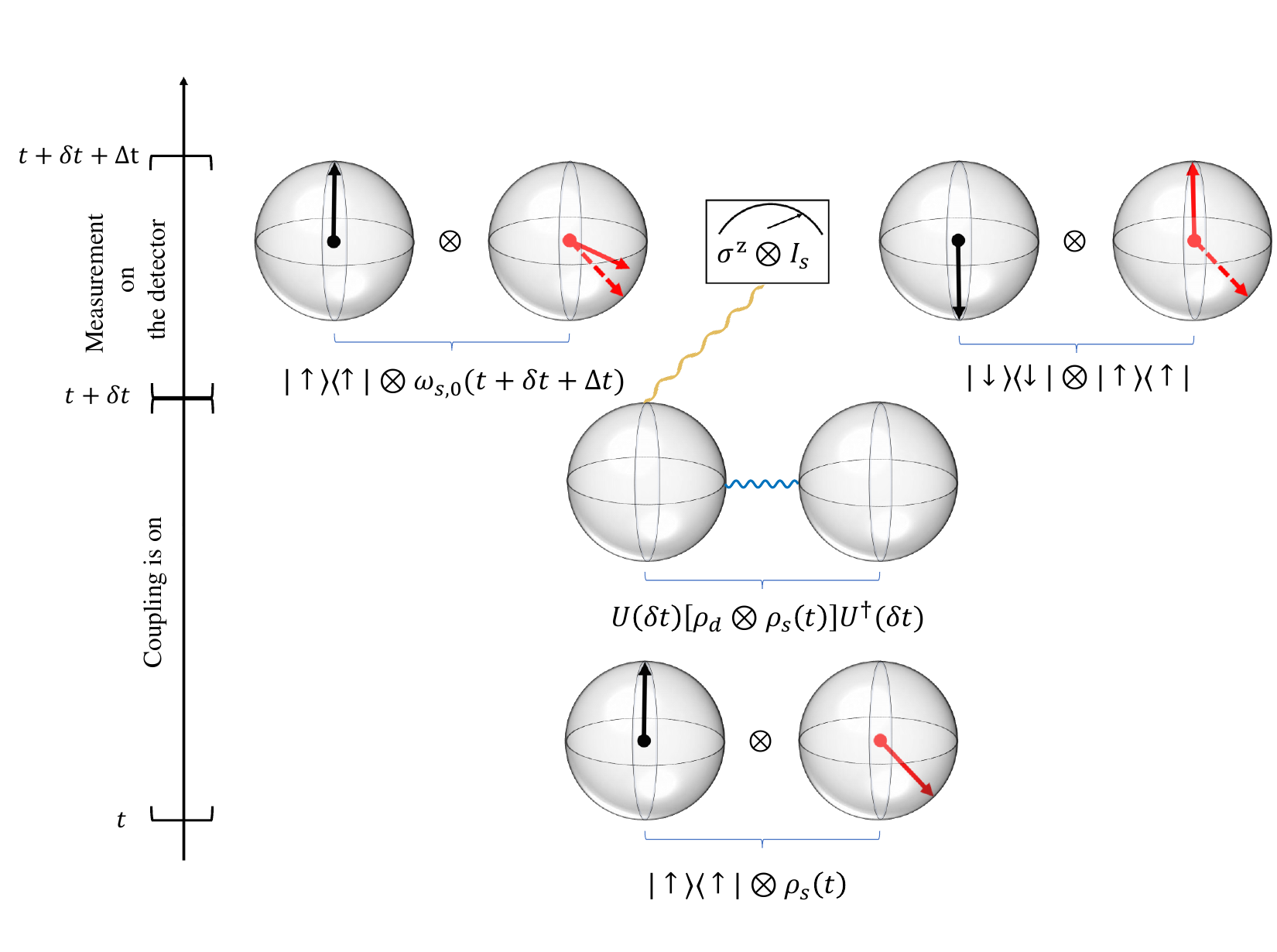}
        \caption{
        Steering step of a single qubit toward the target state $\rho_\oplus = \ket{\uparrow}\bra{\uparrow}$.
   In step (i), at time $t$, the total system is formed by two non-interacting qubits: the detector (black Bloch vector), initialized in the state $\rho_d = \ket{\uparrow}\bra{\uparrow}$, and the steered system represented by the state $\rho_s(t)$ (red Bloch vector).  In step (ii), from time $t$ to $t+\delta t$, the two subsystems interact via the Hamiltonian $H_\text{ds} = J(\sigma^+\otimes\sigma^-+\hc)$ [cf. Eq.~\eqref{eq:01}] and become correlated, which  is represented by the blue wiggly line. The joint state is then $\rho_{ds}(t+ \delta t) =U(\delta t)\rho_d\otimes \rho_s(t)U^\dagger(\delta t)$ with $U(\delta t) = \exp(-i H_\text{ds}\delta t)$. Step (iii) occurs at time $t + \delta t + \Delta t$ where a local projective measurement is performed over the detector to get one of the two observables of $\sigma^z\otimes I_s$. The wiggly yellow line represents this measurement. If the detector outcome gives $\ket{\uparrow}$ (a no-click), the vector of the steered system gets a nudge that continuously evolved from the previous state (dashed red Bloch vector), giving $\omega_{s,0}(t+\delta t + \Delta t)$. Instead, if the detector state gives $\ket{\downarrow}$ (a click), we find that the system jumped toward the target state (the north pole of the Bloch sphere). It is assumed that the time it takes to perform the projective measurement $\Delta t$ vanishes. These steps are repeated several times, and an unbiased average over multiple protocol runs (a blind measurement) is performed. }
    \label{fig:003}
\end{figure*}

\section{Ideal protocol}\label{sec:ideal_protocol_and_setup}
In this section, we briefly describe the quantum steering protocol introduced in Ref.~\cite{roy2020measurement}, which we call the ``ideal protocol'' (i.e., the protocol without any errors), applied to a single spin-1/2 system (or qubit). We present the protocol within the formalism of quantum trajectories, where the steered system follows a stochastic quantum evolution depending on the detector readouts. This inherently discrete-time evolution becomes a continuous stochastic process after adopting the  \emph{weak-measurement} (WM) limit, which requires an appropriate rescaling of the detector-system coupling constant \cite{roy2020measurement}. When the detector readouts are discarded (i.e., traced out), a procedure we denote as \emph{blind measurement}, the system follows dissipative dynamics, and its evolution is governed by the Lindblad equation (LE) in the WM limit. 

\subsection{Steps of measurement-based steering}
The ideal protocol involves implementing the following iterative steps to steer the quantum state of a system $\rho_s(t)$ toward the predesignated \emph{target} state $\rho_\oplus$ (see Fig. \ref{fig:003}):

\begin{enumerate}[(i)]
    \item At given fixed time $t'$, a quantum system is described by the state $\rho_s(t')$ and a quantum detector is prepared in state $\rho_d=\ket{\Phi_d}\bra{\Phi_d}$. In what follows, we will focus on the case of the simplest detector---a qubit---with a two-dimensional Hilbert space.  
    The detector-system state is represented by  $\rho_{ds}(t')=\rho_d\otimes\rho_s(t')$, as the subsystems are assumed not to interact. 
    
    \item In order to perform the measurement, the system is coupled to the detector using the Hamiltonian
    \begin{equation}\label{eq:01}
        H_\text{ds} \coloneqq J\ket*{\Phi_d^\perp}\bra{\Phi_d}\otimes A+ \hc, 
    \end{equation}
    where ``$\hc$'' stands for the Hermitian conjugate, $J$ is the  coupling constant, the detector state $\ket*{\Phi_d^\perp}$ is orthogonal to $\ket{\Phi_d}$, and $A$ is an operator satisfying $A\ket{\Psi_\oplus} = 0$ and $AA^\dagger \ket{\Psi_\oplus} = \ket{\Psi_\oplus}$.  

    \item Subsequently, at time $t' + \delta t + \Delta t$, the detector is measured projectively using the observable 
    \begin{equation}
    S_d \coloneqq \ket{\Phi_d}\bra{\Phi_d}-\ket*{\Phi_d^\perp}\bra*{\Phi_d^\perp}.
    \label{eq:Sd}
    \end{equation}
    This local projective measurement disentangles the joint detector-system state and creates a measurement backaction on the system state (see below). 
    
    Let $\ket{0} \coloneqq \ket{\Phi_d}$ and $\ket{1} \coloneqq \ket*{\Phi_d^\perp}$. After measuring the detector---a process that takes a $\Delta t$ time---the resultant detector-system state is given by 
    \begin{equation*}
    \qquad \rho_{\text{ds},\alpha}(t'+\delta t+ \Delta t)
        =\ket{\alpha}\bra{\alpha} \otimes \omega_{s,\alpha}(t' +\delta t + \Delta t),
    \end{equation*}
    where
    \begin{equation}\label{eq:04}
        \omega_{s,\alpha}(t' + \delta t + \Delta t) = \frac{M_\alpha(\delta t)\rho_s(t') M_\alpha^\dagger(\delta t)}{P(\alpha)}
    \end{equation}
    is the updated steered state and $\alpha \in \{0,1\}$. In the preceding equation, 
\begin{equation}\label{eq:04.1}
    M_\alpha(\delta t) \coloneqq \bra{\alpha}\exp(-i H_\text{ds} \delta t)\ket{\Phi_d}
\end{equation}
 is the generalized measurement operator (or Kraus operator)  \cite{wiseman2009quantum,breuer2002theory,jacobs2014quantum} representing the obtained measurement outcome,  occurring with probability $P(\alpha)$.  Note that the probability conservation condition $\sum_\alpha P(\alpha)=1$ imposes a constraint on the Kraus operators such that $\sum_\alpha M_\alpha^\dagger M_\alpha= I_s$.  Should $\alpha = 1$, we define  that the detector measurement showed a ``click'' result. On the other hand,  $\alpha = 0$ corresponds to a  ``no-click'' readout. 
 
 Henceforth, following the axioms of quantum mechanics \cite{cohen1986quantum}, we shall assume that the time it takes to perform the measurement $\Delta t$ vanishes.  
    \item  The detector is reset to its initial state $\rho_d$, and steps (i)-(iii) are repeated for the subsequent measurement.
\end{enumerate}

When formulating the ideal protocol, we have assumed that the energy scale $J_d$ of the detector Hamiltonian is way smaller than the measurement strength, i.e., $J_d/(J^2\delta t) \ll 1$, so we can waive the influence of the detector's Hamiltonian. The system Hamiltonian is trivially added to any stochastic master equation and thus to the corresponding Lindblad equation, but the detector Hamiltonian enters in a non-trivial way when one performs averages.  For the sake of simplicity, we also ignore the self-evolution of the steered system assuming $J_s\delta t\ll 1$ as we focus on the derivation and structure of stochastic and dissipative terms in the master equations. $J_s$ is the energy scale of the system Hamiltonian.

Steps (i)-(iv) can be conveniently gathered in the following discrete-time stochastic master equation (SME): 
\begin{equation}\label{eq:04.2}
    \delta \omega_s(t) = \sum_{\alpha = 0,1}\frac{M_\alpha(\delta t)\omega_s(t) M_\alpha^\dagger (\delta t)}{\expval*{M_\alpha^\dagger (\delta t) M_\alpha (\delta t)}_{\delta t}}\delta N_\alpha(t)-\rho_s(t),
\end{equation}
where $t\geq t'$, $\delta \omega_s(t) \coloneqq \omega_s(t+\delta t)- \omega_s(t)$, $\omega_s(t') = \rho_s(t)$, $\delta N_\alpha(t)$, for fixed $\alpha$, is an indicator function appearing with probability $$P(\alpha) = \expval*{M_\alpha^\dagger (\delta t) M_\alpha (\delta t)}_{\delta t}.$$ 

As an example, the steps (i)-(iii) are illustrated schematically in Fig.~\ref{fig:003}, where the system's target state is the north pole of the Bloch sphere $\ket{\Psi_\oplus} = \ket{\uparrow}$ (with $\sigma^z\ket{\uparrow} = \ket{\uparrow}$), the detector is prepared in the state $\ket{\uparrow}$, and the interaction Hamiltonian is 
\begin{equation}
H_\text{ds} =J( \sigma^+\otimes \sigma^- + \hc), 
\label{eq:Hds-proj}
\end{equation}
where $\sigma^+ = \ket{\uparrow}\bra{\downarrow}$, $\sigma^- = \ket{\downarrow}\bra{\uparrow}$, and $\sigma^z\ket{\downarrow} = -\ket{\downarrow}$.
A click measurement 
immediately projects the system to the pure state in this situation. In the example illustrated in Fig.~\ref{fig:003}, the resulting state is the desired target state--the north pole of the Bloch sphere. When a no-click measurement is obtained, the system is shown to have evolved continuously and irreversibly toward the target state. We elaborate further about these two types of evolution in Sec. \ref{ssec:wml}.

For the blind measurement, we discard the individual outcomes by taking the partial trace over the detector degrees of freedom, and therefore, the system state evolution can be written as
\begin{subequations}
\label{eq7}
    \begin{align}
        \rho_s(t'+\delta t) &\coloneqq \sum_{\alpha= 0,1}P(\alpha) \omega_{s,\alpha}(t'+\delta t) \\
        &=\sum_{\alpha= 0,1}M_\alpha(\delta t)\rho_s(t') M_\alpha^\dagger(\delta t).\label{eq:007}
    \end{align}
    \end{subequations}
Taking the partial trace over the detector's Hilbert space is a non-selective measurement; that is, the results of the detector are not read. Clearly, taking the average over the two possible outcomes in Eq.~\eqref{eq:04.2}, given the prior state $\rho_s(t')$, coincides with Eq.~\eqref{eq7}.

Under certain conditions (see Sec.~\ref{ssec:wml}), the updated state under a blind measurement,
Eq.~\eqref{eq7}, given the prior state $\rho_s(t')$, is $$\rho_s(t'+\delta t) = \rho_s(t')+ \mathcal{L}\rho_s(t')\delta t + \mathcal{O}(\delta t'^2),$$ where $\mathcal{L}$ is the superoperator generating dissipative Markovian dynamics, and Eq.~\eqref{eq:04.2} acquires the form of the well-known stochastic master equation of jump type [cf. Eq. \eqref{eq:07}].

\subsection{Weak-measurement limit}\label{ssec:wml}

Depending on the form of the detector-system Hamiltonian, Eq.~\eqref{eq:01}, there might be scenarios in which performing local projective measurements on the detectors leads to local measurement operators [see Eq.~\eqref{eq:04.1}] acting either as projectors or other quantum operations such as bit-flips and reflections. For example, performing a blind measurement [cf Eq.~\eqref{eq7}] fails to steer a prior state toward a state closer to the target state when the product 
$J \delta t$ in Eq.~\eqref{eq:04.1} is a multiple of $\pi$. To see this, let us express the measurement operators in Eq.~\eqref{eq:04.1} with respect to the orthonormal basis (ONB)
$\mathcal{B}_\oplus \coloneqq \{\ket{ \Psi_\oplus}, \ket*{\Psi_\oplus^\perp}\} = \{ \ket{\uparrow}, \ket{\downarrow}\}$:
\begin{align}
    M_0(\delta t) &= \begin{pmatrix}1 & 0 \\ 
    0 & \cos(J\delta t)\end{pmatrix}, \label{eq:yu1}    \\
    M_1(\delta t) &= \begin{pmatrix}0 & \sin(J\delta t) \\ 
    0 & 0\end{pmatrix}. \label{eq:yu1-1}
\end{align}
If $J\delta t / \pi \in \mathbb{Z}$, we would be either implementing a quantum operation on the steered state equivalent to a $\sigma^z$-gate \cite{nielsen2002quantum}, i.e., $M_0(\delta t) = \sigma^z$ and $M_1(\delta t) = 0$, or an identity, i.e., $M_0(\delta t) = I_s$ and $ M_1(\delta t) = 0$,  so no steering occurs. 

On the other hand, if $J\delta t /(\pi/2) \in \mathbb{Z}/2\mathbb{Z}$, one of the measurement operators becomes a projector $M_0(\delta t) = \ket{\uparrow}\bra{\uparrow}$,  and the other a raising operator $M_1(\delta t ) = \sigma^+ = \ket{\uparrow}\bra{\downarrow}$. We refer to the latter condition on $J\delta t$ as a \emph{strong-measurement limit}. Regardless of the click result, the target state $\rho_\oplus = \ket{\uparrow}\bra{\uparrow}$ is always reached in a single measurement step.

Although a qubit may reach the target state within the strong-measurement limit, fine-tuning of the product $J\delta t$ is required, which is a significant disadvantage. When steering a many-body system toward an entangled state using the ideal protocol described here, numerous detectors acting on different parts of the system are needed, and the Hamiltonians associated with those local steering operations may be non-commuting. In the fine-tuned protocol where each measurement is a strong one, a given measurement may, therefore, undermine the steering efficiency of a subsequent step; that is, different effective projective measurements on different regions of the system may undermine the propagation of entanglement that might be needed if the target state has a given degree of entanglement.

In order to overcome possible problems of this sort, we present the scenario in which both the interaction strength $J$ and frequency of measurements $1/\delta t$ of a qubit interacting with a given chain of detectors grow such that the product $J^2 \delta t$ remains constant as $\delta t$ tends to zero. This measurement regime is known as the \emph{weak-measurement} (WM) limit \cite{roy2020measurement}, and, as we shall see in brief, it avoids the fine-tuning and probabilistic issues encountered before, ensuring successful steering in a finite time.

Let measurement strength (or coupling strength) be
\begin{equation}\label{eq:02}
    J = \sqrt{\frac{\gamma}{\delta t}},
\end{equation}
where $\gamma > 0$ will have a meaning of the measurement or channel strength (see below). Now, the WM limit amounts to 
setting the value of measurement time-step small enough such that $J\delta t = \sqrt{\gamma \delta t} \ll 1$. In addition, the following condition is assumed:
\begin{equation}
\lim_{\delta t \to 0}J^2 \delta t = \gamma = \text{const}.
\label{eq:02-1}
\end{equation}
After performing a series expansion in $\delta t$ and setting it as a true differential, the discrete-time SME in Eq.~\eqref{eq:04.2} becomes the continuous-time SME of the jump type
\cite{barchielli1991measurements,wiseman2009quantum,breuer2002theory,attal2010stochastic, gross_qubit_2018,brun_simple_2002}
    \begin{multline}\label{eq:07}
        \dd \omega_s(t) = -\frac{\gamma}{2}\left\{A^\dagger A - \Tr_s[A^\dagger A \omega_s(t)], \omega_s(t) \right\}\dd t \\
        + \left( \frac{A\omega_s(t)A^\dagger}{\Tr_s[A^\dagger A\omega_s(t)]} - \omega_s(t) \right)\dd N(t),
    \end{multline}
    where $\{\bullet, \bullet\}$ denotes the anticommutator, $\dd \omega_s(t) = \omega_s(t+\dd t) - \omega_s(t)$, and $\dd N(t)$ is an increment of an inhomogeneous counting process with the expectation value (mean) given by
    \footnote{Taking the mean of an inhomogeneous Poissonian increment in the context of stochastic master equations is rather delicate, as the average must be first taken in the future and then in the past. This is thoroughly discussed in Ref.~\cite{barchielli1991measurements}.}
    \begin{equation}\label{eq:08}
        \mathbb{E}[\dd N(t)] = \gamma \Tr_s[A^\dagger A\omega_s(t)]\dd t.
    \end{equation}

When there is a click, the increment $\dd N(t)$ is equal to unity, and from Eq.~\eqref{eq:07}, we see a jump toward the target state---set, for example,  $A = \sigma^+$. It follows from Eq.~\eqref{eq:08} that a click readout rarely occurs as its mean value is proportional to $\dd t$. When there is no click, $\dd N(t) = 0$ and Eq.~\eqref{eq:07} describes a ``nudge'' the prior state receives from the backaction of the local measurement. In other words, the system's state evolved continuously irreversibly from $\omega_s(t)$ to $\omega_s(t+\dd t)$. 

The rescaling of $J$ according to 
Eqs.~\eqref{eq:02} and \eqref{eq:02-1} is a \emph{sufficient} condition to implement the WM limit and to obtain a non-trivial evolution in the averaged dynamics \cite{attal2010stochastic}. 

In addition, we note that the WM limit \emph{does not} necessarily imply that the detector-system coupling is in any sense weak (in particular, in the setting where the full Hamiltonian consists of a single term $H_\text{ds}$, there is no other energy scale for $J$ to compare with). In fact, this coupling is related to the \emph{singular coupling limit} one encounters when deriving the LE of a system strongly coupled with a delta-correlated reservoir \cite{breuer2002theory,rivas2012open}. Thus, from the SME \eqref{eq:07}, the ``weak'' in the WM limit stands for the small perturbation that affects the \emph{system} after a projective measurement is performed on the detector and a \textit{no-click} result is observed. When a click is observed, the state receives a ``kick,'' and it jumps to the target state, which is by no means weak
(however, such jumps are rare in the WM limit). 

Turning to the blind measurement step  (i.e., when no selection over the detector readouts is performed), an average over all the trajectories (or readouts) must be taken in the SME, resulting in the WM limit, in that the system dynamics is governed by the LE  \cite{breuer2002theory,wiseman2009quantum}
\begin{subequations}
\label{eq09}
    \begin{align}
        \partial_t\rho_s(t) &= \gamma \mathcal{D}(A)\rho_s(t) \label{eq:09} \\
        &= \gamma A\rho_s(t)A^\dagger - \frac{\gamma}{2}\{A^\dagger A,\rho_s(t) \},
    \end{align}
    \end{subequations}
where 
\begin{equation}\label{eq:09.1}
    \rho_s(t) \coloneqq \mathbb{E}[\omega_s(t)]
\end{equation}
denotes the average over realizations (runs of the measurement protocol). We will refer to the superoperator $\mathcal{D}(A)$ interchangeably as \emph{dissipator} or \emph{simple generator} of dissipative dynamics \cite{baumgartner2008analysis,baumgartner2008analysisII}. By construction of the steering protocol, this dissipator annihilates the target state, i.e., $\mathcal{D}(A)\rho_\oplus = 0$, which, in turn, implies that this state is the \emph{stationary state} solution of the above LE,
\begin{equation}\label{eq:010}
    \rho_\infty = \lim_{t \to \infty }\rho_s(t) = \lim_{t \to \infty }\exp(\mathcal{L}t)\rho_s(0) = \rho_\oplus,
\end{equation}
where $\rho_s(0)$ is the initial state of the system
Note that $\rho_s(t) = \exp(\mathcal{L}t)\rho_s(0)$ is the formal solution of Eq.~\eqref{eq09} with $\mathcal{L} = \gamma \mathcal{D}(A)$ acting as the Lindbladian superoperator.

\section{ Classification and quantification of errors}\label{sec:definition_of_the_errors}

In this section, we introduce some errors that can adversely affect the ideal protocol and steer the system toward an erroneous target state rather than the desired one. For simplicity and for the sake of clarity, we assume throughout this paper that a single error occurs at any given time. As we shall see, each error may induce both dissipative and unitary channels in the LE governing the fully averaged dynamics, yet certain errors only induce one type of channel. Thus, the consideration of several errors acting simultaneously only requires the addition of the corresponding channels (see  Appendix~\ref{sec:multiple_errors_at_the_same_time}). We will implement two distance measures that compare the steered state with an ideal target state. These measures are fidelity and trace distance. Using these measures will enable us to gauge the impact of steering errors on the protocol. We will also employ linear entropy to determine the degree of ``mixedness'' of the steered state when errors occur. The robustness analysis aims to understand how steering errors can affect the reliability and accuracy of the protocol.

\subsection{Quantifying the errors}\label{ssec:quantifying_the_errors}

Let $\mathcal{S}(\mathcal{H}_s)$ be the set of density matrices defined on the 
system Hilbert space $\mathcal{H}_s$, and let $\rho$, $\omega \in \mathcal{S}(\mathcal{H}_s)$. To study how errors alter the protocol, we use three quantifiers that are fidelity 
\begin{equation}\label{eq:2.4}
    F(\rho,\omega) = \left( \Tr \sqrt{\sqrt{\rho}\,\omega \sqrt{\rho}}\right)^2,
    \end{equation}
trace distance
\begin{align}
     D_1(\rho,\omega) &= \frac{1}{2}\Tr[\sqrt{(\rho-\omega)^2}], \label{eq:2.5}
\end{align}
and the linear entropy (also called ``impurity'' as describing the deviation from a pure quantum state)
\begin{equation}\label{eq:2.7}
    L(\rho) \coloneqq 1 - \Tr\rho^2,
\end{equation}
to compare a steered state with the ideal pure target state $\rho_\oplus$ whose linear entropy is unity. Unless explicitly stated, we shall always set the ideal target state as ${\rho}_\oplus=\ket{\uparrow}\bra{\uparrow}$, the erroneous one as $\tilde\rho_\oplus$ and in the same basis, we write the steered state of interest as
\begin{equation}\label{eq:11.30}
    \rho_s(t) = \begin{pmatrix}
    \zeta(t) & \chi(t) \\
    \chi(t)^* & 1- \zeta(t)
    \end{pmatrix}.
\end{equation}
Thus, if we compare $\rho_s(t)$ with $\rho_\oplus$,
the fidelity and the trace distance become
\begin{align}
    F(t) &\coloneqq  F(\rho_\oplus, \rho_s(t)) = \zeta(t), \label{eq:me1} \\
    D_{1}(t) &\coloneqq D_1(\rho_\oplus, \rho_s(t))=  \sqrt{\left[1- \zeta(t) \right]^2 + \abs{\chi(t)}^2}. \label{eq:me2}
\end{align}
We shall denote the above two quantifiers together with the impurity in the stationary regime as $D_{1,\infty}, F_\infty$, and $L_\infty$, correspondingly. 
With the quantifiers at hand, we will determine the robustness of the protocol by performing a series expansion with respect to selected \emph{steering parameters}, e.g., angles defined on the Bloch sphere, channel strengths, probabilities, etc.

\subsection{Types of errors: static and dynamic}\label{ssec:types_of_errors}
In this paper, we discuss two types of errors: \emph{static} and \emph{dynamic}. Static errors refer to the parameters that do not change during the steering protocol. In contrast to static errors, the parameters of the model for dynamic errors  fluctuate during each steering step or from step to step.

We will examine the following two types of static errors:
\begin{enumerate}[(i)]
    \item \emph{Erroneous detector-system coupling parameter}. As we saw in Sec.~\ref{ssec:wml}, there are certain values the product $J\delta t$ can take on that prevent steering from occurring. Therefore, we will define this as our first static error (Sec.~\ref{ssec:erroneous_coupling_parameter}).   
    \item \emph{Erroneously prepared detectors}. The steering protocol requires preparing the detectors in a specific state after each measurement step. Then, it is natural to discuss the case where, because of any external perturbation, we are only capable of preparing the detectors in a state $\tilde \rho_d$ that is not the desired one (i.e., $\tilde \rho_d\neq\rho_d$), and see how the steering protocol with erroneously prepared detectors yields a ``spoiled'' target state (Sec.~\ref{ssec:errors_in_the_detector}).
\end{enumerate}

Dynamic errors are subdivided into two categories: \emph{time-dependent} and \emph{quenched}. If the detector-system Hamiltonian changes in time within one step (between two measurements of the detector), we call the dynamical error time-dependent. If the parameters of the protocol within the step are constant in time but change from step to step, we call the dynamic error quenched. In this work, we will discuss four types of dynamic errors:
 \begin{enumerate}[(i)]
     \item \emph{Fluctuating steering direction}. This error is exclusively quenched and describes the scenario where different steering directions appear at each measurement step. These steering directions can be discrete or continuous (Sec.~\ref{ssec:errors_steering_direction}).
    \item \emph{Fluctuating detector-system interaction strength}.  This error becomes quenched if the coupling constant is drawn from a probability distribution at each steering step. Alternatively, the coupling constant can become a multiplicative white noise (during a single measurement step), making this error time-dependent (Sec.~\ref{ssec:errors_coupling_constant}).
    \item \emph{Environmentally induced perturbation}. A perturbation operator with multiplicative white noise is added to the steering Hamiltonian to represent the interaction of the detector-system with a noisy environment  (Sec.~\ref{ssec:steering_hamiltonian}).
    \item \emph{Erroneous measurement direction}. After each steering step, the basis of the local observable measured on the detector may change, making this dynamic error quenched (Sec.~\ref{ssec:errros_in_the_measurement_direction}).
 \end{enumerate}

\section{Static errors}\label{sec:static_errors}

In this section, we will start by addressing the first static error. We will examine how the system's Bloch vector is affected when $J\delta t$ is a multiple of $\pi$ or an odd multiple of $\pi/2$ and define what we refer to as a \emph{valid} coupling parameter. We will further analyze the effect of the second static error, which occurs when the detector state is prepared incorrectly. In this case, the erroneous preparation of the detector state leads to additional effective Hamiltonian dynamics and modifies the dissipative dynamics observed in the ideal steering protocol. These additional contributions to the steering dynamics will disrupt the target state. 

\subsection{Erroneous coupling parameter}\label{ssec:erroneous_coupling_parameter}

Alluding to the discussion after Eqs.~\eqref{eq:yu1}-\eqref{eq:yu1-1}, let us suppose we want to steer the state $\rho_s(t)$ toward the north pole of the Bloch sphere represented by the state $\rho_\oplus = \ket{\uparrow}\bra{\uparrow}$. 
The detectors are prepared in the state $\rho_d = \ket{\uparrow}\bra{\uparrow}$. With the given target and detector states, the interaction Hamiltonian $H_\text{ds}$ is given by Eq.~\eqref{eq:Hds-proj}, and the associated measurement operators are the same as in Eqs.~\eqref{eq:yu1}-\eqref{eq:yu1-1}.
For a given prior state in the Bloch representation
\begin{equation}
\bm{r}(t) = \textbf{(}x(t),y(t),z(t)\textbf{)} =  \Tr_s[\rho_s(t)\bm{\sigma}],    
\end{equation}
where $\bm{\sigma} = (\sigma^x, \sigma^y,\sigma^z)$ is a vector of Pauli matrices, the updated steered state is then
\begin{equation}\label{eq:n1}
    \bm{r}(t+\delta t) = \begin{pmatrix} \cos(J\delta t) x(t)\\ \cos(J\delta t) y(t)\\ 1 - \cos^2(J\delta t)[1-z(t)] \end{pmatrix}.
\end{equation}
There is no distinction between column and row vectors in our analysis.  
Clearly, from Eq.~\eqref{eq:n1}, if $J \delta t/ \pi  \in \mathbb{Z}$, the $x$ and $y$-components of $\bm{r}(t+\delta t)$ get either reflected or remain invariant so that the state does not get closer to the target state represented by the Bloch vector $$\bm{r}_\oplus = \Tr_s[\rho_\oplus \bm{\sigma}] = (0,0,1).$$ 

If $J \delta t /(\pi/2) \in \mathbb{Z}/ 2\mathbb{Z}$, then the state of the system will jump toward the target state after the measurement, represented by $\bm{r}(t+\delta t) = \bm{r}_\oplus$. However, it is not possible to describe the evolution of $\rho_s(t)$ using Lindbladian dynamics for these specific values of $J\delta t$. Therefore, the coupling strength is erroneous if $J \delta t$ is an integer multiple of $\pi$ or an odd integer multiple of $\pi/2$. If this is not the case, we consider the coupling strength valid. In later sections, we will explore the consequences of this error when it becomes dynamic. 

Henceforth, unless explicitly stated, we shall adopt the WM limit. 

\subsection{Errors in the detector state initialization}\label{ssec:errors_in_the_detector}
After each measurement step, the measurement evolution requires freshly prepared (in a specific state) detectors. Since the detector is also a quantum object, it may interact with the environment so that its state may change from the desired one. The unwanted detector-environment interaction can be cast in the form of Kraus operators $\{K_i\}$, leading to the following detector's averaged density matrix  \cite{jacobs2014quantum}:
\begin{equation}\label{eq:sta1}
    \tilde \rho_d \coloneqq \mathcal{E}[\rho_d] = \sum_i K_i\rho_d K_i^\dagger, 
\end{equation}
where the sum runs over a finite index set. We could, of course, consider the quenched version of this error where (see Appendix~\ref{sec:A1}), at each interaction, the detector state is randomly chosen from an ensemble, e.g., the detector states are $\ket*{\Phi_d^i} = \cos(\theta_i/2)\ket{\Phi_d}+e^{i\varphi_i}\sin(\theta_i/2)\ket*{\Phi_d^\perp}$, and appear with probability $p(i)$ such that $\tilde \rho_d = \sum_i p(i)\ket*{\Phi_d^i}\bra*{\Phi_d^{i}}$, but focus only on the static version of this error.

Without considering any specific set of $K_i$'s, let us assume that their action transforms, in the ONB $\mathcal{B}_d \coloneqq \{\ket{\Phi_d}, \ket*{\Phi_d^\perp} \}$, the ideal detector state to the state 
\begin{equation}\label{eq:13}
    \tilde\rho_d = \begin{pmatrix}
    a & \abs{b}\exp(i\phi) \\ \abs{b}\exp(-i\phi) & 1-a
    \end{pmatrix},
\end{equation}
which is different from the desired detector state $\rho_d = \ket{\Phi_d}\bra{\Phi_d}$. When $a = 1$, which automatically forces $b =0$, we recover $\rho_d$. 

\subsubsection{Dynamics of the steered density matrix} \label{ssec:dynamics_of_rho}

Although we are mainly interested in the stationary state of $\rho_s(t)$, understanding how the steered system reaches the stationary is essential, as reaching it might be an impossible task because of experimental limitations in the measurement rate and the interaction time, among other issues. 

With the detector-system interaction given by Eq.~\eqref{eq:01} and the detector state $\tilde\rho_d$, the blind measurement evolution in the WM limit leads to the Lindbladian dynamics (see Appendix~\ref{sec:A1}):
\begin{multline}\label{eq:14}
\partial_t\rho_s(t)=\bigl[-i\kappa \, \ad(\tilde{h}) + \gamma_+ \mathcal{D}(A) + \gamma_-\mathcal{D}(A^\dagger) \bigr]\rho_s(t),
\end{multline}
where $ \gamma_+ \coloneqq a\gamma$, $\gamma_- \coloneqq (1-a)\gamma$ and
\begin{equation}\label{eq:14.1}
    \ad(\tilde{h})\rho_s(t) \coloneqq [\tilde{h},\rho_s(t)]
\end{equation}
is the adjoint action of 
\begin{equation}\label{eq:15-1}
    \tilde{h}  \coloneqq \exp(i\phi)A + \hc 
\end{equation} 
The coherences of the detector state induce the effective Hamiltonian
\begin{equation}\label{eq:15}
    H = \kappa \tilde{h},  
\end{equation}
whose strength is given by \begin{equation}
\kappa\coloneqq \lim_{\delta t \to 0}J\abs{b}.
\end{equation}
This scaling should be chosen so that $\kappa$ remains constant as we go to the continuum-time limit; otherwise, $H$ would have infinite strength.  

There are three generators of the dynamic semigroup governing the dynamics of $\rho_s(t)$ \cite{baumgartner2008analysis,baumgartner2008analysisII}: the dissipator $\mathcal{D}(A)$ whose stationary state is the ideal target state $\rho_\oplus$ [see Eq.~\eqref{eq09}]; the additional dissipator $\mathcal{D}(A^\dagger)$ annihilating $\rho_\oplus^\perp$ [$\Tr_s(\rho_\oplus \rho_\oplus^\perp) = 0$]; and the unitary generator $\ad(\tilde{h})$. Without loss of generality, let us set $\ket{\Phi_d} = \ket{\uparrow}$ and $\ket*{\Psi_\oplus} = \ket{\uparrow}$. Thus, $A = \sigma^+$ and the Lindbladian in Eq.~\eqref{eq:14} becomes 
\begin{multline}\label{eq:15.1}
    \mathcal{L} = -i\kappa\, \ad\textbf{(}\exp(i\phi)\sigma^+ + \hc  \textbf{)} + \gamma_+ \mathcal{D}(\sigma^+)\\  + \gamma_- \mathcal{D}(\sigma^-).
\end{multline}
By adopting the Bloch representation $$\rho_s(t) = \frac{1}{2}[I_s + \bm{r}(t)\cdot \bm{\sigma}],$$
the effective dimensionless Hamiltonian in the Lindbladian \eqref{eq:15.1} can be rewritten as 
$\tilde{h} = \bm{n}\cdot \bm{\sigma}$, where $\bm{n} = (\cos\phi, -\sin\phi,0)$.
Hence, the unitary channel in $\mathcal{L}$ induces Rabi oscillations around the unit vector $\bm{n}$ with a Rabi frequency of $\kappa/2$. 

From the Bloch representation of $\rho_s(t)$, we get the set of coupled ordinary differential equations for the Bloch components: 
\begin{align}
    \dot{x}(t) &= -\frac{\gamma x(t)}{2} - 2\kappa z(t)\sin\phi, \label{eq:15.12} \\
    \dot{y}(t) &= -\frac{\gamma y(t)}{2} - 2\kappa z(t)\cos\phi, \label{eq:15.a}\\
    \dot{z}(t) &= \gamma(2a-1) - \gamma z(t) + 2\kappa y(t)\cos\phi \notag \\ 
&+2\kappa x(t)\sin\phi.\label{eq:15.b}
\end{align}
For $\kappa \neq \gamma/8$, the solutions of Eqs.~\eqref{eq:15.12}-\eqref{eq:15.b} read
\begin{align}
    x(t) &= \frac{1}{\kappa}\left[ \cos(\phi) g(t) + \sin(\phi) f(t)\right], \label{eq:15.120} \\
    y(t) &= \frac{1}{\kappa}\left[ \cos(\phi)f(t) - \sin(\phi)g(t)\right], \label{eq:15.121}\\
    z(t) &= 2\left[C_1\exp(\Omega_+ t) + C_2\exp(\Omega_- t) + \lambda \right] - 1, \label{eq:15.122}
\end{align}
where
\begin{align}
    f(t) &= C_1\,(\Omega_+ + \gamma)\exp(\Omega_+ t)\notag \\  &+ C_2\,(\Omega_- + \gamma)\exp(\Omega_- t) + \gamma(\lambda - a), \label{eq:15.123} \\
    g(t) &= C_3\exp(-\frac{\gamma}{2}t),
     \label{eq:15.124} 
     \end{align}
the integration constants $C_{i}$ ($i=1,2,3$) depend on the initial state and the protocol parameters, and     \begin{align}
    \lambda &\coloneqq \frac{\gamma^2a + 4\kappa^2}{\gamma^2 + 8\kappa^2}, \label{eq:15.125} \\
    \Omega_\pm &\coloneqq \pm \sqrt{\left( \frac{\gamma}{4} \right)^2 - (2\kappa)^2} - \frac{3\gamma}{4}.
    \label{eq:15.126}
\end{align}

As can be seen from Eqs.~\eqref{eq:15.120}-\eqref{eq:15.126}, for $\kappa \neq \gamma/8$, the dynamics of $\bm{r}(t)$ is mainly controlled by $\Omega_\pm$, and it may exhibit one of the following two regimes: 
\emph{underdamped} when $\kappa > \gamma/8$, 
and 
\emph{overdamped} when $\kappa < \gamma/8$. 
Clearly, the solution corresponding to the ideal protocol (i.e., when $\tilde\rho_d = \rho_d$) belongs to the overdamped regime in the limit $\kappa \rightarrow 0$, yielding 
\begin{align}
    x(t) &= x(0)\exp(-\frac{\gamma}{2}t), \label{eq:15.127-1}\\ 
    y(t) &= y(0)\exp(-\frac{\gamma}{2}t), \label{eq:15.127}\\
    z(t) &= 1 - \left[1-z(0) \right]\exp(-\gamma t), \label{eq:15.129}
\end{align}
where the approach to the target state is exponential in time with the rate $\gamma$ given by Eq.~\eqref{eq:02-1}.
The \emph{critically damped} regime occurs when $\kappa = \gamma/8$ and the solution is obtained directly from Eqs. \eqref{eq:15.12}-\eqref{eq:15.b}:
\begin{align}
    x(t) &= \frac{4(1-2a)}{9}\sin \phi+ \exp(-\frac{3\gamma }{4}t)\, C_x(t), \label{eq:cri1}\\
    y(t) &= \frac{4(1-2a)}{9}\cos \phi+ \exp(-\frac{3\gamma }{4}t)\,  C_y(t), \label{eq:cri2} \\
    z(t) &= \frac{8(1-2a)}{9}+ \exp(-\frac{3\gamma }{4}t )\, C_z(t), \label{eq:cri3}
\end{align}
with 
\begin{widetext}
\begin{align}
    C_x(t) &=  \frac{1}{4} (\gamma  t+4) \left[x(0)\sin\phi
    +y(0) \cos \phi \right]\sin\phi+
    \left[\frac{4}{9} (2a-1)+ \gamma  t \left(\frac{1}{3} (2a-1)-\frac{1}{4} z(0)\right) \right]\sin\phi \notag  \\ & +  \exp(\frac{\gamma  }{4}t)\, 
    \left[x(0)\cos\phi -y(0) \sin \phi \right]\cos\phi, \\
    C_y(t) &=   \frac{1}{4}  (\gamma  t+4) [x(0) \sin \phi+y(0) \cos \phi ]\cos \phi +\left[\frac{4}{9} (2a-1)+ \gamma  t \left(\frac{1}{3} (2a-1)-\frac{1}{4} z(0)\right) \right]\cos \phi  \notag \\ &
    + \exp(\frac{\gamma  }{4}t)  [y(0) \sin \phi -x(0) \cos \phi ]\sin \phi, \\
    C_z(t) &=  \frac{1}{4} \gamma  t [x(0) \sin \phi + y(0) \cos \phi ] +\gamma  t \left[\frac{1}{3} (2a-1)-\frac{1}{4} z(0)\right] + z(0) -\frac{8}{9}(2a-1).
\end{align}
\end{widetext}

\begin{figure*}
\centering
\begin{tabular}{cccc}
\includegraphics[width=0.3\textwidth]{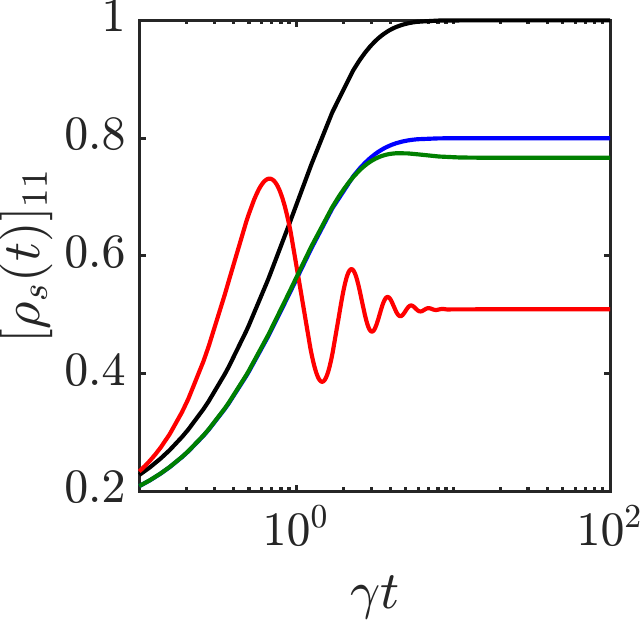} &
\includegraphics[width=0.3\textwidth]{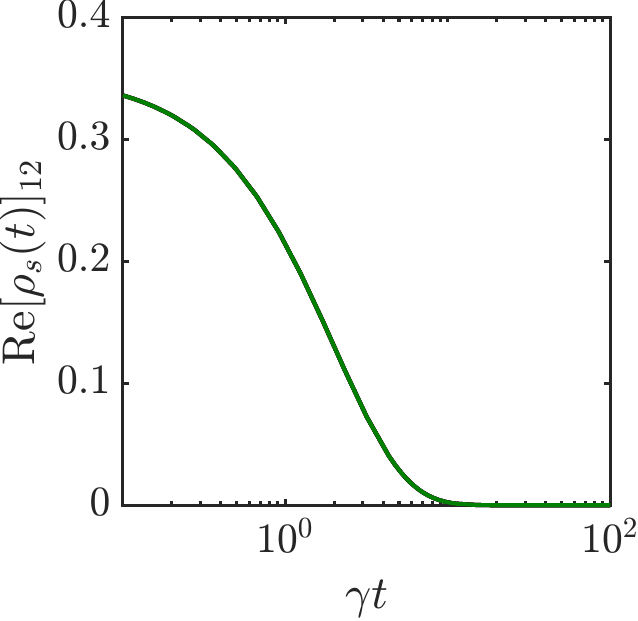} &
\includegraphics[width=0.3\textwidth]{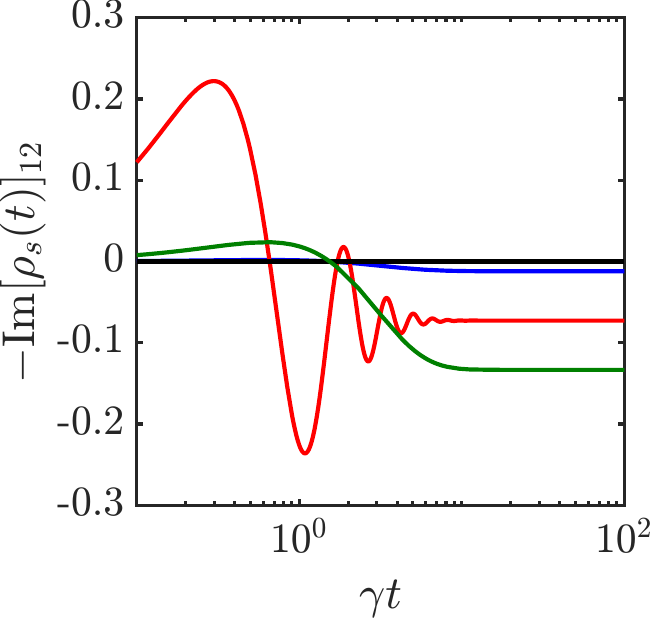} \\
(a)  & (b) & (c) \\[6pt]
\end{tabular}
\begin{tabular}{cccc}
\includegraphics[width=0.3\textwidth]{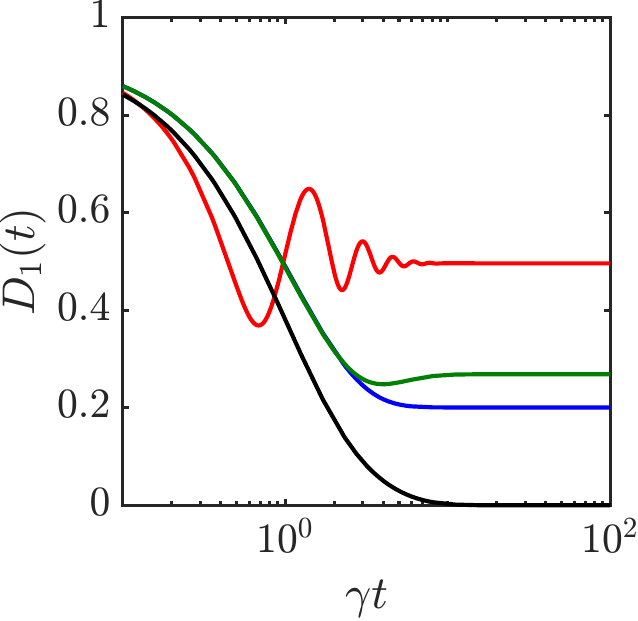} &
\includegraphics[width=0.3\textwidth]{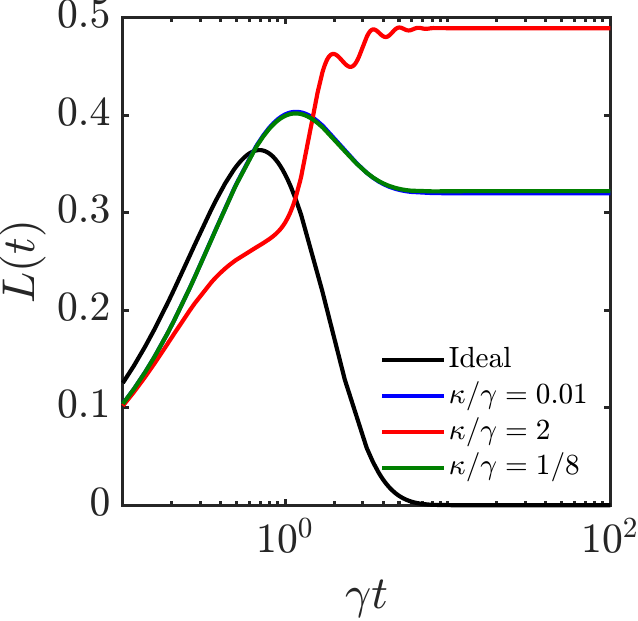} \\
(d)  & (e)  \\[6pt]
\end{tabular}
    \caption{Steering of a single qubit via continuous time evolution  toward the ideal target state $\ket{\Psi_\oplus} = \ket{\uparrow}$ with the ideal parameters (black) [Eqs.~\eqref{eq:15.127-1}-\eqref{eq:15.129}] and with erroneously prepared detectors [Eq.~\eqref{eq:13}] in three different dynamic regimes: underdamped, $\kappa/\gamma = 2$ (red); overdamped, $\kappa/\gamma =  0.01 $ (blue) [Eqs.~\eqref{eq:15.120}-\eqref{eq:15.122} for $\kappa/\gamma \neq 1/8$]; and critically damped [Eq.~\eqref{eq:cri1}-\eqref{eq:cri3}], $\kappa/\gamma =  1/8$ (green). Entries of $\rho_s(t)$ in the three dynamic regimes: (a) $[\rho_s(t)]_{11} = [1 + z(t)]/2$.  Panels (b) and (c) respectively show the real and the negative imaginary parts of $[\rho_s(t)]_{12} = [x(t) -iy(t)]/2$. In (b), all the curves coincide. (d) Trace distance, given by Eq.~\eqref{eq:me2}, between the steered state (erroneous or not) and the ideal target state. (e) Linear entropy [Eq.~\eqref{eq:2.7}]. In all cases, $a = 0.8$ and $\phi = 0$ in Eq.~\eqref{eq:13}, and the initial Bloch vector is $\bm{r}(0) = (1,1,-1)/\sqrt{3}$.}
\label{fig:00}
\end{figure*}

Results for the three dynamical regimes exposed above are compared with those of the ideal steering in Fig.~\ref{fig:00} for particular values of $\kappa/\gamma$. By recalling that $\rho_s(t) = [I_s + \bm{r}(t)\cdot \bm{\sigma}]/2$, in the ideal steering protocol [Eqs.~\eqref{eq:15.127-1}-\eqref{eq:15.129}] the real and imaginary parts of the off-diagonal elements of $\rho_s(t)$ rapidly go to zero, and $$[\rho_s(t)]_{11} = \frac{1}{2}[1 + z(t)]$$ goes to unity twice as fast. This situation no longer holds for erroneously prepared detectors, as shown in panels (a)-(c) of Fig.~\ref{fig:00}.
Furthermore, in Fig.~\ref{fig:00}, panels (d) and (e) show that when $\kappa/\gamma=2$, which corresponds to the underdamped regime, the impurity and trace distance are the lowest until they intersect with the ideal steering curves. In the stationary state, the steered state in the underdamped regime becomes the least pure and is the furthest from the ideal target state compared with the ideal steering and both the overdamped ($\kappa/\gamma = 0.01$) and critically damped ($\kappa/\gamma = 1/8$) regimes. The latter two regimes have a stationary state impurity close to $0.3$ from above. However, the stationary state value of the trace distance corresponding to the overdamped regime is smaller than for the critically damped regime.
It is worth noting that if we adopted an active-decision steering protocol such as the one in Ref.~\cite{herasimenko} while considering erroneously prepared detectors leading to the underdamped regime dynamics of $\rho_s(t)$, it would be reasonable to stop the steering protocol at the precise moment the trace distance reaches the undershoot since this is the closest point to the ideal target state. This would not apply in the other two regimes because it takes more time for the trace distance to reach its minimum value. On the other hand, it is possible to combine the current protocol in the underdamped regime (which would no longer be considered erroneous) with a very strong unitary channel to quickly reach the undershoot as close as possible to the target state and then implement purely dissipative steering. Combining these two steering protocols and potentially other optimization schemes may improve the overall performance of the steering protocol. An example of an optimized protocol of this kind is explored in Ref.~\cite{kumar2022optimized}.

\subsubsection{Stationary state analysis}\label{ssec:stationary_state_analysis}

Next, we will conduct a stationary state analysis of Eq.~\eqref{eq:14} to understand how the steering parameters $\mathcal{P} = \{\gamma,\kappa, a,\phi \}$ influence the target state $\tilde\rho_\oplus$. By using the same orthonormal basis that led to the Lindbladian~\eqref{eq:15.1} we find its stationary state
\begin{equation}\label{eq:16}
     \tilde \rho_\oplus=  \begin{pmatrix}
\dfrac{1}{2} + \dfrac{\gamma^2(2a-1)}{2(\gamma^2 + 8\kappa^2)}  & \dfrac{i2e^{i\phi}\gamma \kappa (1-2a)}{\gamma^2 + 8\kappa^2} \\
-\dfrac{i2e^{-i\phi}\gamma \kappa (1-2a)}{\gamma^2 + 8\kappa^2} & 
\dfrac{1}{2} - \dfrac{\gamma^2(2a-1)}{2(\gamma^2 + 8\kappa^2)}
\end{pmatrix},
\end{equation}
where a clear dependence on the channel strengths can be seen. By turning to the Bloch representation of $\tilde\rho_\oplus$, we can conveniently observe the allowed stationary regions in the Bloch ball with the aid of the Bloch vector
\begin{equation}\label{eq:17}
      \bm{r}_{\infty} = \begin{pmatrix} x_\infty \\ y_\infty \\ z_\infty \end{pmatrix} = \frac{(2a-1)\gamma}{\gamma^2 + 8\kappa^2} \begin{pmatrix} 
    -4 \kappa \sin\phi\\
    -4 \kappa \cos\phi\\
    \gamma
    \end{pmatrix}.
\end{equation}

The regions are the following ones: the origin of the Bloch sphere $\norm{\bm{r}_\infty} = 0$ is attained when $a = 1/2$, which in turn implies that the two dissipators in $\mathcal{L}$ [Eq.~\eqref{eq:15.1}] have the same decay rate. This point, which is the maximally mixed state, is accessible regardless of the presence of the unitary generator in $\mathcal{L}$. In this specific case, $F_\infty = D_{1,\infty} = 1/2$ [Eq. \eqref{eq:me1}].

The second stationary region is the $z$-axis of the Bloch ball. Regardless of the value of $a$, having $b = 0$ forces  $\kappa = 0$, so that there is no unitary generator in $\mathcal{L}$ and the Bloch vector is $\bm{r}_\infty = (0,0,2a-1)$. This result is a consequence of the competition between the dissipators $\mathcal{D}(\sigma^+)$ and $\mathcal{D}(\sigma^-)$, where the former dissipator steers toward ${\rho}_\oplus$ and the latter toward ${\rho}^\perp_\oplus$. The trace distance and fidelity in this case are $D_{1,\infty} = \abs{1-a}$ and $F_\infty = a$, respectively.

The third stationary region is a \emph{stationary ellipsoid}. These are notable features in dissipative quantum dynamics and optimal control of two-level systems \cite{kumar2022optimized,lapert2013understanding,giorgi2020microscopic,mukherjee2013speeding,sauer2013optimal}. Let $\kappa \in [0,\infty)$, $a \in (0,1/2)\cup(1/2,1)$, $\phi \in [0,2\pi)$, and $\gamma \in (0, \infty)$. For $\gamma$ and $a$ fixed, one can demonstrate that the vector components in Eq.~\eqref{eq:17} satisfy the relation \footnote{The same  surface is obtained if $\kappa$ is fixed and if  $\gamma$ varies in  $\in (0,\infty)$.}
    \begin{equation}\label{eq:18}
        \frac{x_\infty ^2+ y_\infty^2}{2(a-1/2)^2} + \frac{\left[z_\infty -(a-1/2)\right]^2}{(a-1/2)^2} = 1,
    \end{equation}
which is an oblate ellipsoid. However, these ellipsoids are punctured: for any ellipsoid, the endpoint of its minor axis that would coincide with the origin of the Bloch sphere is removed from Eq.~\eqref{eq:18}. Also, the other end of the minor axis does not intersect the Bloch sphere. This is shown in Appendix \ref{sec:stationary_ellipsoid}. We denote these punctured ellipsoids by $\mathcal{C}$, and provide two examples in Fig. \ref{fig:01} for $\gamma = 5$, $a=0.8$, and $a =0.2$. 

The faulty protocol discussed here exhibits multiple stationary states for a single target state. This means that a specific target state can be reached using different $\tilde\rho_d$. In other words, multiple intersecting ellipsoids exist for different values of $\tilde\rho_d$. Additionally, in contrast to the steering ellipsoids shown in Ref.~\cite{kumar2022optimized}, any ellipsoid $\mathcal{C}$ corresponding to Eq.~\eqref{eq:18} \emph{does} depend on the steering parameters, it is punctured, and does not intersect the Bloch sphere.
 
In summary, the allowed regions for the protocol considered in this section are either the points representing the pure states ${\rho}_\oplus$  or ${\rho}_\oplus^\perp$; a line joining  $\rho_\oplus$ and ${\rho}_\oplus^\perp$, i.e., $\tilde\rho_\oplus = a{\rho}_\oplus + (1-a){\rho}_\oplus^\perp$; or a punctured ellipsoid with the minor axis parallel to the line joining $\rho_\oplus$ and ${\rho}_\oplus^\perp$. 

\begin{figure}[hb]
  \begin{subfigure}{0.22\textwidth}
    \includegraphics[width=\linewidth]{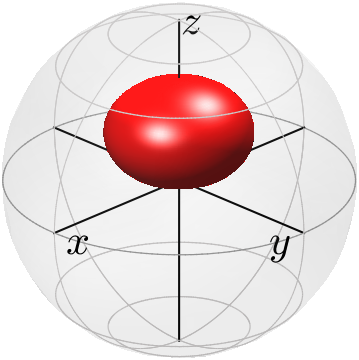}
    \caption{}  \label{fig01:a}
  \end{subfigure}%
 \hspace{1em}  
  \begin{subfigure}{0.22\textwidth}
    \includegraphics[width=\linewidth]{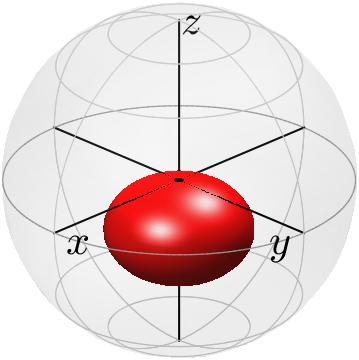}
    \caption{ } \label{fig01:b}
  \end{subfigure}%
  \caption{Stationary ellipsoid representing Eq.~(\ref{eq:18}) with steering parameters (a) $\gamma = 5$, $\kappa \in (0, \infty)$, $a = 0.8$, and (b) $\gamma = 5$, $\kappa \in (0, \infty)$, $a = 0.2$. Each point on the ellipsoids represents a stationary state $\tilde\rho_\oplus$ [Eq.~(\ref{eq:16})] for a specific choice of $\tilde\rho_d$. In panel (b), it is evident that the coordinate origin is not contained in the ellipsoid.} 
  \label{fig:01}
\end{figure}

\subsubsection{Small-error approximation}
Despite having ${\tilde\rho}_\oplus \neq \rho_\oplus$ unless ${\tilde\rho}_d = \rho_d$, a close-to-ideal experimental realization of the protocol would require the errors to be considerably small. Taking this into consideration, in what follows, we will provide a series expansion of the distance measures in the stationary state regime as a function of the steering parameters. This will determine the robustness of the protocol to this error.

The stationary state fidelity [Eq.~\eqref{eq:me1}] and trace distance [Eq.~\eqref{eq:me2}] between $\rho_\oplus$ and $\tilde\rho_\oplus$ together with the linear entropy [Eq.~\eqref{eq:2.7}] of $\tilde\rho_\oplus$ respectively are
\begin{align}
    F_\infty &= \frac{1}{2} + \frac{\gamma^2(2a-1)}{2(\gamma^2 + 8\kappa^2)}, \label{eq:20} \\
    D_{1,\infty} &= \frac{\sqrt{ [(1-a) \gamma^2 + 4\kappa^2]^2 + 4\gamma^2\kappa^2(2a-1)^2  }}{\gamma^2 + 8\kappa^2},\\
    L_\infty  &= 1\notag \\
  &-\frac{ \gamma^2\left(\gamma^2 + 16\kappa^2 \right)\left[ 1 + 2(a - 1)a  \right] + 32\kappa^4 }{\left( \gamma^2 + 8\kappa^2\right)^2}.
\end{align}

Since an error is considered small if $\kappa\rightarrow 0$ and $a \rightarrow 1$, we expand the above expressions first in $\kappa$ and then in $a$ approaching unity from the left. We thus obtain
\begin{align}
F_\infty &= a - \frac{4}{\gamma^2}(2a-1)\kappa^2+ \mathcal{O}(\kappa^4), \label{eq:22.1} \\
D_{1,\infty} &= 1- a + \frac{2}{\gamma^2}\frac{2a-1}{1-a}\kappa^2  +\mathcal{O}(\kappa^4)\label{eq:22.2}\\
L_\infty &= 2a(1-a) +\frac{32}{\gamma^4}(1-2a)^2\kappa^4  +\mathcal{O}(\kappa^6). \label{eq:22.4}
\end{align}
In light of these series expansions, we conclude that the population $(\tilde\rho_d)_{11}$ in Eq.~\eqref{eq:13} dominates the steering at first order in $a$ without any involvement of the ideal decay rate $\gamma$ and the strength of the unitary channel $\kappa$. This situation no longer holds when higher-order terms are considered.
Naturally, $D_{1,\infty}$ tends to zero as $\kappa$ goes to zero faster 
than $a$ to unity, since no coherence must exist in Eq.~\eqref{eq:13} when $\tilde\rho_d = \rho_d$.

\section{Dynamic errors}\label{sec:dynamic_errors}
We now proceed to discuss how dynamic errors affect the steering protocol. In particular, we will investigate the following dynamic errors: (i) fluctuating steering directions, (ii) imperfect control over detector-system interaction coupling, (iii) environmentally induced perturbation in the desired steering Hamiltonian, and (iv) fluctuating directions at which the detectors are projected (i.e., projectively measured).

Error (i) is exclusively quenched and originates from the fact that different steering directions can appear at different steps of the steering dynamics due to an erroneous detector-system interaction. The directions might be continuously or discretely distributed. For the latter case, we show three novel stochastic master equations (SMEs), Eqs.~(\ref{eq:sde1},\ref{eq:ssde5},\ref{eq:sde11}), that exhibit different averaging hierarchies. One of these SMEs [see Eq.~\eqref{eq:sde1}] describes the full stochasticity of the problem, i.e., the stochasticity due to fluctuating steering directions, as well as due to the quantum-mechanical randomness of the outcomes of the measurements for a particular observable. However, we can opt to average out either the measurement stochasticity [see Eq.~\eqref{eq:sde11}] or the steering directions stochasticity [see Eq.~\eqref{eq:ssde5}], which leads to the remaining SMEs. Once both stochasticities are averaged out, all three SMEs lead to the same Lindblad master equation [see Eq.~\eqref{eq:sde4.1.1}], as the averages over the two stochastic processes commute. As we shall see, the unraveled LE is purely dissipative and has as many dissipation channels as steering directions.  Turning to a stationary state analysis, we will compare the ideal target state with the stationary state obtained when two steering directions (or states) appear. These states are symmetrically located on the Bloch sphere relative to the ideal target state. 

Error (ii) becomes quenched if the detector-system coupling strength is drawn from a probability distribution at each steering step. Alternatively, the coupling constant can experience a random time-dependent external perturbation during the measurement steps, e.g., a multiplicative white noise, making this error time-dependent. As we shall see, the LE describing the system dynamics with this particular error will have the same dissipative channel as the one with the ideal steering protocol [cf. Eqs~\eqref{eq:cc1},~\eqref{eq:cc1.2} and ~\eqref{eq:cc6}]. However, it will have an effective channel strength influenced by the noise. 

Considering error (iii), we will examine the LE describing the system dynamics [Eq.~\eqref{eq:c13}] using two approaches: in the first approach, we will obtain a novel SME [Eq.~\eqref{eq:esto1}] that includes both measurement and noise stochasticities. In the second approach, we average the reduced dynamics of the system directly. The resulting LE contains the ideal dissipator (i.e., the one that steers the system toward the desired target state) and other dissipative channels caused by the error.

Since the blind measurement is independent of the detector's basis, the steering protocol remains fully robust against error (iv). It will be shown, however, that by varying the basis for measuring the detectors, it is possible to obtain an SME [Eq.~\eqref{eq:pb6}] including contributions from both jump-type and diffusive-type measurements \cite{ barchielli1991measurements}.

\subsection{Error in the steering direction}\label{ssec:errors_steering_direction}

Suppose we want to steer a qubit toward the ideal target state ${\rho}_\oplus$, with detectors prepared in the state $\rho_d = \ket{\Phi_d}\bra*{\Phi_d}$. At each steering step, the local observable $S_d$ [see Eq.~\eqref{eq:Sd}]
is measured, and, as the steering direction fluctuates, the system gets steered toward an erroneous target state $\omega_i$ (such that $\Tr_s\omega_i^2 = 1$) with probability $p(i)$. Let us denote the set of target states and their associated probabilities as 
\begin{equation}\label{eq:23}
    \mathcal{R} \coloneqq \{ \omega_i;p(i)\}_{i \in \mathcal{I}}
\end{equation}
where $\mathcal{I} = \{1,\ldots, n \}$ is an index set. Note that $\rho_\oplus \in \mathcal{R}$ is not a requirement. Furthermore, the set of discrete probabilities $\{p(i)\}_{i\in \mathcal{I}}$ may become a continuous probability distribution.

The dimensionless detector-system interaction Hamiltonian that steers the system toward $\omega_i$ can be written as [cf. Eq.~\eqref{eq:01}]
\begin{equation}\label{eq:pe2}
h_0^{(i)} = \ket*{\Phi_d^\perp}\bra{\Phi_d} \otimes A_i + \hc, 
\end{equation}
where, for every $i \in \mathcal{I}$,
\begin{equation}\label{eq:pe2.1}
A_i = A(\theta_i, \varphi_i) \coloneqq  R(\theta_i,\varphi_i) A R^\dagger(\theta_i,\varphi_i) 
\end{equation}
is the operator $A$ rotated toward the $i$-th direction under the action of the rotation operator
\begin{equation}\label{eq:7.3}
    R(\theta_i,\varphi_i) \coloneqq \exp(-i\frac{\varphi_i}{2}\sigma^z) \exp(-i\frac{\theta_i}{2}\sigma^y),
\end{equation}
and so $A_i \omega_i = 0$. The unitary operator evolving the detector-system state corresponding to Eq.~\eqref{eq:pe2} is denoted as 
\begin{equation}
    U_i(\delta t) = \exp(-i\sqrt{\gamma \delta t}h_0^{(i)}),
\end{equation}
where the scaling relation $J = \sqrt{\gamma/\delta t}$, which leads to the WM limit, has been set. 
By taking advantage of the rotation angles in Eq.~\eqref{eq:pe2.1}, we denote the states $\omega_i$ by their associated angles $(\theta_i,\varphi_i)$. Hence, Eq.~\eqref{eq:23} can be conveniently written as $\mathcal{R} = \{(\theta_i,\varphi_i); p(i) \}_{i\in \mathcal{I}}$.  

We proceed to describe the dynamics of the steered system via SMEs portraying different averaging hierarchies for the derivation of the LE. Specifically, we will analyze three cases: (a) direct (simultaneous) averaging of the SME over both the \emph{quantum-mechanical} stochasticity introduced by random measurement readouts and the \emph{classical} stochasticity introduced by choice of the measurement direction at each step; (b) averaging first over the steering directions, keeping a particular sequence of readouts, followed by averaging over readouts at the later stage; (c) averaging first over detector readouts for a given sequence of steering directions, followed by averaging over these directions. A summary of the upcoming SMEs and their averages is illustrated in Fig.~\ref{fig:00001}.

An important question regarding these averaging hierarchies concerns the commutativity of the averaging procedures (b) and (c). Indeed, in the fully stochastic consideration based on quantum trajectories, the detector readouts (click or no-click) are conditioned to the given steering direction. However, within the hierarchy (b), the first averaging is performed over all possible steering directions for a fixed (say, click) measurement outcome, whose probability may strongly depend on the steering direction. Note that one might think that such averaging violates the conditional relation between detector outcomes and directions. Nevertheless, as shown below, the resulting LEs for all three averaging hierarchies coincide. 

At the same time, the SMEs for the partially averaged density matrices are different in the three cases of averaging hierarchies. Each of them bears important information on the dynamics of the system, in particular, on the statistics of quantum trajectories, which can be experimentally probed in a finite number of protocol runs.
For example, hierarchy (b) is relevant to the situation when the fluctuations of the steering direction are uncontrolled, while each protocol run yields a definite sequence of readouts. Hierarchy (c) can be experimentally realized by performing multiple runs for a fixed sequence of steering directions intentionally chosen to test the robustness of the protocol. The information extracted from the corresponding SMEs can also be employed for active-decision strategies \cite{herasimenko}, particularly in a termination policy determining the optimum number of steering steps.

\subsubsection{Stochastic steering directions and detector outcomes}

Following the quantum trajectory formalism described in Sec.~\ref{sec:ideal_protocol_and_setup} [cf. Eq.~\eqref{eq:07}], an SME that simultaneously describes the erroneous steering and measurement stochasticities takes the form (see Fig.~\ref{fig:00001})
\begin{widetext}
\begin{equation}\label{eq:sde1}
    \dd \omega_s(t) = \sum_{i \in \mathcal{I}}\chi_i(t)\left[ \gamma \mathcal{D}(A_i)\omega_s(t)\dd t  + \left( \frac{A_i\omega_s(t)A_i^\dagger}{\expval*{A_i^\dagger A_i}_t} - \omega_s(t) \right)\left(\dd N_i(t)- \gamma \expval*{A^\dagger_i A_i}_t\dd t \right)\right].
\end{equation}
\end{widetext}
The derivation of the above SME is shown in Appendix~\ref{sec:appendix_jump_several_directions_static}.
This equation can be understood as follows. Between $t$ and $t + \dd t$, only one of the stochastic variables (or indicators), say, $\chi_i(t)$, equals unity, and the rest are zero. This occurs randomly with a probability
\begin{equation}\label{eq:sde2}
    \mathbb{E}[\chi_i(t)] = p(i) \quad \forall i \in \mathcal{I},
\end{equation}
where $\mathbb{E}$ denotes the trajectory average. At the same time, the Poissonian increment $\dd N_i(t)$ describes a jump [$\dd N_i(t) = 1$] or the lack of it [$\dd N_i(t) = 0$] corresponding to the $i$-th steering direction for which $\chi_i(t)=1$. The strength of each counting process is given by the mean 
\begin{equation}\label{eq:sde3}
    \mathbb{E}[\dd N_i(t)] = \gamma \expval*{A_i^\dagger A_i}_t\dd t = \gamma \Tr[A_i^\dagger A_i \omega_s(t)]\dd t,
\end{equation}
given that the trajectory $\omega_s(t)$ has been realized. To be more precise, in the context of the dynamics described by Eq.~\eqref{eq:sde1}, we must consider the product of each indicator with its correspondent counting process, that is,
\begin{equation}\label{eq:sde3.1}
    \mathbb{E}[\chi_i(t)\dd N_i(t)] = p(i)\gamma \expval*{A_i^\dagger A_i}_t\dd t.
\end{equation}

We note that a direct consequence of Eq.~\eqref{eq:sde1} is that if $\omega_s^2(t) =\omega_s(t)$, then $\omega_s^2(t + \dd t) = \omega_s(t+ \dd t)$, i.e., the evolution of a pure state remains pure. Therefore, Eq.~\eqref{eq:sde1} with
$$\omega_s(t) = \ket{\psi_s(t)}\bra{\psi_s(t)}$$
is equivalent to the stochastic Scrhödinger equation
\begin{multline}\label{eq:sde4.1}
    \dd \ket{\psi_s(t)} = -\frac{1}{2}\sum_{i \in \mathcal{I}}\left(\gamma A_i^\dagger A_i - \gamma \expval*{A_i^\dagger A_i}_t \right)\ket{\psi_s(t)}\chi_i(t)\dd t \\
    + \sum_{i \in \mathcal{I}}\left(\frac{A_i}{\sqrt{\expval*{A_i^\dagger A_i}_t} }- I_s \right)\ket{\psi_s(t)} \chi_i(t)\dd N_i(t),
\end{multline}
where
\begin{align}
\expval*{A_i^\dagger A_i}_t &= \bra{\psi_s(t)}A_i^\dagger A_i\ket{\psi_s(t)}.
\end{align}
 Interestingly, if we set $\chi_i(t) = 1$ for all $i \in \mathcal{I}$ in Eq.~\eqref{eq:sde4.1}, we would obtain a standard stochastic Scrhödinger equation of the jump type describing the continuous monitoring of a quantum system by $n$ detectors \cite{gross_qubit_2018,brun_simple_2002,barchielli1991measurements,breuer2002theory,wiseman2009quantum} instead of a chain of monitoring detectors. 

In Fig.~\ref{fig40:a}, we show a solution (i.e., a trajectory) of Eq.~\eqref{eq:sde1} in the Bloch representation with the steering directions and their probabilities [cf. Eq.~\eqref{eq:23}] given by
\begin{equation}\label{eq:erre}
    \mathcal{R} = \{(\pi/3,0; 0.5),(\pi/3,\pi; 0.5)\},
\end{equation}
a decay $\gamma = 0.1$, $\delta t = 0.1$, and with an initial pure state $\bm{r}(0) = (1,0,-1)/\sqrt{2}$. It can be observed in Fig.~\ref{fig40:a} that the evolved state is always pure and, at random, either evolves continuously or jumps to one of the states in $\mathcal{R}$. A transverse cross-section of the Bloch sphere from Fig.~\ref{fig40:a} is shown in Fig.~\ref{fig40:f}. In this figure, the straight lines represent the jumps. A lack of a stationary state can be observed due to the never-ending jumps and continuous evolution of the statistical operator evolves continuously.  

Taking the full average of Eq.~\eqref{eq:sde1} removes all the stochastic terms, resulting in the LE (see Fig.~\ref{fig:00001})
\begin{equation}\label{eq:sde4.1.1}
\partial_t \rho_s(t) = \sum_{i \in \mathcal{I}} \gamma p(i)\mathcal{D}(A_i)\rho_s(t).
\end{equation}
This master equation is purely dissipative and has as many dissipators as steering directions. Continuing with the above example, averaging over $10^3$ trajectories of the form shown in Fig.~\ref{fig40:a} gives an approximate solution to the corresponding LE. The corresponding averaged trajectory is shown in Figs.~\ref{fig40:d}-\ref{fig40:i} in red using the Bloch representation, and it is compared with the exact solution, in black, of the corresponding Lindblad equation
\begin{align}
    [\rho_s(t)]_{11}&= \frac{9}{10} - \frac{8+5\sqrt{2}}{20}\exp(-\frac{5\gamma}{8} t), \label{eq:lind01}\\
    [\rho_s(t)]_{12}&= \frac{1}{2\sqrt{2}}\exp(-\frac{7\gamma}{8})t. \label{eq:lind02}
\end{align}

As shown in Figs.~\ref{fig40:d} and \ref{fig40:i}, the resulting stationary state is mixed because of the competition of dissipative channels. To complete the picture, the solution of the ideal LE---steering toward the north pole of the Bloch sphere---as given in Eqs.~\eqref{eq:15.127-1}-\eqref{eq:15.129}, is depicted in Figs.~\ref{fig40:e} and \ref{fig40:j}, where again $\bm{r}(0) = (1,0,-1)/\sqrt{2}$ and $\gamma = 0.1$. Figure \ref{fig:301} shows the time dependence of the Bloch components depicted in Fig.~\ref{fig:40}, as well as the quantifiers Eqs.~\eqref{eq:2.7},~\eqref{eq:me1}, and~\eqref{eq:me2}.

\subsubsection{Averaged steering directions}

 We can now examine the situation in which the steering directional stochasticity has been averaged out, leaving only the measurement stochasticity in the system dynamics. In that case, 
 the relation between a  (partially averaged) density matrix $\pi_s(t)$  with the state $\omega_s(t)$---a solution of Eq.~\eqref{eq:sde1}---is given by
\begin{equation}\label{eq:sde8}
    \pi_s(t) \coloneqq \mathbb{E}_i[\omega_s(t)],
\end{equation}
where $\mathbb{E}_i$ denotes the
classical average over the steering directions.
The following SME gives the evolution of the system state:
\begin{widetext}
\begin{equation}\label{eq:ssde5}
\dd \pi_s(t) = \sum_{i \in \mathcal{I}}\gamma p(i) \mathcal{D}(A_i)\pi_s(t)\dd t + \left( \frac{\sum_{i \in \mathcal{I}}p(i)A_i\pi_s(t) A_i^\dagger}{\expval{\sum_{j \in \mathcal{I}}p(j)A_j^\dagger A_j}_t} - \pi_s(t)  \right)\left( \dd N(t) - \gamma \expval{\sum_{i \in \mathcal{I}}p(i)A^\dagger_iA_i}\dd t \right).
\end{equation}
\end{widetext}
The derivation of this SME is presented in
Appendix~\ref{sec:appendix_jump_several_directions_static}, and interestingly, it can also be obtained from the detector-system interaction 
\begin{equation}\label{eq:wn1}
    H_\text{ds}(t\vert \{\xi_i \}) = \sum_{i=1}^N \sqrt{\gamma p(i)} \xi_i(t)h_0^{(i)},
\end{equation}
where $h_0^{(i)}$ is given by Eq.~\eqref{eq:pe2}, and, for all $i,j=1,2,3$ the $\{\xi_i(t)\}_i$ are delta-correlated white noises satisfying
\begin{equation}\label{eq:wn2}
    \mathbb{E}[\xi_i(t)] = 0, \quad \mathbb{E}[\xi_i(t)\xi_j(s)] = \delta_{ij}\delta(t-s).
\end{equation}
The equivalence is shown in Appendix~\ref{sec:several_white_noises}.

Several differences between Eqs.~\eqref{eq:ssde5} and \eqref{eq:sde1} are worth noting. 
Importantly, only the deterministic part of Eq.~\eqref{eq:ssde5}, \begin{equation}\label{eq:newsde}
     \partial_t \pi^\text{det}_s(t) =  - \sum_{i \in \mathcal{I}} \frac{\gamma p(i)}{2}\{A_i^\dagger A_i- \expval*{A_i^\dagger A_i}_t, \pi^\text{det}_s(t)\},
\end{equation}
respects purity, whereas the stochastic part does not. Above, $\pi_s^{\text{det}}$ denotes a density matrix that evolves deterministically. Hence, in general, having $\pi_s^2(t) = \pi_s(t)$ does not imply $\pi_s ^2(t +\dd t) = \pi_s(t+ \dd t)$, which prevents associating a stochastic Schrödinger equation with Eq.~\eqref{eq:ssde5}. 
If a click is registered, the resultant state is mixed.
Because of such jumps toward a mixed state, we also note that this equation does not have the form of a conventional SME. Moreover, Eq.~\eqref{eq:newsde} also describes the deterministic evolution of a qubit continuously and simultaneously monitored by $n$ detectors, where no jump to a pure state is registered \cite{barchielli1991measurements}. Hence,
an observer who has access to a trajectory $\pi_s(t)$ that solves Eq.~\eqref{eq:newsde}, and has only partial information about the experimental setup, would be unable to discern whether $\pi_s(t)$ describes the former model or the one described by Eq.~\eqref{eq:ssde5} where no jump is registered. 

The Bloch vector of the stable, stationary state solution of Eq.~\eqref{eq:newsde} is \begin{equation}\label{eq:sde6}
    \bm{r}_\infty = -\frac{\Tr[\sum_{i\in \mathcal{I}} p(i)A_i^\dagger A_i \bm{\sigma}]}{\norm{\Tr[\sum_{i\in \mathcal{I}} p(i)A_i^\dagger A_i \bm{\sigma}]}}.
\end{equation}
(see Appendix~\ref{sec:appendix_jump_several_directions_static}). Hence, if $\dd N(t) = 0$ in Eq.~\eqref{eq:ssde5}, the updated state $\pi_s(t+\dd t)$ gets closer to the state represented by Eq.~\eqref{eq:sde6}, as $\pi_s(t)$ continuously evolves toward the former stationary state. 
 
Furthermore, given a particular trajectory $\pi_s(t)$ the probability of the click between $t$ and $t+\dd t$ is given by 
\begin{equation}\label{eq:sde7}
    \mathbb{E}_\alpha[\dd N(t)] = \gamma \expval{\sum_{i\in \mathcal{I}}p(i)A_i^\dagger A_i}_t\dd t, 
\end{equation}
where the expectation value $\expval{\bullet}_t$ is taken with respect to $\pi_s(t)$, and $\mathbb{E}_\alpha$ denotes the quantum average with respect to the detector readouts after the classical average $\mathbb{E}_i$ over the steering directions is taken. 

In connection to the last point, the fully averaged density matrix of the system is expressed as
\begin{equation}
\rho_s(t) = \mathbb{E}_\alpha[\pi_s(t)] = (\mathbb{E}_\alpha \circ \mathbb{E}_i)[\omega_s(t)] = \mathbb{E}[\omega_s(t)].   \end{equation}
This also implies, as anticipated, that taking the expectation value $\mathbb{E}_\alpha$ in Eq. \eqref{eq:ssde5} coincides with Eq. \eqref{eq:sde4.1.1}. In a similar spirit, the relation between the stochastic variables of Eq. \eqref{eq:sde1} with those of Eq. \eqref{eq:ssde5} with respect to the partial averaging $\mathbb{E}_i$ are (see Appendix~\ref{sec:appendix_jump_several_directions_static})
\begin{align}
\mathbb{E}_i[\chi_i(t)] &= p(i), \label{eq:sde9} \\
    \mathbb{E}_i[\chi_i(t)\dd N_i(t)] &=  \frac{p(i)\Tr[A_i^\dagger A_i \pi_s(t)]\dd N(t)}{\Tr[\sum_{j\in \mathcal{I}}p(j)A_j^\dagger A_j\pi_s(t)]}.
\end{align}

In Fig.~\ref{fig40:b}, we illustrate a particular trajectory of Eq.~\eqref{eq:ssde5} in the Bloch representation, and in Fig.~\ref{fig40:g} we show a transversal cross-section of the Bloch sphere. The erroneous states are given by Eq.~\eqref{eq:erre}. There, the Bloch vector of this trajectory evolves continuously along the surface of the sphere from its initial pure state and then continues with a jump toward a mixed state. Afterward, it can be seen that the state follows another continuous trajectory (now inside the Bloch ball) as it tries to reach the state $\bm{r}_\infty = (0,0,1)$ [see Eq.~\eqref{eq:sde6}].

Before introducing the following average hierarchy, we would like to contrast the quenched dynamic error considered here with the static error of erroneously prepared detectors (Sec.~\ref{ssec:errors_in_the_detector}) to highlight the importance of Eq.~ \eqref{eq:ssde5}.

We could have formulated the dynamic error of fluctuating steering directions in a different yet equivalent manner. Instead of having detectors prepared in $\rho_d = \ket{\Phi_d}\bra{\Phi_d}$ interacting with the system via the fluctuating (dimensionless) Hamiltonians in Eq.~ \eqref{eq:pe2}, and always measuring the local observable $S_d = \ket{\Phi_d}\bra{\Phi_d} - \ket*{\Phi_d^\perp}\bra*{\Phi_d^\perp}$, we could instead have detectors randomly prepared in states of the form $\ket{\Phi^{i}_d} = R(\theta_i,\varphi_i)\ket{\Phi_d}$ [cf. Eq.~\eqref{eq:7.3}] interacting with the steered system via the Hamiltonian $$\tilde{h}_0^{(i)} = \ket*{\Phi^{i,\perp}_d}\bra*{\Phi^{i}_d}\otimes A_i  +\hc$$ 

Now, after the interaction takes place, the local observable that must be measured is 
$$S_d^{(i)} = \ket*{\Phi_d^i}\bra{\Phi_d^i} - \ket*{\Phi^{i,\perp}_d}\bra*{\Phi^{i,\perp}_d}$$ 
instead of $S_d$, and the possible system outcomes are still given by Eq.~\eqref{eq:sde1}, if no average is taken. 

On the other hand, if the detector readouts are averaged, the state of the system is given by Eq.~ \eqref{eq:ssde5}. Taking a step further, performing a blind measurement (or a trajectory average) with the already averaged detector directions gives, once again, Eq.~\eqref{eq:sde4.1.1}. This LE is fundamentally different from Eq.~\eqref{eq:14}, obtained from erroneously prepared detectors interacting with the system through the \emph{same} Hamiltonian. Although the two dissipators appearing in Eq.~\eqref{eq:14} can be obtained by having two possible orthogonal detector states, there is no combination of steering directions that, upon total averaging, would induce a unitary channel, as each erroneous detector state interacts with the system with a \emph{different} Hamiltonian (see discussion above).   

\begin{figure*}[t]
  \begin{subfigure}{0.2\textwidth}
    \includegraphics[width=\linewidth]{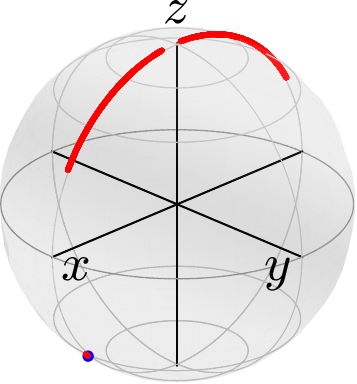}
    \caption{}  \label{fig40:a}
  \end{subfigure}%
  \begin{subfigure}{0.2\textwidth}
    \includegraphics[width=\linewidth]{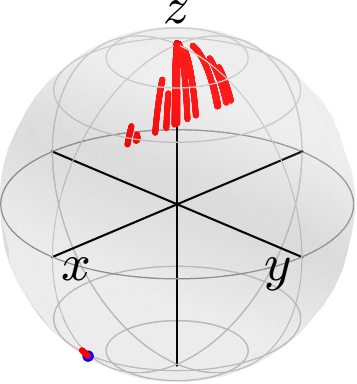}
    \caption{ } \label{fig40:b}
  \end{subfigure}%
  \begin{subfigure}{0.2\textwidth}
    \includegraphics[width=\linewidth]{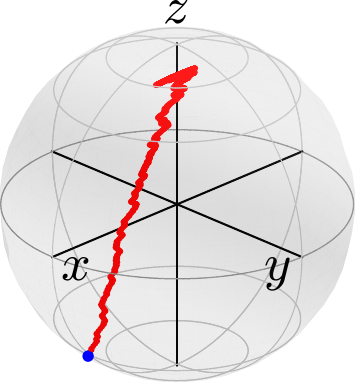}
    \caption{ } \label{fig40:c}
  \end{subfigure}%
   \begin{subfigure}{0.2\textwidth}
    \includegraphics[width=\linewidth]{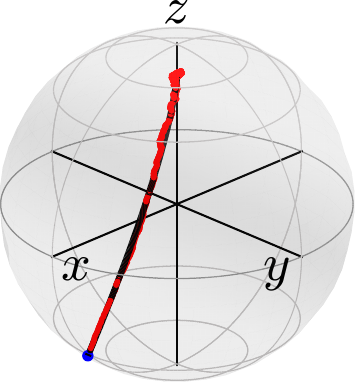}
    \caption{ } \label{fig40:d}
  \end{subfigure}%
   \begin{subfigure}{0.2\textwidth}
    \includegraphics[width=\linewidth]{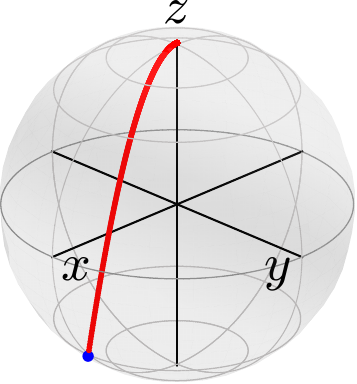}
    \caption{ } \label{fig40:e}
  \end{subfigure}%
 \vskip\baselineskip
 \begin{subfigure}{0.2\textwidth}
    \includegraphics[width=\linewidth]{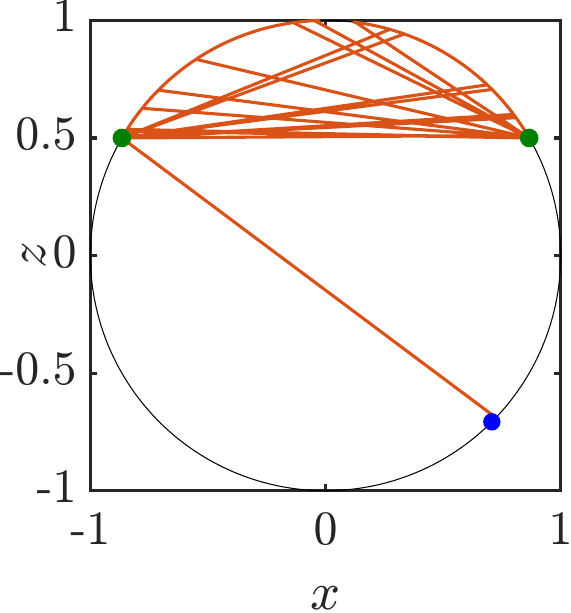}
    \caption{}  \label{fig40:f}
  \end{subfigure}%
  \begin{subfigure}{0.2\textwidth}
    \includegraphics[width=\linewidth]{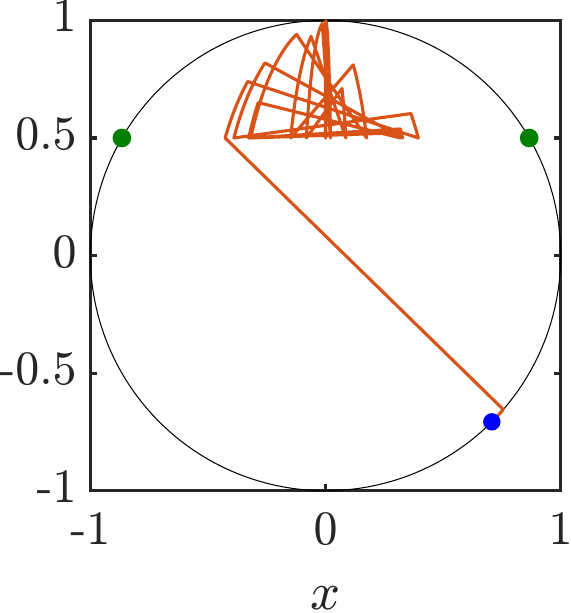}
    \caption{ } \label{fig40:g}
  \end{subfigure}%
  \begin{subfigure}{0.2\textwidth}
    \includegraphics[width=\linewidth]{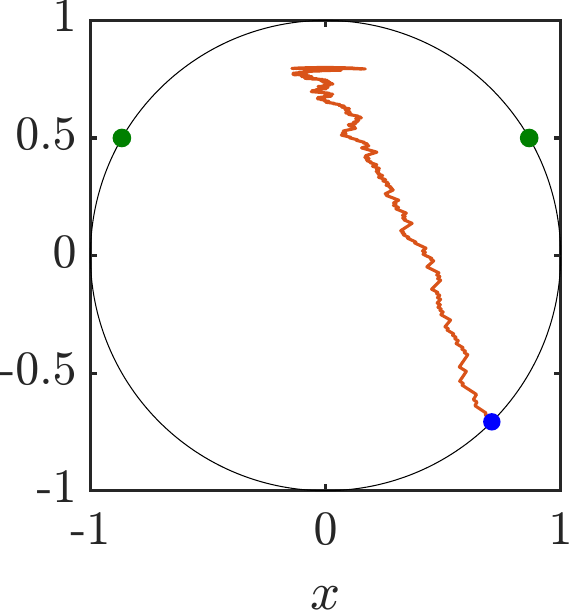}
    \caption{ } \label{fig40:h}
  \end{subfigure}%
   \begin{subfigure}{0.2\textwidth}
    \includegraphics[width=\linewidth]{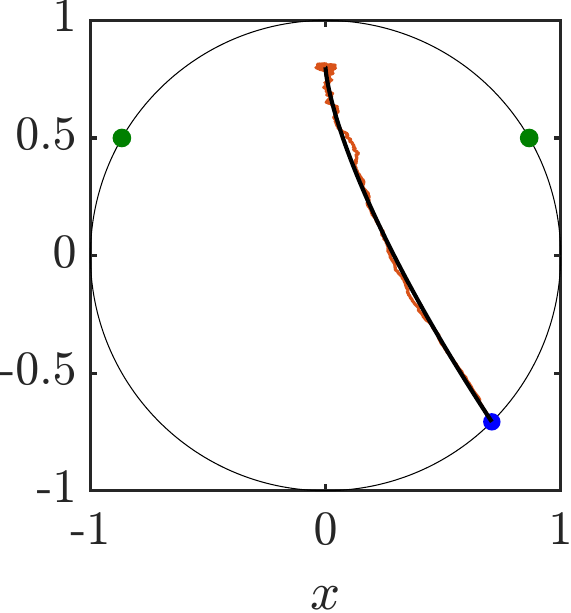}
    \caption{ } \label{fig40:i}
  \end{subfigure}%
   \begin{subfigure}{0.2\textwidth}
    \includegraphics[width=\linewidth]{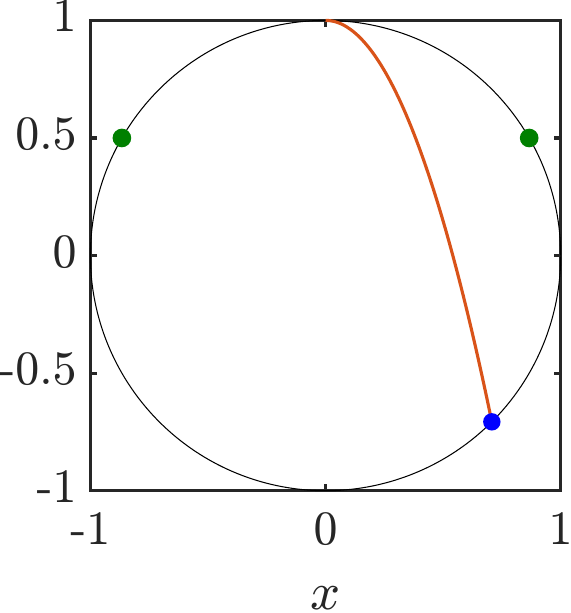}
    \caption{ } \label{fig40:j}
  \end{subfigure}%
  \caption{Dynamics of the Bloch vector of a steered qubit with $\mathcal{R} = \{(\pi/3,0,0.5),(\pi/3,\pi,0.5) \}$ as the set of target states [cf. Eq.~\eqref{eq:23}] under different averaging hierarchies for $\gamma = 0.1$, initial point $\bm{r}(0) = (1,0,-1)/\sqrt{2}$ (represented by the blue dot), and $\delta t = 0.1$. The two target states in $\mathcal{R}$ are represented by the two green dots in (f)-(j). Panels (a)-(c) and their respective transverse cross-sections in panels (f)-(i) display several representative quantum trajectories, corresponding to full stochasticity retained (a), averaging over the steering directions (b), and averaging over detector outcomes (c), which are solutions of Eqs.~\eqref{eq:sde1},~\eqref{eq:ssde5}, and~\eqref{eq:sde11}, respectively. Panels (e) and (j) show the averaged dynamics of the Bloch vector for ideal steering toward the north pole [see Eq.~\eqref{eq09}]. The jumps from a prior state are represented in panels (f) and (g) by straight lines. Panels (d), in black, and (i), in red, represent the analytical [see Eqs.~\eqref{eq:lind01}-\eqref{eq:lind02}] and numerical fully averaged dynamics of the erroneous steering with $\mathcal{R}$ given as above. The numerical average is performed over $10^3$ trajectories. We note that the final state, although not pure, is closer to the north pole than the midpoint of the straight line (not shown) connecting the two green dots. In fact, for all averaging hierarchies, under a sufficient number of steering steps, the dynamics of the Bloch vector associated with the steered state gets locked above the line joining the two (erroneous) target states.}
  \label{fig:40}
\end{figure*}

\subsubsection{Averaged detector outcomes}

Let us now discuss the third averaging hierarchy, where first, the measurement stochasticity is averaged out, but the directional stochasticity is kept. For this case, the detector readouts in Eq.~\eqref{eq:sde1} are averaged so that the SME that describes the system dynamics is given by (see Appendix~\ref{sec:appendix_jump_several_directions_static} for the derivation; see also Fig.~\ref{fig:00001})
 \begin{equation}\label{eq:sde11}
     \dd \sigma_s(t) = \sum_{i \in \mathcal{I}}\gamma \mathcal{D}(A_i)\sigma_s(t)\chi_i(t)\dd t,
 \end{equation}
where [cf. Eq.~\eqref{eq:sde8}]
\begin{equation}\label{eq:sde12}
    \sigma_s(t) \coloneqq \mathbb{E}_\alpha[\omega_s(t)],
\end{equation}
relates the partially averaged density matrix $\sigma_s(t)$ with the quantum trajectories $\omega_s(t)$ containing the two stochastic processes. One can see that taking the mean of Eq.~\eqref{eq:sde11} yields Eq.~\eqref{eq:sde4.1.1}. The SME~\eqref{eq:sde11} describes the random appearance of all the possible dissipation channels where only one is active within each steering step. Contrary to the two previous SMEs, Eqs.~\eqref{eq:sde1} and \eqref{eq:ssde5}, Eq.~\eqref{eq:sde11} never respects purity. An example of this is shown in Fig.~\ref{fig40:c}, where $\mathcal{R}$ is again given by 
Eq.~\eqref{eq:erre}.

\begin{figure*}[t]
\centering
 \begin{subfigure}{0.9\textwidth}
    \includegraphics[width=\linewidth]{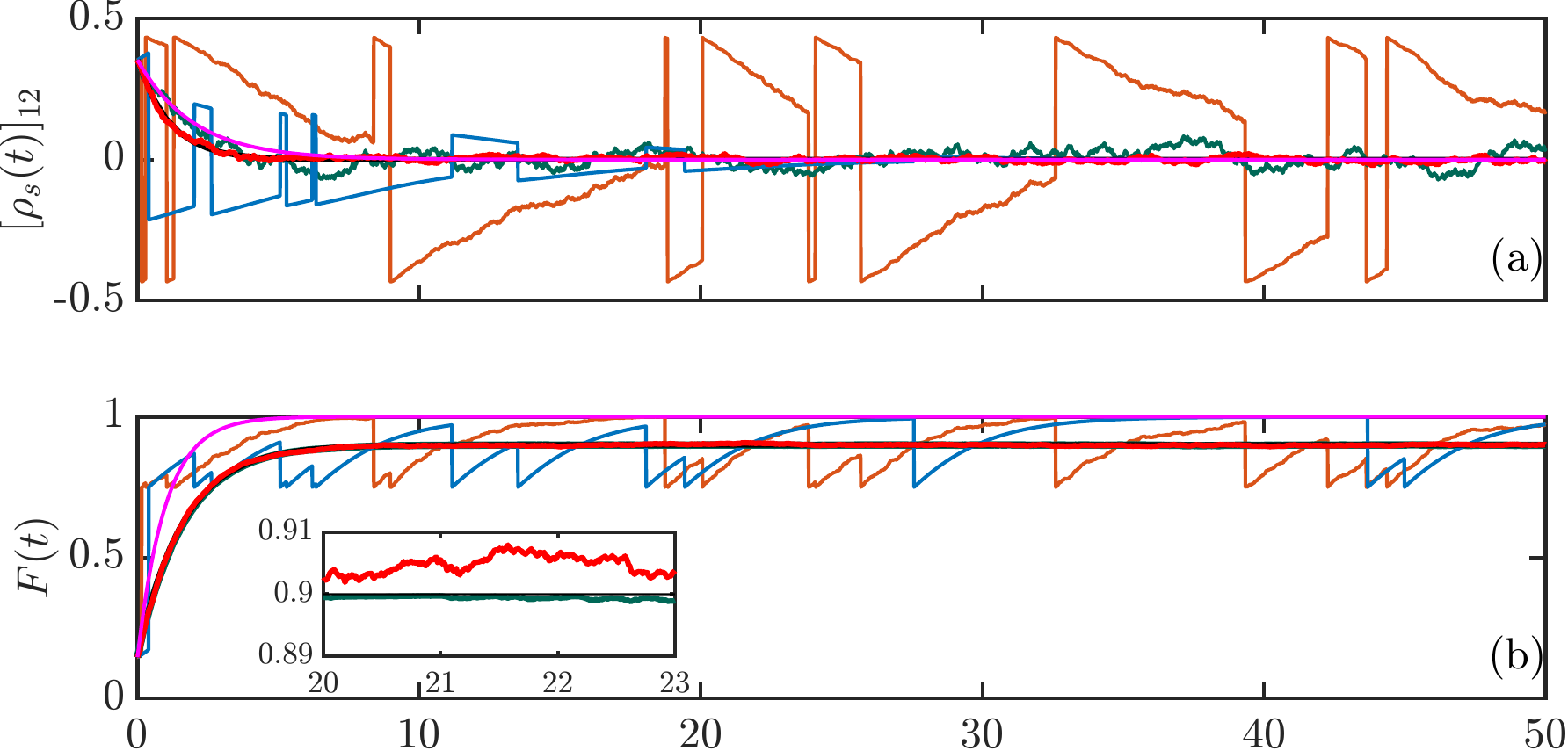}
  \end{subfigure}%
   \vskip\baselineskip
   \begin{subfigure}{0.9\textwidth}
    \includegraphics[width=\linewidth]{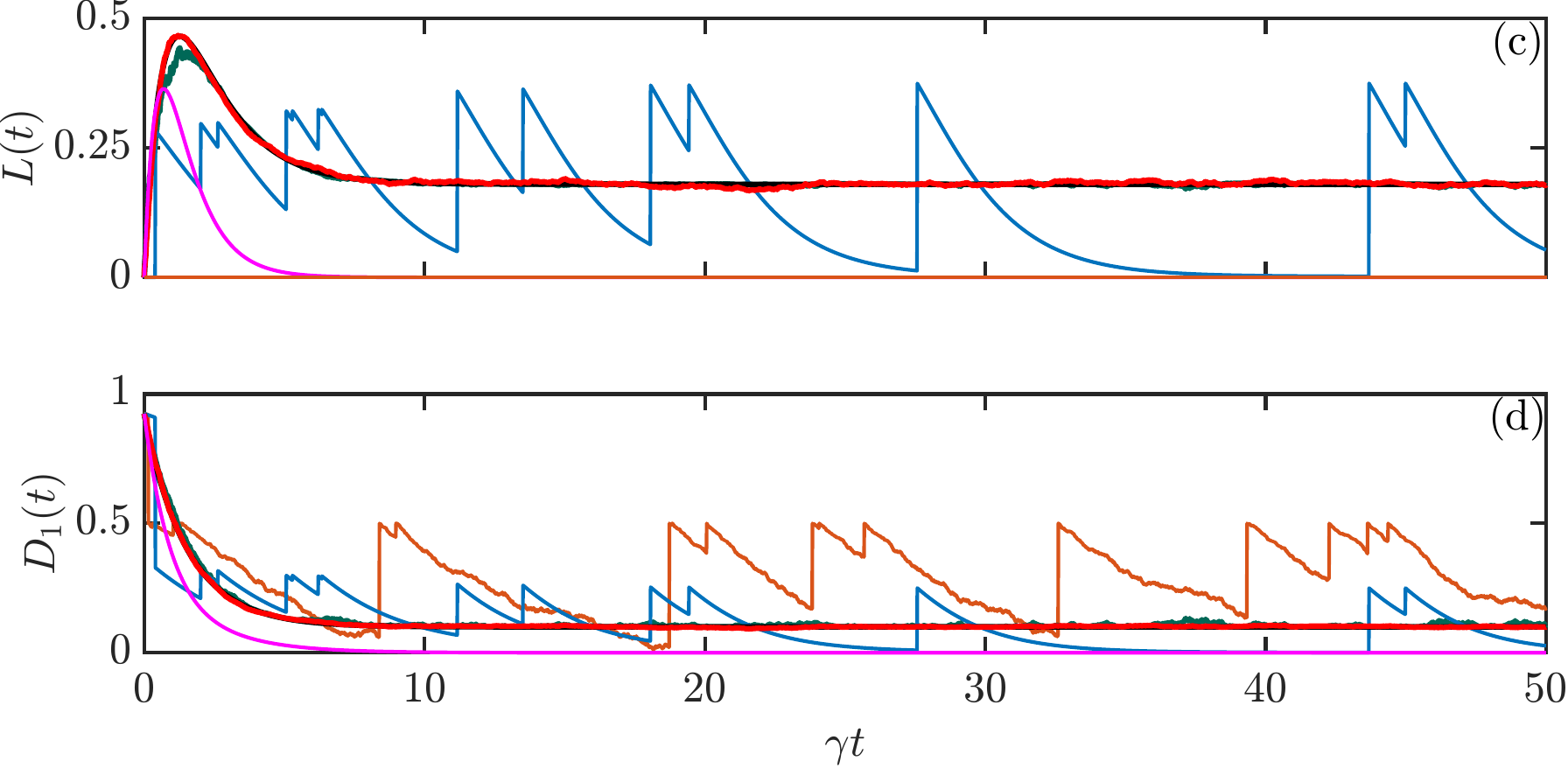}
  \end{subfigure}%
  \caption{Entries and distance measures for different averaging hierarchies as functions of time with the set of steering states Eq.~\eqref{eq:erre}. (a) Coherences of the density matrices. (b) Fidelity $F(t)$ [Eq.~\eqref{eq:me1}]. (c) linear entropy $L(t)$ [Eq.~\eqref{eq:2.7}] of several density matrices---corresponding either to stochastic master equations or Lindblad equations---with the ideal target state $\rho_\oplus = \ket{\uparrow}\bra{\uparrow}$. (d) Trace distance $D_1(t)$ [Eq.~\eqref{eq:me2}].  Orange curves correspond to the results of a quantum trajectory of Eq.~\eqref{eq:sde1} (cf. Figs.~\ref{fig40:a} and~\ref{fig40:f}), where both the classical and quantum stochastic process are present. Blue curves correspond to a quantum trajectory of Eq.~\eqref{eq:ssde5} (cf. Figs.~\ref{fig40:b} and~\ref{fig40:g}), where the classical average over the steering directions $\mathbb{E}_i$ is taken. Red curves show the quantum trajectory solving Eq.~\eqref{eq:sde11} (cf. Figs.~\ref{fig40:c} and~\ref{fig40:h}), where the quantum average $\mathbb{E}_\alpha$ over detector readouts is taken. Magenta and black curves correspond to the solution of the ideal LE~\eqref{eq09} (cf. Figs.~\ref{fig40:e} and \ref{fig40:j}) and the fully (i.e., with respect to quantum and classical stochasticity) averaged LE~\eqref{eq:sde4.1.1} (cf. Figs.~\ref{fig40:d} and \ref{fig40:i}), respectively. Green curves correspond to the average over $10^3$ quantum trajectories.
  } 
  \label{fig:301}
\end{figure*}

\subsubsection{Comparison of the steering dynamics for different averaging hierarchies}

In Fig.~\ref{fig:301}, we present the distance measures and the entries of density matrices of trajectories corresponding to the different averaging hierarchies, as well as for an ideal (single steering direction-north pole), fully averaged steering. For the trajectory displaying full stochasticity [see panel (b) in orange],  when no jump occurs, the fidelity grows continuously and monotonically, sometimes getting very close to unity [as shown by $D_1(t)$ in panel (d) also in orange]. This behavior repeats after each jump until the steering stops. It is also worth noting that, as long as the initial state is pure, the linear entropy [depicted in panel (c)] is always zero for this trajectory and every other described by the SME~\eqref{eq:sde1}.

\begin{figure*}[t!]
    \centering
    \includegraphics[width=\linewidth]{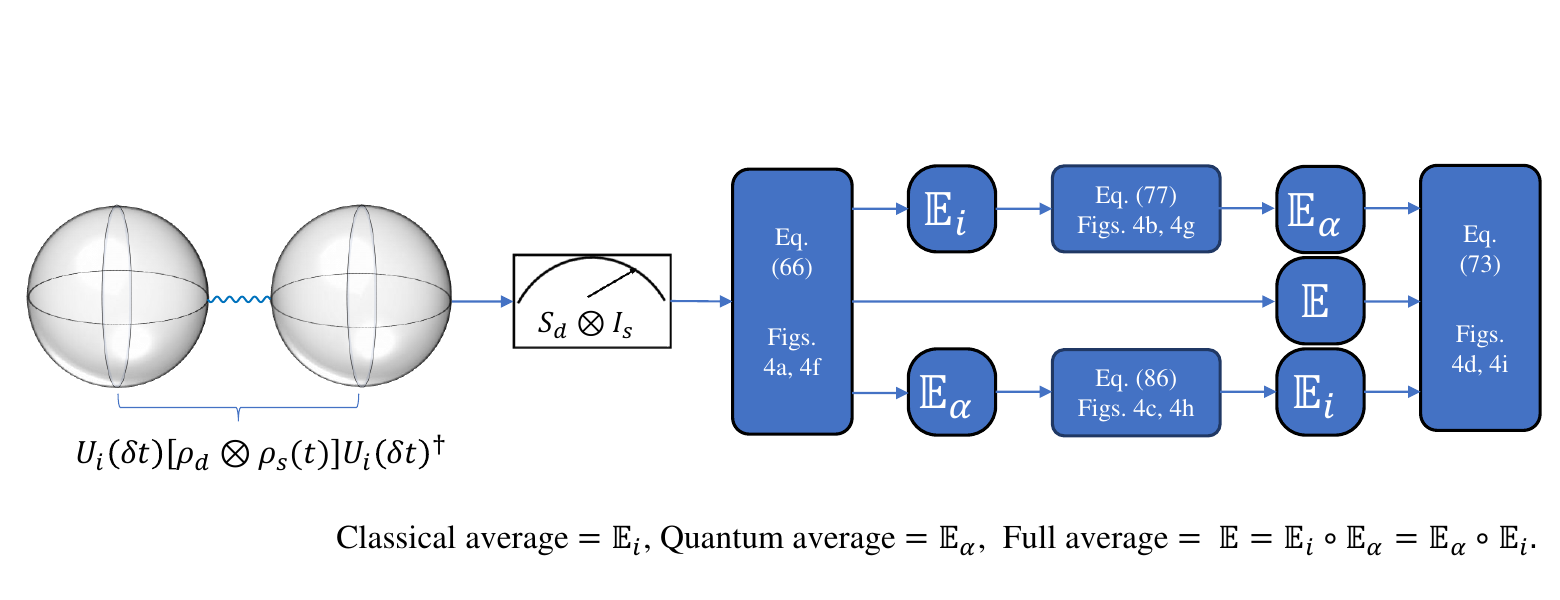}
    \caption{Illustration of the averaging hierarchies for measurements with randomly chosen measurement directions at each protocol step. After detector and system jointly evolve under the unitary operator $U_i(\delta t) \coloneqq \exp(-i\sqrt{\gamma \delta t}h_0^{(i)}),$ which is induced by $h_0^{(i)}$ given by Eq.~\eqref{eq:pe2}, the observable $S_d\otimes I_s$, with $S_d = \ket{\Phi_d}\bra*{\Phi_d} - \ket*{\Phi_d^\perp}\bra*{\Phi_d^\perp}$, is measured. The detector-system interaction $h_0^{(i)}$ for the measurement direction $i$ appears randomly with a probability $p(i)$. The stochastic master equation describing the system outcomes once $S_d\otimes I_s$ is measured is given by Eq.~\eqref{eq:sde1}. An example of a trajectory corresponding to this stochastic master equation is shown in Fig.~\ref{fig40:a} where the steering directions $\mathcal{R}$ are given by Eq.~\eqref{eq:erre}. If no selection of the steering directions is made, Eq.~\eqref{eq:sde1} must be averaged with respect to the directions. This classical average is denoted by $\mathbb{E}_i$. Upon taking the classical average, the resulting stochastic master equation is given by Eq.~\eqref{eq:ssde5}, and a particular trajectory is shown in Fig.~\ref{fig40:b} for the same $\mathcal{R}$ as before. Averaging Eq.~\eqref{eq:ssde5} with respect to the detector outcomes, which is a quantum mechanical average denoted by $\mathbb{E}_\alpha$, unravels the Lindblad master equation~\eqref{eq:sde4.1.1}. Returning to Eq.~\eqref{eq:sde1}, one can first take the quantum average and get Eq.~\eqref{eq:sde11}. A particular trajectory of this SME is shown in Fig.~\ref{fig40:c} with the same $\mathcal{R}$ as before. Averaging the latter equation over the steering directions (i.e., taking the classical average) also gives the LE.~\eqref{eq:sde4.1.1}. The overall procedure shows that the order in which the averages are taken is immaterial to obtain the same LE. Hence, the total average is $\mathbb{E} = \mathbb{E}_\alpha\circ \mathbb{E}_i = \mathbb{E}_i \circ \mathbb{E}_\alpha$, implying that  Eq.~\eqref{eq:sde4.1.1} can be directly unraveled from Eq.~\eqref{eq:sde1}. Figure~\ref{fig40:d} shows a trajectory (i.e., a solution) of the former equation, representing the fully averaged dynamics with the same steering directions treated above. 
 }
    \label{fig:00001}
\end{figure*}
Once the averaging over steering directions is carried out (see Figs.~\ref{fig40:b} and \ref{fig40:g}), the trajectory evolution is determined by Eq.~\eqref{eq:ssde5}. From Fig.~\ref{fig:301}-(b) (in blue), we can observe how the fidelity increases monotonically until a jump occurs. In contrast to the fully stochastic scheme, the jump is toward a mixed state that gets closer and closer to the center---this point is eventually reached---of the imaginary line joining the two erroneous target states represented by the two green dots in Figs~\ref{fig40:f}-\ref{fig40:j}. Similarly to the previous scheme, the ideal target state can be reached [cf. Eq. \eqref{eq:sde6}] if no jump occurs during a sufficient time lapse. 

If we average over the detector readouts (cf. Figs.~\ref{fig40:c} and \ref{fig40:h}), the three distance measures behave similarly to those obtained from the fully averaged dynamics (cf. Figs.~\ref{fig40:d} and \ref{fig40:i}), as shown in panels of Fig.~\ref{fig:301} by green, red and black curves, respectively. Note that the difference relies on a fluctuation of the measures in the former stochastic case around the stationary state values. The fluctuations observed in the red curves, representing the numerically-averaged dynamics, are due to the finite number of averaged trajectories [see, e.g., the sub-figure in panel (b)]. 

We would like to reiterate that, contrary to the case of averaged detector readouts, obtaining the fully averaged trajectory followed by the Bloch vector (see Fig.~\ref{fig40:d}) from both the fully stochastic picture and the averaged directions is by no means self-evident. 
Furthermore, even though the fully averaged Bloch trajectories corresponding to the erroneous and ideal protocols have different characteristics such as curvatures and end-points (compare Figs.~\ref{fig40:d} and \ref{fig40:i} with Figs.~\ref{fig40:e} and \ref{fig40:j}), the respective dependencies of the distance measures on time for a given initial state have overall similar qualitative features (see Fig.~\ref{fig:301}). For example, in Fig.~\ref{fig:301}c for $L(t)$, both magenta (erroneous) and black (ideal) curves first grow, attain a maximum, and then decay exponentially towards a constant value. The difference is only in the saturation value (finite for the erroneous protocol versus zero for the ideal one).

In Fig.~\ref{fig:00001}, we summarize in a diagram the averaging hierarchies corresponding to Eqs.~\eqref{eq:sde3}, \eqref{eq:ssde5} and \eqref{eq:sde11}, and how they unravel Eq. \eqref{eq:sde4.1.1}. A correspondence with the trajectories in Fig. \ref{fig:40} is also shown. 

\subsubsection{Stationary state of the fully averaged dynamics and small-error approximation}\label{ssec:stationary_state_of_the_fully_averaged}

There is an infinite number of probability distributions (both discrete and continuous) over the sphere that can be associated with the steering directions Eq.~\eqref{eq:23}. However, for simplicity, we will use the set of two steering directions
\begin{equation}\label{eq:dir}
 \mathcal{R} = \{(\theta,0;p),(\theta,\pi;1-p) \}
\end{equation}
as it simplifies the process of obtaining analytical results. In Appendix~\ref{sec:two_continuous_steering_directions}, we consider two continuously distributed steering directions: a uniform distribution between two angles and a von Mises distribution. Note that $\rho_\oplus = \ket{\uparrow}\bra{\uparrow} \notin \mathcal{R}$ for $\theta >0$. Given the steering set Eq.~\eqref{eq:dir}, we shall study how robust the protocol is to this error by performing a series expansion of the quantifiers in $p$ and $\theta$. We shall see that $p=1/2$ is the most favorable condition, as the leading order of the quantifiers is of fourth order in $\theta$.

Two relevant entries of the stationary state solution  of the fully averaged LE [cf. Eq.~\eqref{eq:sde4.1.1}]
\begin{equation}\label{eq:field0}
    \partial_t\rho_s(t) =\gamma \left[ p  \mathcal{D}\textbf{(}A(\theta)\textbf{)} + (1-p) \mathcal{D}\textbf{(}A(-\theta)\textbf{)} \right]\rho_s(t),
\end{equation}
corresponding to the steering set Eq.~\eqref{eq:dir} are 
\begin{widetext}
\begin{align}
        [\tilde\rho_{\oplus}]_{11} &= \frac{4 + (1-p)p + 4[1+(1-p)p]\cos\theta + (p-1)p(4\cos3\theta + \cos 4\theta)}{8 + 2(1-p)p + 2(p-1)p \cos 4\theta}, \label{eq:7.9}\\
         [\tilde\rho_{\oplus}]_{12} &= \frac{2(-1+2p)\sin\theta}{4+(1-p)p + (-1+p)p\cos 4 \theta}. \label{eq:7.9.1}
\end{align}
\end{widetext}
In Fig.~\ref{fig:3}, we depict the distance measures, Eqs.~\eqref{eq:me1}-\eqref{eq:me2}, comparing $\tilde\rho_\oplus$ with $\rho_\oplus$, as well as impurity of $\tilde\rho_\oplus$, Eq.~\eqref{eq:2.7}, as functions of $p$ and $\theta$. The domain of the latter variable is set as $\theta \in [0,\pi]$ since all the quantifiers are even with respect to $\theta$. The impurity (see Fig.~\ref{fig3:b}) has a thick crest centered at the plane $\theta = \pi/2$, and its maximum $\max{L_{\infty}(p,\theta)} = 0.5$ is located at the intersection of the planes $\theta = \pi/2$ and $p = 1/2$. The density matrix corresponding to the values $(p,\theta) = (1/2,\pi/2)$ is maximally mixed, as the two dissipators in Eq.~\eqref{eq:field0} have the same strength and try to steer the system toward orthogonal states.

There is an interesting behavior of the trace distance as observed in Fig.~\ref{fig3:c}. Upon intersecting the surface $D_{1,\infty}(p,\theta)$ with constant planes of small $\theta$, the resulting curves, which are differentiable, resemble an absolute value centered at $p=1/2$ as shown in Fig.~ \ref{fig4:a}. The behavior of $D_{1,\infty}(p,\theta)$ for several fixed probabilities as a function of the angle is shown in Fig. \ref{fig4:b}, where the curves are pretty close to zero for $p= 1/2 - \varepsilon$, $0 < \varepsilon \ll 1$---this also holds for $p= 1/2+\varepsilon$ because $D_{1,\infty}(p,-\theta) = D_{1,\infty}(p,\theta)$.

Returning to the properties observed in Fig.~\ref{fig:4}, we can see a crossover between absolute-value--shaped curves centered around $p=1/2$, which are differentiable, to curves with a smaller curvature at $p=1/2$. This behavior can be seen analytically by performing a series expansion of $D_{1,\infty}(p,\theta)$ around different points. For example, for $p \in [0,1/2)\cup (1/2,1]$, we have
\begin{widetext}
\begin{multline}\label{eq:td0}
    D_{1,\infty}(p,\theta) = 
   \abs{1-\frac{p}{2}}\left( \theta +\frac{ 60 (p-1) p-1 }{24}\theta^3+\frac{120 (p-1) p \left[78 (p-1) p-49\right]+1 }{1920}\theta^5  \right)\\
    +\frac{16(p-1)p \textbf{[} 630 (p-1) p \left\lbrace 8 (p-1) p \left[157 (p-1) p-163\right]-229\right\rbrace+27011\textbf{]}-1}{1290240\abs{1-\frac{p}{2}} }\theta^7 + \mathcal{O}(\theta^9),
\end{multline}
\end{widetext}
where there are only odd powers of $\theta$. This is no longer the case if we expand $D_{1,\infty}(p,\theta)$ at $p = 1/2$ with respect to $\theta$: 
\begin{subequations}
\label{eqtd1}
\begin{align}
    D_{1,\infty}(p = 1/2, \theta) &= \frac{4\sin^2(\theta/2)}{3 + \cos(2\theta)} \label{eq:td1} \\ 
    &= \frac{\theta^4}{16} + \frac{\theta^6}{48} + \mathcal{O}(\theta^8). \label{eq:td1-1}
\end{align}
\end{subequations}
By comparing Eqs.~\eqref{eq:td0} and \eqref{eq:td1-1}, the protocol shows more robustness when $\theta \ll 1$ and $p=1/2$, as the leading power in $\theta$ is of order four instead of one.

The change of evenness in powers of $\theta$ in the series expansion of $D_{1,\infty}(p,\theta)$ is only observed in the trace distance; Figs.~\ref{fig3:a}-\ref{fig3:b} evidence the smoothness of $F_\infty(p,\theta)$ and $L_\infty(p,\theta)$ in $(0,1)\times(-\pi,\pi) \subset \mathbb{R}^2$. More precisely, expanding $F_\infty(p,\theta)$ and $L_\infty(p,\theta)$ in $\theta$ and $p \neq 1/2$ gives
\begin{align}
    F_\infty(p,\theta) &= 1 - \frac{1}{4}\left( 1 - 4p + 4p^2 \right)\theta^2\notag \\ 
    &\!\!\!\!\!\!\!+ \frac{1}{48}\left( 1 + 8p -104p^2 + 192p^3 - 96p^4 \right)\theta^4 + \mathcal{O}(\theta^6), \\
    L_\infty(p,\theta) &= 2p\left( 1 - 4p +6p^2 -3p^3\right)\theta^4 + \mathcal{O}(\theta^6).
\end{align}
\begin{figure}[t]
  \begin{subfigure}{0.25\textwidth}
    \includegraphics[width=\linewidth]{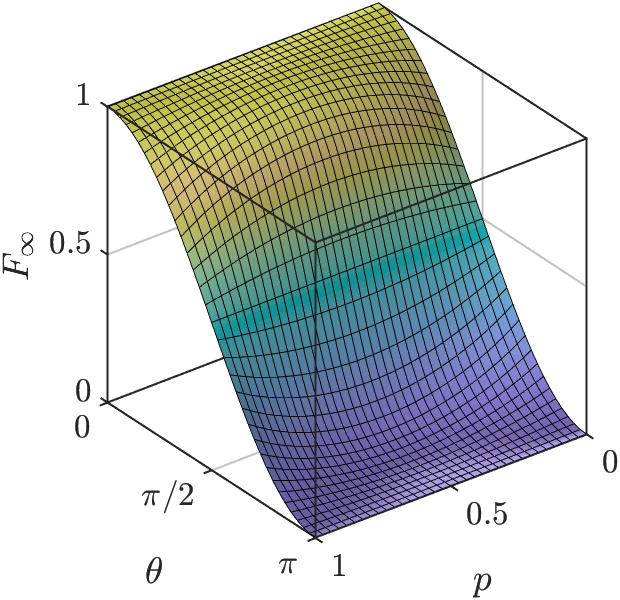}
    \caption{}  \label{fig3:a}
  \end{subfigure}%
  \begin{subfigure}{0.25\textwidth}
    \includegraphics[width=\linewidth]{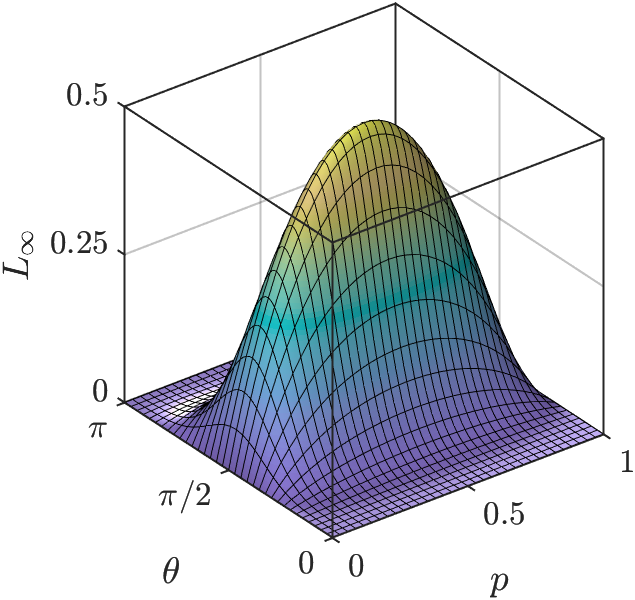}
    \caption{ } \label{fig3:b}
  \end{subfigure}%
\vskip\baselineskip
  \begin{subfigure}{0.25\textwidth}
    \includegraphics[width=\linewidth]{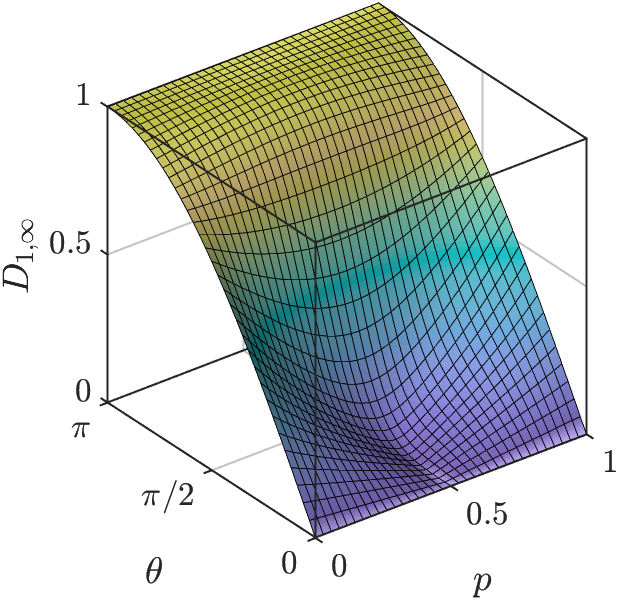}
    \caption{ } \label{fig3:c}
  \end{subfigure}%
  \caption{ Distance measures between the ideal target state $\rho_\oplus = \ket{\uparrow}\bra{\uparrow}$ and the actual target state $\tilde\rho_\oplus$, together with its impurity, as functions of $(p,\theta)$. (a) Fidelity $F_\infty(p,\theta)$ [Eq.~\eqref{eq:me1}]. (b) Impurity $L_\infty(p,\theta)$ [Eq.~\eqref{eq:2.7}]. There, we observe the trivial fact that states close to $(p,\theta)= (1/2,\pi/2)$, which corresponds to the maximally mixed state, are those with the greatest impurity. (c) Trace distance $D_1(p,\theta)$ [Eq.~\eqref{eq:me2}]. This surface shows an absolute-value-shaped region centered at $p = 1/2$.  } 
  \label{fig:3}
\end{figure} 
\begin{figure*}[t]
  \begin{subfigure}{0.40\textwidth}
    \includegraphics[width=\linewidth]{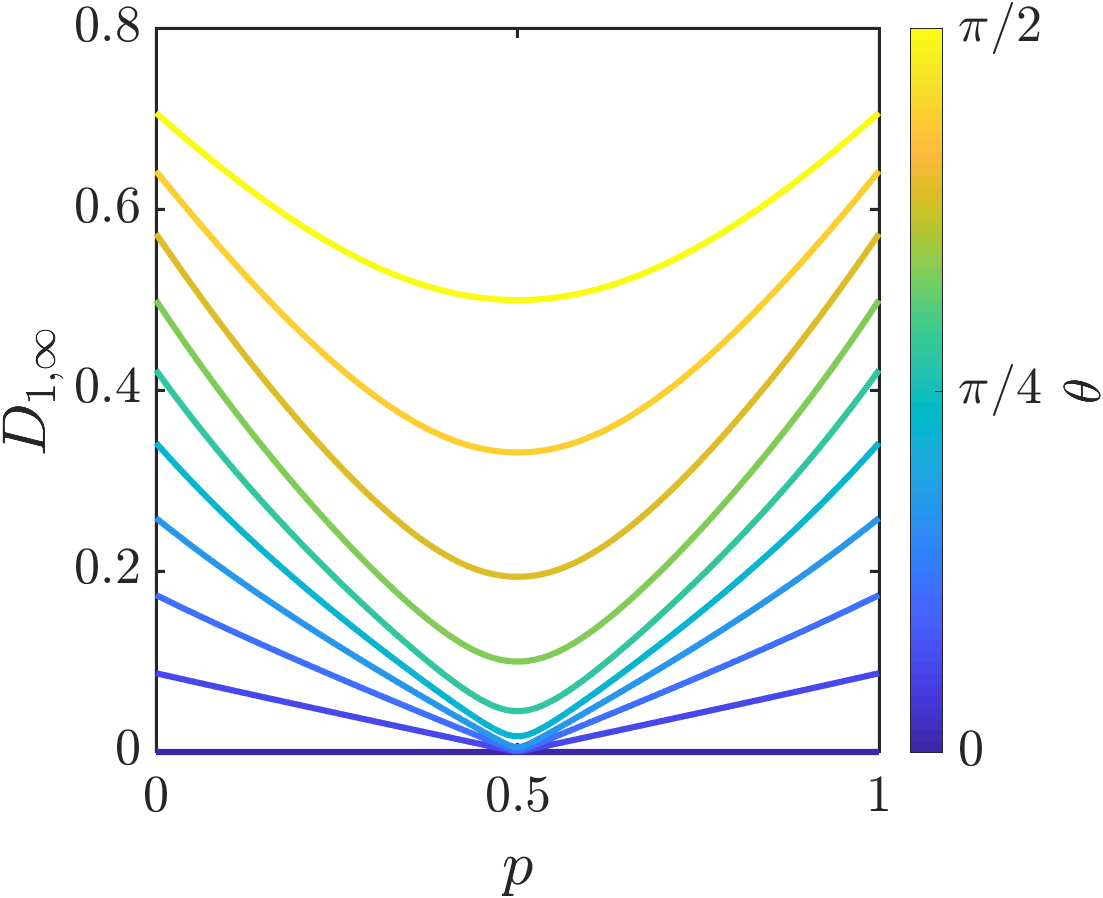}
    \caption{}  \label{fig4:a}
  \end{subfigure}%
  \begin{subfigure}{0.38\textwidth}
    \includegraphics[width=\linewidth]{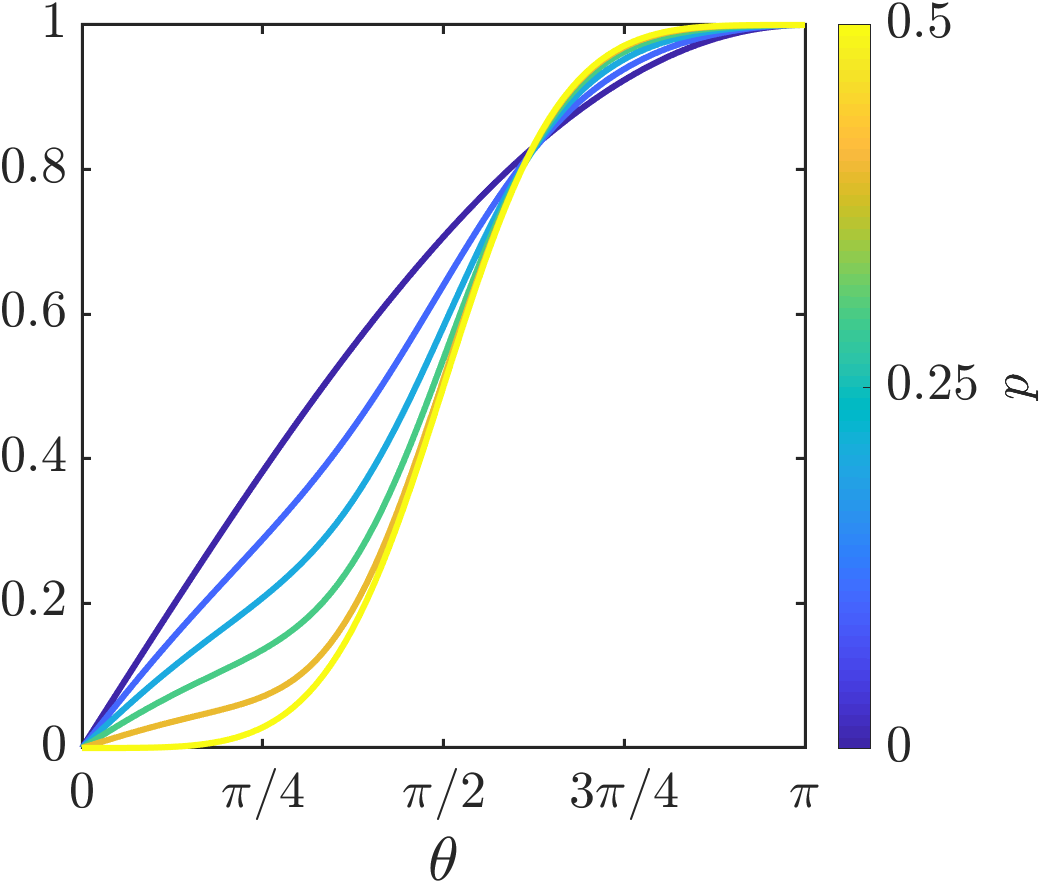}
    \caption{ } \label{fig4:b}
  \end{subfigure}%
  \caption{(a) Sections of $D_{1,\infty}(p,\theta)$ [Eq.~\eqref{eq:me2}] with constant $\theta \in [0,\pi/2]$. The curves resemble an absolute value as the angle gets closer to zero. (b) Sections of $D_{1,\infty}(p,\theta)$ with constant $p \in [0,1/2]$. The change of concavity for curves with $0 \leq \theta \lesssim \pi/2$ is reflected on the change of the powers in $\theta$ in the series expansion of $D_{1,\infty}(p,\theta)$ Eqs.~\eqref{eq:td0}-\eqref{eqtd1}.}
  \label{fig:4}
\end{figure*}
The above two series coincide term-by-term with their expansion  at $p = 1/2$,
\begin{align}
    F_{\infty}(p = 1/2, \theta) &= \frac{1}{2} + \frac{2\cos\theta}{3 + \cos2\theta} = 1- \frac{\theta^4}{16} +  \mathcal{O}(\theta^6), \label{eq:td2}\\
    L_\infty(p=1/2,\theta) &= \frac{2\sin^4\theta}{(3 + \cos2\theta)^2} = \frac{\theta^4}{8} +  \mathcal{O}(\theta^6).\label{eq:td3}
\end{align}
Note that the second order term in $\theta$ of $F_\infty(p,\theta)$ vanishes at $p = 1/2$, and, in contrast with $D_{1,\infty}(p,\theta)$ with $p \neq 1/2$ [see Eq. \eqref{eq:td0}], the fidelity and impurity have only even powers of $\theta$.

Returning to the discrete distribution of states given by Eq.~\eqref{eq:erre}, the uniqueness of the actual target state and its closeness to the ideal one observed in Eq.~\eqref{eqtd1} can be better understood if we adopt the geometrical point of view of the fully averaged dynamics and make use of the averaging hierarchies previously shown.

Let us first address the former point by implementing the superposition principle of fields. Each of the dissipators present in the Lindbladian in Eq.~\eqref{eq:field0}
induces a vector field in the Bloch ball with a single fixed point on the sphere's surface. Therefore, when adding two fields, there will only be a single stable fixed point \emph{within} the Bloch sphere. This is easily extended when there are more dissipators of the form $\mathcal{D}\textbf{(}A(\theta,\varphi)\textbf{)}$ with angles being either continuous or discretely distributed. The fields are added two by two until all are summed, resulting in a single stationary state---In the continuous case, they are integrated. 

The closeness of the actual target state to the ideal one when the steering states are given by Eq.~\eqref{eq:dir} with $p = 1/2$ and $\theta \ll 1$---even for a broader range of angles---is evident in the three averaging hierarchies presented at the beginning of the section. When the stochasticity of both the direction and the detector outcomes is maintained, the steered Bloch vector will eventually jump to one of the target states in $\mathcal{R}$. Then its trajectory will be exclusively contained in the arc connecting the two target states and passing through the north pole (see Figs.~\ref{fig40:a} and \ref{fig40:f}).

If the steering directions are averaged, the Bloch vector will try to reach the north pole (the ideal target state) when no jump occurs. If a jump occurs, the $z$-component of the Bloch vector will coincide with the $z$-component of both non-ideal target states. Subsequently, the trajectory will be continuous and, once more, directed toward the ideal target state until another jump occurs. Nevertheless, the absolute value of the $x$-component of the Bloch vector reduces after each jump. This behavior repeatedly occurs with the consequence of having the jumps concentrated in the middle point of the imaginary line $\bar{l}$ passing through the two erroneous target states in $\mathcal{R}$. Likewise, the continuous trajectories also get closer and closer to the $z$-axis segment that passes through the ideal target state (the north pole) and the middle point of $\bar{l}$. Therefore, all the trajectories followed by the Bloch vector will be contained in the intersection of the spherical cap, whose base contains $\bar{l}$, with the plane containing these two points through which the segment $\bar{l}$ passes (see Figs.~\ref{fig40:b} and \ref{fig40:g}). 

Lastly, when the detectors' outcomes are ignored, only one of the dissipators---$\mathcal{D}\textbf{(}A(\theta)\textbf{)}$ or $\mathcal{D}\textbf{(}A(-\theta)\textbf{)}$---will be active during each steering step. After a sufficiently long steering time, regardless of the initial point, the trajectory followed by the Bloch vector will wiggle close to a point above $\bar{l}$ as it is evident in Figs.~\ref{fig40:c} and \ref{fig40:h}.  Once more, all the trajectories within the three averaging hierarchies with the $\mathcal{R}$ mentioned above are within the spherical cap described in the previous case.  The overall result upon performing a total average over the stochastic variables---and as long as $0< \theta < \pi/2$---is having a stationary state always lying above $\bar{l}$; and the closer the angle $\theta$ is to zero, the closer the stationary state is to the north pole (see Fig.~\ref{fig:40}).

\subsection{Error in the  detector-system coupling strength}\label{ssec:errors_coupling_constant}

The detector-system coupling strength may fluctuate because of imperfect control. This section discusses how this type of error affects the final state of the steering protocol. This error can be quenched or time-dependent. However, we shall see that the averaged dynamics of the system state will be described by the same LE as in the ideal steering case, with the dark state being the ideal target state, implying that the protocol is robust to this error. 

Without loss of generality, let us choose again the ONBs $\mathcal{B}_\oplus = \{\ket{\uparrow}, \ket{\downarrow}\}$ and $\mathcal{B}_d = \{\ket{\uparrow}, \ket{\downarrow}\}$, such that the interaction Hamiltonian is [cf. Eq.~\eqref{eq:Hds-proj}]
\begin{equation}
\label{eq:cc0}
H_\text{ds} = J(t)(\sigma^+\otimes\sigma^- + \text{h.c.}) = J(t) h_0,  
\end{equation}
so that the ideal target state is  the north pole of the Bloch sphere. Here, $J(t)$ controls the detector-system interaction strength and can have quenched
[with $J(t)=\sum_k J_k \delta_{t,k\delta t}$, where $\delta_{t,t'}$ is the Kronecker delta-symbol and $J_k$ is randomly chosen from a predefined set $\mathcal{J}$ at each step] or generic time-dependent error contributions. Below, we discuss both cases one by one.

\subsubsection{Quenched error}

Let $\mathcal{I}\coloneqq \{1, \ldots, N \}$, and let $\mathcal{J} \coloneqq \{J_n \}_{n\in \mathcal{I}}$ be a valid set of detector-system coupling strengths, i.e., such that for a given interaction time $\delta t$, $J_n\delta t/\pi  \notin  \mathbb{Z}$ and $J_n\delta t/(\pi/2)  \notin  2\mathbb{Z}/\mathbb{Z}$, and there is at least one $J_n\neq 0$ with non-zero probability to occur. Suppose that, after each measurement step, the coupling acquires a value in $\mathcal{J}$, say, $J_n$, with a given probability $p(n)$ with the conditions stated above. 

The quantum trajectories accounting for the stochastic appearance of coupling constants in $\mathcal{J}$ can also be described by SMEs like Eqs.~\eqref{eq:sde1}, \eqref{eq:ssde5}, and \eqref{eq:sde11} by mapping $A_n \rightarrow A$ for all $n \in \mathcal{I}$, and $\gamma$ to $\gamma_n \coloneqq \lim_{\delta t \to 0}J_n^2\delta t$. 
Under this mapping, it is easy to see that a single target state will be shared by $n$ dissipators that are proportional to each other. Therefore, it is sufficient to focus only on the fully averaged dynamics described by 
\begin{equation}\label{eq:cc1}
    \partial_t\rho_s(t) = \sum_{n \in \mathcal{I}}p(n)\gamma_n \mathcal{D}(\sigma^+)\rho_s(t).
\end{equation}
This LE has the same dissipator as in 
Eq.~\eqref{eq09} with $A\cong \sigma^+$, but the effective channel strength is different: $\gamma\to \sum_{n \in \mathcal{I}}p(n) \gamma_n$. Thus, the trajectories followed by the states are precisely the same as for the ideal steering, yet they traverse the Bloch ball at a different speed. 

If $J$ belongs to a continuous distribution, such as a Gaussian one $\mu_J$ with zero mean and variance $\sigma$, the fully averaged state after a steering step is given by 
\begin{align}\label{eq:cc1.1}
    \rho_s(t + \delta t) &= \int_\mathbb{R} \Tr_d \left\{ \exp[-i J\delta t  \ad(h_0)]\rho_d \otimes {\rho}_s(t)\right\}\dd \mu_J\notag \\ &= \Tr_d\left[ \exp( -\frac{\sigma^2\delta t^2}{2} \ad^2(h_0) )\, \rho_d \otimes {\rho}_s(t)\right],
\end{align}
where $\dd \mu_J$ is the one-dimensional Gaussian measure. Making use of the WM limit  $\lim_{\delta t \to 0}\delta t \sigma^2 = \gamma = \text{const}$ upon performing a series expansion in Eq.~\eqref{eq:cc1.1}  and using the ONB $\mathcal{B}_\oplus$, gives the LE 
\begin{equation}\label{eq:cc1.2}
    \partial_t{\rho}_s(t) = \gamma \mathcal{D}(\sigma^+){\rho}_s(t),
\end{equation}
where we took the usual decomposition of $h_0$ as in Eq.~\eqref{eq:cc0}. This WM limit is similar to the one used in  Ref.~\cite{francoise_quantum_2006}.

\subsubsection{Time-dependent error}

Promoting the detector-system coupling to be a white noise variable, i.e., $J \mapsto \Upsilon\xi(t)$ with 
\begin{equation}\label{eq:cc2}
    \mathbb{E}[\xi(t)] = 0, \quad \mathbb{E}[\xi(t)\xi(s)] = \delta(t-s), 
\end{equation}
makes the error time-dependent in our categorization (see Sec.~\ref{sec:definition_of_the_errors}). Hence, by using the It\^o calculus rules \cite{breuer2002theory} 
\begin{equation}\label{eq:cc3}
    (\dd X_t)^2 = \dd t, \quad \dd X_t \dd t = 0, \quad (\dd t)^2 = 0,
\end{equation}
with the Wiener increment
\begin{equation}
\dd X_t \coloneqq \int_0^t\xi(s)\dd s,
\end{equation}
the unitary operator for an infinitesimal time reduces to 
\begin{align}\label{eq:cc4}
    U(\dd t) = e^{-i\Upsilon h_0 \dd X_t } 
    = I_{ds} - i \Upsilon h_0\dd X_t -\frac{\Upsilon^2}{2}h_0^2 \dd t.
\end{align}
Therefore, the updated system state after taking the blind measurement is
    \begin{equation}\label{eq:cc5}
    \rho_s(t+ \dd t ) = \rho_s(t) + \Upsilon^2\mathcal{D}(A)\rho_s(t) \dd t,
\end{equation}
which leads to the LE
\begin{equation}\label{eq:cc6}
    \partial_t\rho_s(t) = \Upsilon^2\mathcal{D}(A)\rho_s(t).
\end{equation}

In conclusion, the only difference among the LEs given by Eqs.~\eqref{eq09},  \eqref{eq:cc1}, \eqref{eq:cc1.2}, and \eqref{eq:cc6} is the dissipation strength, giving that the protocol is entirely robust to errors in the coupling constant.

\subsection{Errors in the steering Hamiltonian}\label{ssec:steering_hamiltonian}

In an experimental situation akin to the one in which the detectors are erroneously prepared (see Sec.~\ref{ssec:errors_in_the_detector}), we may assume the existence of an environment, e.g., a heat bath, over which we do not have absolute control, and that interacts with the system and detector \emph{during} the steering and not before. As a result, the initial state $\rho_{\text{eds}}(0) = \rho_e(0)\otimes \rho_{d}\otimes \rho_s(0)$, involving the environment state $\rho_e(0)$, evolves under the Hamiltonian 
\begin{equation}\label{eq:c1} 
    H_{\text{eds}} = H_e\otimes I_{\text{ds}} + I_e\otimes H_\text{ds}  + \Tilde{H}_{e}\otimes\Tilde{H}_{\text{ds}}
\end{equation}
where $H_e$ is the Hamiltonian of the environment,  $H_\text{ds}$ is the ideal steering Hamiltonian, and 
$\Tilde{H}_{e}\otimes \Tilde{H}_{\text{ds}}$ is the Hamiltonian dictating the interaction of the environment with the detector and steered system. 

Adapting the steps outlined in Sec. \ref{sec:ideal_protocol_and_setup}, we first allow the total system---now being the environment-detector-system---to evolve during a time $\delta t$ and then take the partial trace over the environmental degrees of freedom. This gives the reduced detector-system density matrix
\begin{equation}\label{eq:noi1}
    \rho_{\text{ds}}(\delta t ) = \Tr_e[\mathcal{U}(\delta t) \rho_{\text{eds}}(0)\mathcal{U}^\dagger (\delta t)],
\end{equation}
where $\mathcal{U}(\delta t) = \exp(-i\delta tH_{\text{eds}})$. Immediately after, the observable $S_d = \ket*{\Phi_d}\bra{\Phi_d} - \ket*{\Phi_d^\perp}\bra*{\Phi_d^\perp}$ is measured on the detector, leaving us with the reduced steered state [cf. Eq.~\eqref{eq:04}],
\begin{equation}\label{eq:noi2}
    \rho_{s,\alpha}(\delta t) = \frac{\bra{\alpha}\rho_{ds}(\delta t)\ket{\alpha}}{P(\alpha)},
\end{equation}
where $\alpha \in \{0,1\}$ accounts for the measurement outcome with probability $P(\alpha)$, and $\ket{0} = \ket{\Phi_d}$, $\ket{1} = \ket*{\Phi_d^\perp}$.

The reduced state, Eq.~\eqref{eq:noi2}, can be thought of as the result of applying the measurement (Kraus) operator
\begin{equation}\label{eq:noi3}
    K_\alpha(\delta t) \coloneqq \bra{\alpha}\sum_{i,j} \sqrt{p_j}\bra{\psi_i}\mathcal{U}(\delta t)\ket{\psi_j}\ket{0}
\end{equation}
to the initial state $\rho_s(0)$, 
where the initial state of the environment is written as $\rho_e(0) = \sum_i p_i\ket{\psi_i}\bra{\psi_i}$. 
An integral can replace this sum if the spectrum of the environment is continuous. In comparison with the previous scenarios, Eq.~\eqref{eq:noi3} illustrates how the dynamics of the steered system become more complex and perhaps intractable, even if all the terms in the Hamiltonian \eqref{eq:c1} terms are known. 

As the focus of our study is a sequence of generalized measurements performed on the system to steer it toward a particular target state, in what follows, we seek to translate the microscopic theory Eq.~\eqref{eq:c1} in the total Hilbert space $\mathcal{H}_e\otimes \mathcal{H}_d\otimes \mathcal{H}_s$ into a phenomenological model in $\mathcal{H}_d\otimes \mathcal{H}_s$, to make the problem tractable. More precisely, we will capture the influence of the environment on the detector-system in a stochastic operator perturbing the ideal steering Hamiltonian. However, this replacement will require a set of  assumptions and conditions on the environment and its interaction with the detector-system.

\subsubsection{Noise representation}

We start by rescaling some of the terms in Eq.~\eqref{eq:c1},
\begin{equation}\label{eq:noi4}
    H_{\text{eds}} = \lambda^{-2}H_e\otimes I_{\text{ds}} + I_e\otimes H_\text{ds} + \lambda^{-1}\Tilde{H}_{e}\otimes \Tilde{H}_{\text{ds}},
\end{equation}
such that we are interested in the limit $\lambda \rightarrow 0$. This is known as the \emph{singular-coupling limit} \cite{breuer2002theory,rivas2012open}, and it
implies that the characteristic relaxation time of the environment tends to zero, which, in turn, guarantees the elimination of any memory effects linked to the environment. According to the vanishing memory effect, the correlation function of the environment is proportional to a delta function
\cite{rivas2012open}, namely,
\begin{equation}\label{eq:noi5}
    C(t-s)=\!\int_\mathbb{R}\!e^{-i\omega(t-s)/\lambda^2}\!\Tr[\Tilde{H}_e(\omega) \Tilde{H}_e \rho_e(0)]\frac{\dd s}{\lambda^2} \propto \delta(t-s),
\end{equation}
where $\Tilde{H}_e(\omega)$ is an eigenoperator of $\ad(H_e)$ with eigenvalue $-\omega$. Another requirement is that the integral of the correlation function with $s = 0$ over all times is equal to a positive constant, i.e., 
\begin{equation}\label{eq:noi6}
    \eta = \int_\mathbb{R}\, C(t)\,\dd t.
\end{equation}

Given the above assumptions, we can replace the reduced dynamics of the detector-system, $$\rho_{\text{ds}}(t) = \Tr_e\rho_{\text{eds}}(t),$$ 
by the average over realizations of the unitary-driven density matrix
\begin{equation}\label{eq:noi7}
    \chi_{\text{ds}}(t \vert \xi) = U(t\vert \xi) \rho_d\otimes \rho_s(0) U^\dagger(t\vert \xi)
\end{equation}
given a realization of the stochastic noise $\xi$ (see below).
Here, 
\begin{equation}\label{eq:c4}
    U(t\vert \xi) \coloneqq \overrightarrow{\mathcal{T}}\exp(-i\int_0^{t}  H(s\vert \xi)\dd s)
\end{equation}
is the time-ordered (as represented by the symbol $\overrightarrow{\mathcal{T}}$), unitary time-evolution operator generated by the Hamiltonian
\begin{equation}\label{eq:c5}
    H(t\vert \xi) \coloneqq H_\text{ds} + \sqrt{\eta} \Upsilon \xi(t)\, \Tilde{h}_{\text{ds}}
\end{equation}
acting only in the detector-system space. Here, the effective time-dependent coupling $\Upsilon\xi(t)$ accounts for the environmentally induced noise, $\Tilde{h}_{\text{ds}}$ is the effective dimensionless Hamiltonian describing the noisy detector-system interaction, and $\eta$ is defined by Eq.~\eqref{eq:noi6}. The normalization conditions of the white noise are the same as in Eq.~\eqref{eq:cc2}.

To summarize, the properties of the environment described by the microscopic model given by Eq.~\eqref{eq:noi4} and $\rho_e(0)$ are encoded in the perturbation term of the phenomenological model Eq.~\eqref{eq:c5}, under the assumption that the environment-detector-system coupling in the former equation is singular and that the coupling constant in the latter equation is given by Eq.~\eqref{eq:noi6}. Under these conditions, it is then guaranteed that 
\begin{equation}\label{eq:noi8}
    \rho_{\text{ds}}(t) = \Tr_e \rho_{\text{ds}}(t) = \mathbb{E}[\zeta_{\text{ds}}(t\vert \xi)],
\end{equation}
where $\mathbb{E}$ denotes the average over all ``classical'' noise trajectories $\mathcal{\xi}$, and $\zeta_{\text{ds}}(t\vert \xi)$ is the detector-system density matrix for a given noise realization.

The replacement of the exact, reduced dynamics of an open system by the realization-average of a stochastic, unitary evolution represented by Eq.~\eqref{eq:noi7} is called the \emph{noise representation}, and justifies the widely used Hamiltonians of the form given by Eq.~\eqref{eq:c5} (see, e.g., Refs.~\cite{muller2016stochastic,kumar2020quantum,nakazato1999two,salgado2002lindbladian,onorati2017mixing,fischer1998averaged,avron_lindbladians_2015}). More conditions on the applicability of the noise representation are studied in Refs.  \cite{szankowski2020noise,szankowski2021measuring} and the references therein. 

In what follows, we start with the blind measurement scheme and find the LE governing the averaged dynamics of $\rho_s(t)$. We call this procedure a \emph{direct averaging}. Later, we opt for a different strategy to derive an SME of the jump-diffusive type that unravels the same LE we obtain in the direct averaging. We finally finish by studying three particular forms of $\tilde{h}_{\text{ds}}$ in Eq.~\eqref{eq:c5}.

\subsubsection{Direct averaging}

Let us start by taking the expectation value as in  Eq.~\eqref{eq:noi8}, but in the interaction picture with $H_\text{ds}=Jh_0$ as the free Hamiltonian, and then return to the Schrödinger picture. This gives
\begin{multline}\label{eq:c9}
\rho_{\text{ds}}(t) = \exp[-itJ\,\ad(h_0)]\\ \times \mathbb{E}\left[   \overrightarrow{\mathcal{T}}\exp(-i\int_{0}^t  \Upsilon\xi(s)\sqrt{\eta}\, \ad\textbf{(}\hat{\tilde{h}}(s)\textbf{)} \, \dd s)\right]\, \rho_d\otimes \rho_s(0),
\end{multline}
where $\hat{\tilde{h}}(s) \coloneqq \exp[is\, J\ad(h_0)]H$ is any Hamiltonian term $H$ in the interaction picture.  After noting that the time-ordering operator commutes with the expectation $\mathbb{E}$ (this is proven in Appendix~\ref{sec:commutation_expectation}), the above equation can be cast in the form 
$$\rho_{\text{ds}}(t) = \mathcal{E}_t\rho_{\text{ds}}(0)$$ 
with the dynamical map
\begin{equation}\label{eq:c10}
    \mathcal{E}_t = \exp[-itJ\ad(h_0)]\, \overrightarrow{\mathcal{T}}\exp(-\frac{\Upsilon^2 \eta}{2}\int_0^t  \ad^2\textbf{(}\hat{\tilde{h}}(s)\textbf{)}\dd s).
\end{equation}
Here, the $n$-th power of the adjoint action of an operator $A$ over $B$ is recursively defined as 
$$\ad^n(A)B = [\ad^{n-1}(A)]\ad (A)B = \ad^{n-1}(A)[A,B].$$ 
Thus, $\ad^2(A)B = [A,[A,B]].$

Taking the time derivative of $\rho_{\text{ds}}(t) = \mathcal{E}_t\rho_{\text{ds}}(0)$ gives 
\begin{equation}\label{eq:c11}
\partial_t\rho_{\text{ds}}(t) = -iJ[h_0, \rho_{\text{ds}}(t)] - \frac{\Upsilon^2\eta}{2}[\tilde{h}_{\text{ds}}, [\tilde{h}_{\text{ds}},\rho_{\text{ds}}(t)]].
\end{equation}
This LE contains two channels: a unitary channel describing the detector-system interaction \emph{without} performing any partial trace (or detector readout); and a dissipator originating from the interaction between the detector-system with the environment. Now, the formal solution of the above LE is  $\rho_{\text{ds}}(t) = \exp(\mathcal{L}t)\rho_{\text{ds}}(0)$ with 
$$\mathcal{L} = -iJ\ad (h_0) - \frac{\Upsilon^2\eta}{2}\ad^2 (\tilde{h}_{\text{ds}}).$$ 
By comparing this solution with Eq.~\eqref{eq:c10} we obtain the interesting identity
\begin{widetext}
\begin{equation}\label{eq:identity}
    \exp[-itJ\,\ad(h_0)]\overrightarrow{\mathcal{T}}\exp[-\frac{\Upsilon^2\eta}{2}\int_0^t \exp[isJ\,\ad (h_0)]\ad^2(\tilde{h}_{\text{ds}}) \exp[-isJ\,\ad (h_0)]  \dd s] = \exp(-iJt\,\ad h_0 - \frac{\Upsilon^2 \eta }{2}  t\, \ad^2(\tilde{h}_{\text{ds}})),
\end{equation}
\end{widetext}
and can find the LE obeyed by $\rho_s(t) = \Tr_d\rho_{\text{ds}}(t)$. To this end, we start from 
\begin{equation}\label{eq:c12}
    \rho_{s}(t + \delta t) = \Tr_d\left[\exp(\mathcal{L}\delta t)\, \rho_d\otimes \rho_{s}(t)\right],
\end{equation}
and use the WM limit with the usual representations of $\rho_{d}$ and $h_{0}$ with respect to the ONBs $\mathcal{B}_\oplus = \{\ket{\Psi_\oplus}, \ket*{\Psi_\oplus^\perp} \}$ and $\mathcal{B}_d = \{\ket{\Phi_d}, \ket*{\Phi_d^\perp} \}$ to get (see Appendix~\ref{sec:LE_errors_hamitlonian}).
\begin{equation}\label{eq:c13}
\partial_t\rho_s(t) = \left[ \gamma \mathcal{D}(A) + \tilde{\gamma} \mathcal{D}(G) + \tilde{\gamma} \mathcal{D}(B)  \right]\rho_{s}(t),
\end{equation}
where $G = G^\dagger$, $C = C^\dagger$, and $B$ are block matrices of the Hamiltonian [cf. Eq.~\eqref{eq:c5}]
\begin{equation}\label{eq:c13.1}
    \Tilde{h}_{\text{ds}} = \begin{pmatrix}
    G & B^\dagger \\ B & C
    \end{pmatrix},
\end{equation}
and
\begin{equation}
\tilde{\gamma} = \eta \Upsilon^2.
\end{equation}

We note that we obtained Eq.~\eqref{eq:c13} by first taking the average over realizations of the white noise $\xi(t)$ from $t$ to $t+\delta t$ and then performing a blind measurement. Alternatively, it can be shown (see Appendix~\ref{sec:lastproof}) that one can perform a blind measurement at $t+\delta t$ given the \emph{same} stochastic trajectory and \emph{then} take the average over realizations of the stochastic variable. 

We provide a few relevant comments about the LE~\eqref{eq:c13}. First, we note that this LE has no hybrid channels generated by $G$, $B$ or $A$; that is, it has no cross-terms of the form $G\rho_s(t)B^\dagger- \{B^\dagger G,\rho_s(t) \}/2$, and so on. Second, there is no channel associated with $C$.

In addition, Eq.~\eqref{eq:c13} explicitly shows that a local 
environment-detector interaction, effectively described by $\Tilde{h}_{\text{ds}} = O\otimes I_s$ in Eq.~\eqref{eq:c5}, has no effect on the dynamics of $\rho_s(t)$. This is easily seen by noting that for this interaction  $G, B \propto I_s$ in Eq.~\eqref{eq:c13.1}, and thus $\mathcal{D}(I_s) = 0$. This means that performing a blind measurement takes care of the environmental influence over the detectors, and this perturbation is not transferred to the state of the system. Likewise, if the interaction is local in the system space, e.g.,  $\Tilde{h}_{\text{ds}} = I_d\otimes G$, the LE loses the dissipator associated with the off-diagonal blocks in Eq.~\eqref{eq:c13.1}: $\mathcal{D}(B) = 0$.

\subsubsection{Jump-diffusive stochastic master equation}\label{ssec:jump_diffusive_stochastic}

An SME describing the possible individual quantum trajectories obeyed by the steered state under the influence of the white noise and the continuous monitoring of the detectors reads as (see Appendix~ \ref{sec:jump-diffussive-sme})
\begin{widetext}
\begin{multline}\label{eq:esto1}
    \dd \omega_s(t) = -i\sqrt{\tilde{\gamma}}[ G, \omega_s(t)]\dd X(t)+  \tilde{\gamma}\mathcal{D}(G)\omega_s(t)\dd t + \left( \expval*{\gamma A^\dagger A + \tilde{\gamma}  B^\dagger B}_t\omega_s(t) - \frac{1}{2}\{\gamma A^\dagger A + \tilde{\gamma} B^\dagger B, \omega_s(t)\}\right)\dd t\\ 
    + \left( \frac{\gamma A\omega_s(t) A^\dagger + \tilde{\gamma}  B\omega_s(t) B^\dagger}{\expval*{\gamma A^\dagger A + \tilde{\gamma}  B^\dagger B}_t}- \omega_s(t) \right)\dd N(t),
\end{multline}
\end{widetext}
where $\dd X(t)$ is a Wiener increment with zero mean and variance $\dd t$, $\dd N(t)$ is a Poissonian increment with strength 
\begin{equation}\label{eq:esto1.1}
    \mathbb{E}[\dd N(t)] = \expval*{\gamma A^\dagger A + \tilde{\gamma}  B^\dagger B }_t\dd t.
\end{equation}
The overall It\^o table is
\begin{align}
    (\dd X(t))^2 &= \dd t, \\
    (\dd N(t))^2 &= \dd N(t), \\
    \dd N(t)\dd t &= \dd X(t) \dd t = \dd X(t)\dd N(t) = 0, \\
    (\dd t)^2 &=0. \label{Ito-table}
\end{align}

The SME~\eqref{eq:esto1} is quite peculiar. When a click is registered [$\dd N(t) = 1]$, the steered state is found in a mixed state unless $\gamma = \tilde{\gamma}$ and $B = A$, or $B =0$. On the other hand, if no click is registered [$\dd N(t) = 0$], the system evolution contains deterministic and fluctuating contributions. The fluctuating term is given by a unitary channel generated by $G$. In contrast, the deterministic terms are due to an It\^o correction of the unitary fluctuating generator in the form of a dissipator, the deterministic backaction induced by the other environment-induced terms (i.e., the terms proportional to $B$), and the backaction due to the detectors. Contrary to the usual diffusive \cite{breuer2002theory,wiseman2009quantum,barchielli1991measurements,attal2010stochastic} and hybrid jump-diffusive SMEs \cite{kuramochi_simultaneous_2013,dong_quantum_2019,pellegrini_markov_2010,altintan_hybrid_2020}, see also Eq.~\eqref{eq:pb6}, the Wiener increment in Eq.~\eqref{eq:esto1} multiplies a unitary generator instead of a non-unitary one. Finally, by taking the expectation value of Eq.~\eqref{eq:esto1}, using the rules specified by Eqs.~\eqref{eq:esto1.1}-\eqref{Ito-table}, we arrive at the LE~\eqref{eq:c13}.

\subsubsection{Examples}\label{ssec:examples}
Let us illustrate Eqs.~\eqref{eq:c13} and \eqref{eq:esto1} with two particular perturbation Hamiltonians [see Eq.~\eqref{eq:c5}] of the form $\Tilde{h}_{\text{ds}} = I_d\otimes G$ with $G = \sigma^z$ and $G = \sigma^x$. For both cases, the detectors are prepared in the state $\rho_d = \ket{\uparrow}\bra{\uparrow}$, and  the ideal target state is given by $\rho_\oplus = \ket{\uparrow}\bra{\uparrow}$.

\begin{figure*} 
\centering
  \begin{subfigure}{0.22\textwidth}
    \includegraphics[width=\linewidth]{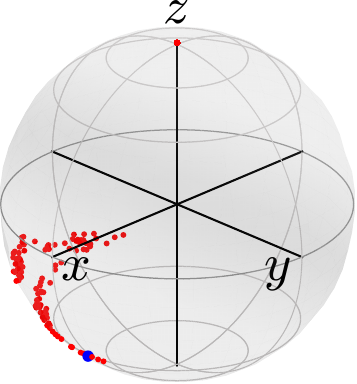}
    \caption{}\label{fig11:a}
  \end{subfigure}%
   \hspace{1em}
  \begin{subfigure}{0.22\textwidth}
    \includegraphics[width=\linewidth]{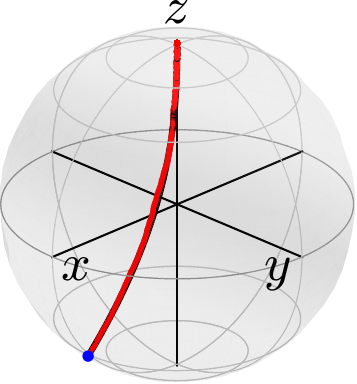}
    \caption{}\label{fig11:b}
  \end{subfigure}%
    \hspace{1em}
    \begin{subfigure}{0.22\textwidth}
    \includegraphics[width=\linewidth]{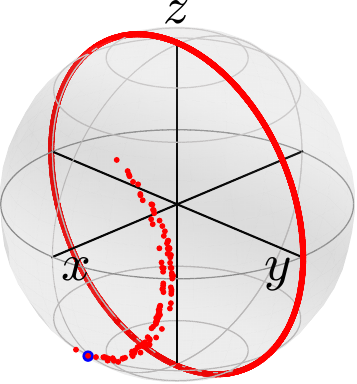}
    \caption{}\label{fig11:c}
  \end{subfigure}%
   \hspace{1em}
  \begin{subfigure}{0.22\textwidth}
    \includegraphics[width=\linewidth]{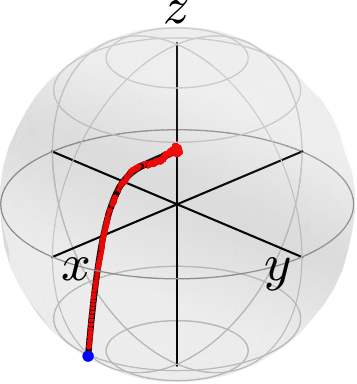}
    \caption{}\label{fig11:d}
  \end{subfigure}%
  \vskip\baselineskip 
   \begin{subfigure}{0.22\textwidth}
    \includegraphics[width=\linewidth]{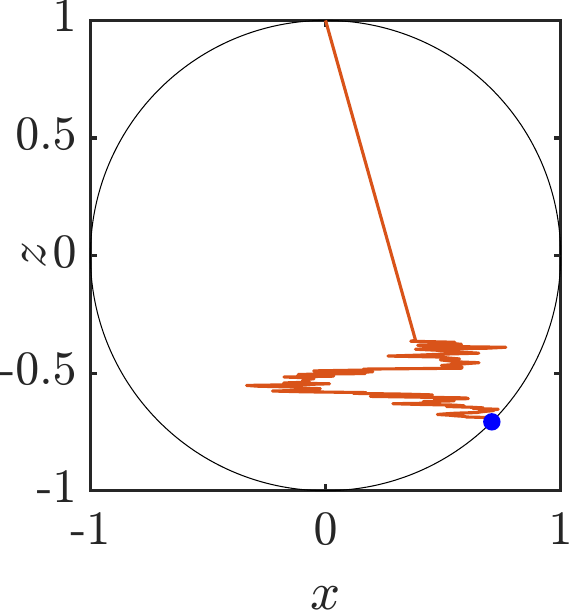}
    \caption{}\label{fig11:e}
  \end{subfigure}%
     \hspace{1em}
   \begin{subfigure}{0.22\textwidth}
    \includegraphics[width=\linewidth]{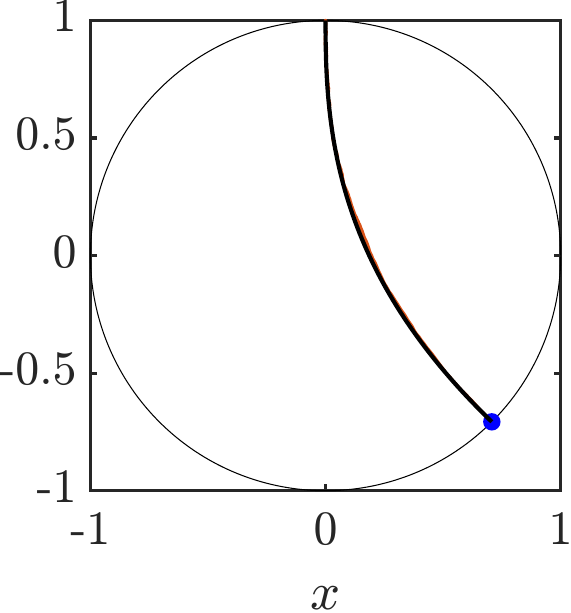}
    \caption{}\label{fig11:f}
  \end{subfigure}%
     \hspace{1em}
      \begin{subfigure}{0.22\textwidth}
    \includegraphics[width=\linewidth]{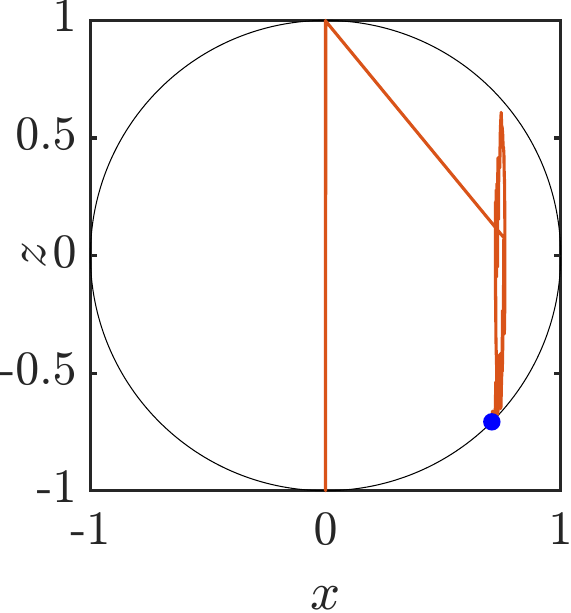}
    \caption{}\label{fig11:g}
  \end{subfigure}%
     \hspace{1em}
   \begin{subfigure}{0.22\textwidth}
    \includegraphics[width=\linewidth]{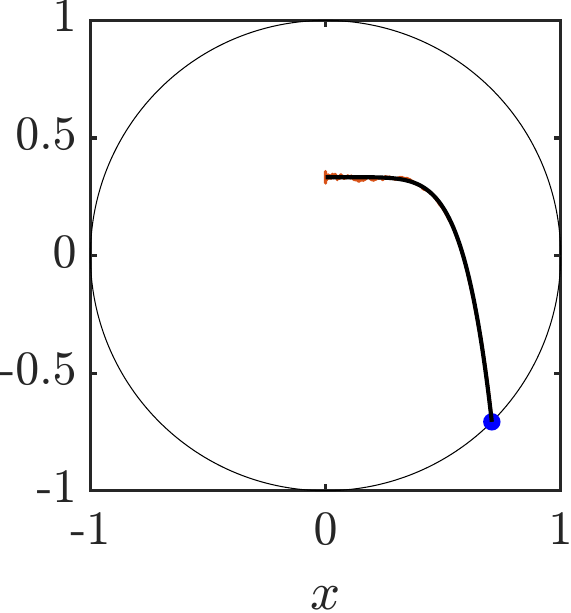}
    \caption{}\label{fig11:h}
  \end{subfigure}%
  \caption{Steering of a single qubit toward the north pole of the Bloch sphere, i.e., $\rho_\oplus = \ket{\uparrow}\bra{\uparrow}$, with the perturbed detector-system Hamiltonian $H_\text{ds}(t) = \sqrt{\gamma/\delta t}(\sigma^+\otimes \sigma^- + \hc) +  I_d\otimes \sqrt{\tilde{\gamma}}\xi(t)G,$ where $G = \sigma^z$ in panels (a)-(b) and (e)-(f), and $G = \sigma^x $ in panels (c)-(d) and (g)-(h). The decay rates are $\gamma = \tilde{\gamma} =  0.1$ and the initial state is $\bm{r}(0) = (1,0,-1)/\sqrt{2}$, which is represented by a blue dot. For $G = \sigma^z$ and $G = \sigma^x$, respectively, (a) and (c) display particular quantum trajectories that are solutions of Eq.~\eqref{eq:esto1}. In panels (a) and (e), a combination of the fluctuating rotation around the $z$-axis, with the continuous backaction steering induced by the detectors, can be observed until the state jumps to the north pole, where the evolution stops. Panels (b) and (f) show an average of $10^3$ (in red) trajectories---similar to the one in panel (a)---and the analytical solution to the LE~\eqref{eq:LEsz} (in black). The ideal target state is still the stationary point of the dissipative dynamics, implying full robustness to this type of perturbation. 
  In panels (c) and (g), we observe the interplay of the continuous backaction and the rotation around the $x$-axis until a jump to the north pole occurs. Hereafter, the Bloch vector stays locked on the great circle of the $y$-$z$--plane. As a result, the average over $10^4$ trajectories displayed in panels (d) and (h) (in red) show that the actual target state [Eq.~\eqref{eq:esto2}] is mixed and does not coincide with the ideal one. As the strength of the fluctuating rotation around the $x$-axis increases, the stationary state becomes more mixed. The analytical solution to the corresponding LE [Eq.~\eqref{eq:LEsx}] is shown in black.}
  \label{fig:11}
  \end{figure*}
  
In Fig.~\ref{fig11:a}, we show a particular solution of Eq.~\eqref{eq:esto1} with  $G = \sigma^z$, where we can see the fluctuating rotation of the Bloch vector $\bm{r}(t) = \Tr[\omega_s(t)\bm{\sigma}]$ around the $z$-axis while it gets inevitably closer to the target state, owing to the backaction of the detectors. In that particular trajectory, the system jumps to the north pole, where it stops evolving. A straight line represents the jump in the transverse cut of the Bloch sphere shown in Fig.~\ref{fig11:e}. 

Averaging over $10^4$ trajectories gives rise to the purely dissipative dynamics shown in Figs.~\ref{fig11:b} and \ref{fig11:f} in red, which is an approximate solution of the LE [cf. Eq.~\eqref{eq:c13}]
\begin{equation}\label{eq:LEsz}
\partial_t\rho_s(t) = \left[\gamma \mathcal{D}(\sigma^+) + \tilde{\gamma} \mathcal{D}(\sigma^z) \right]\rho_s(t).
\end{equation}
The solution to this equation is
\begin{align}
    [\rho_s(t)]_{11} &= 1 + \lbrace{ [\rho_s(0)]_{11} - 1 \rbrace}\exp(-\gamma t), \\
    [\rho_s(t)]_{12} &= [\rho_s(0)]_{12}\exp(-\frac{\gamma + 4\Tilde{\gamma}}{2}t),
\end{align}
and is shown in Figs.~\ref{fig11:b} and \ref{fig11:f} in black. In those figures, we can see that the numerically averaged trajectory agrees to a great extent with the analytical solution.  Moreover, we note that the only change compared with the ideal solution Eq.~\eqref{eq:15.127}-\eqref{eq:15.129} (in the Bloch representation) is in the rate at which the $x$ and $y$-components decay, which is faster in the case here treated. 

Although the averaged trajectories change compared to the ideal dissipative dynamics (cf. Fig.~\ref{fig40:e}), the ideal target state is invariant. This is a trivial consequence of $\ket{\Psi_\oplus} = \ket{\uparrow}$ being an eigenstate of $G = \sigma^z$. We provide further insight in Fig.~\ref{fig:356} by illustrating the time-dependence of the relevant quantities of the erroneous steering with $G = \sigma^z$. There, we can see again [cf. Fig.~\ref{fig:301}] the similarities between the impurities and trace distance of the ideal and averaged (both numerically and analytical) dynamics. Due to the presence of the dissipator $\mathcal{D}(\sigma^z)$, which dampens the coherences of the state even faster than in the ideal case, the steered state becomes maximally mixed as it traverses the $z$-axis of the Bloch sphere. While doing so, it approaches the ideal state faster than in the ideal steering.  

Turning to the case where $G = \sigma^x$, the ideal target state is no longer an eigenstate of this operator. As a consequence, the actual final state of the averaged dynamics is 
\begin{equation}
    [\tilde\rho_\oplus]_{11} = 1 - \frac{\tilde{\gamma}}{\gamma + 2\tilde{\gamma} }, \quad [\tilde\rho_\oplus]_{12} = 0.  \label{eq:esto2}
\end{equation}
This state is the result of the competition between the two dissipative channels in the LE 
\begin{equation}\label{eq:LEsx}
    \partial_t\rho_s(t) = \left[ \gamma\mathcal{D}(\sigma^+) + \tilde{\gamma} \mathcal{D}(\sigma^x) \right]\rho_s(t),
\end{equation}
where the first dissipator tries to collapse the entire Bloch ball toward the north pole and the second one toward the $x$-axis. 

A single trajectory of the corresponding SME and its average is shown in Figs.~\ref{fig11:c} and \ref{fig11:d}, respectively. The depicted single trajectory results from the combined effect of the fluctuating rotation around the $x$-axis and the continuous, non-unitary evolution directing the Bloch vector toward the north pole. Once the north pole is reached by a jump---or via continuous evolution for other trajectories---the Bloch vector stays locked on the great circle in the $y$-$z$--plane. Since this occurs for every trajectory, once the average is performed, the stationary state is no longer pure and lies on the $z$-axis below the north pole. Figure~\ref{fig:357} shows the time progression of the Bloch components depicted in Figs.~\ref{fig11:c}-\ref{fig11:d}, and in their respective transverse cuts Figs.~\ref{fig11:g}-\ref{fig11:h}, in addition to the usual quantifiers.

The stationary state quantifiers comparing the actual target state and the ideal target state are to leading order 
\begin{align}
    F_\infty &= \frac{\gamma + \tilde{\gamma}}{\gamma + 2\tilde{\gamma}} \approx 1  - \frac{\tilde{\gamma}}{\gamma}, \label{eq:esto4}\\
    D_{1,\infty} &=  \frac{\tilde{\gamma}}{\gamma + 2\tilde{\gamma}} \approx \frac{\tilde{\gamma}}{\gamma}, \label{eq:esto5}\\
    L_\infty &= \frac{2\tilde{\gamma}(\gamma + \tilde{\gamma})}{(\gamma + 2\tilde{\gamma})^2} \approx \frac{2\tilde{\gamma}}{\gamma}, \label{eq:esto6}
\end{align}
where we have performed a series expansion in small $\Tilde{\gamma}$ with fixed $\gamma$. We observe here that the dissipator strength $\Tilde{\gamma}$ affects the steering to first order, which means that this error significantly affects the steering protocol.

 \begin{figure*}[t]
\centering
 \begin{subfigure}{0.805\textwidth}
\includegraphics[width=\linewidth]{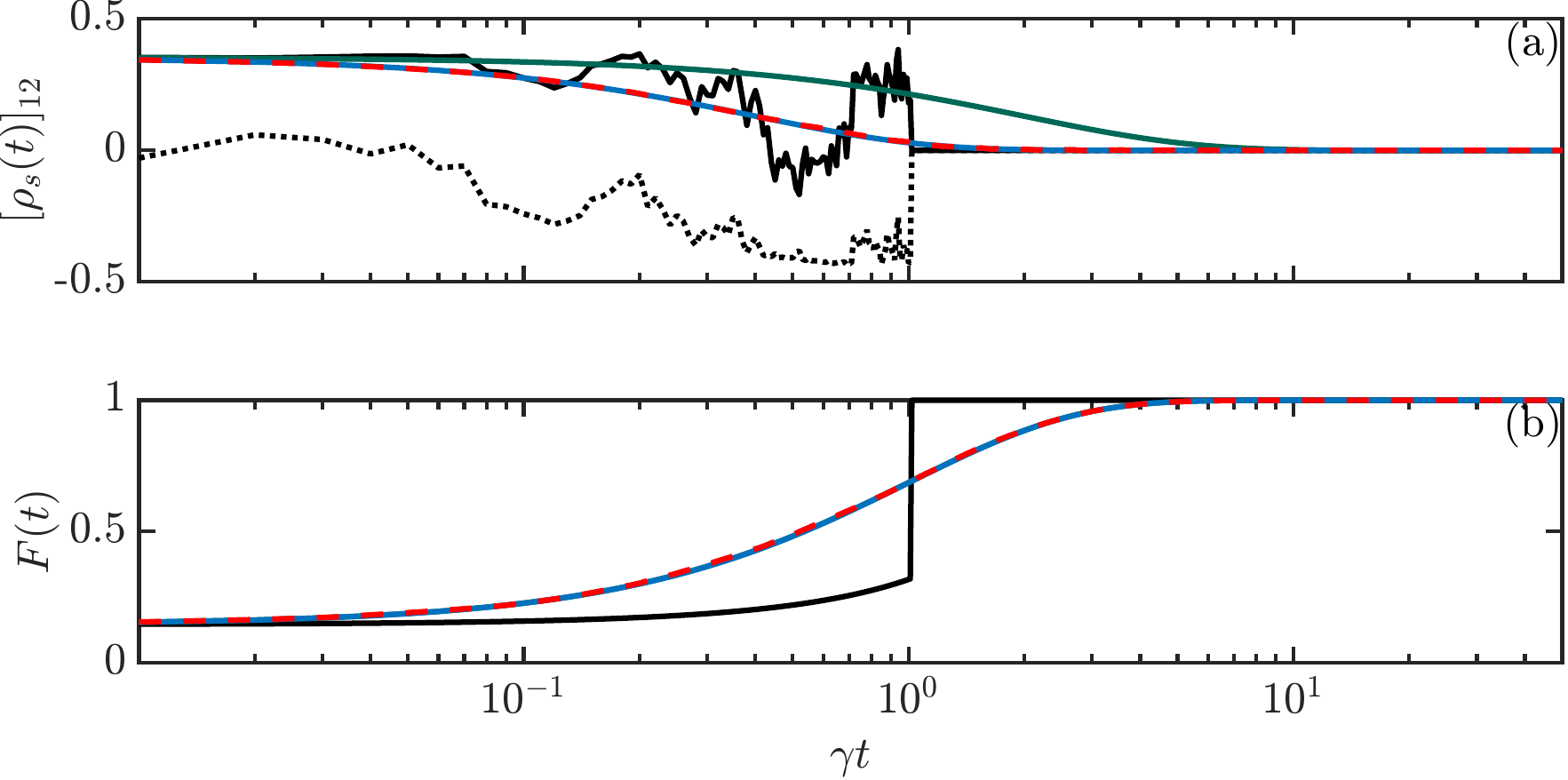}
  \end{subfigure}%
  \hspace{1em}
   \vskip\baselineskip
   \begin{subfigure}{0.8\textwidth}
    \includegraphics[width=\linewidth]{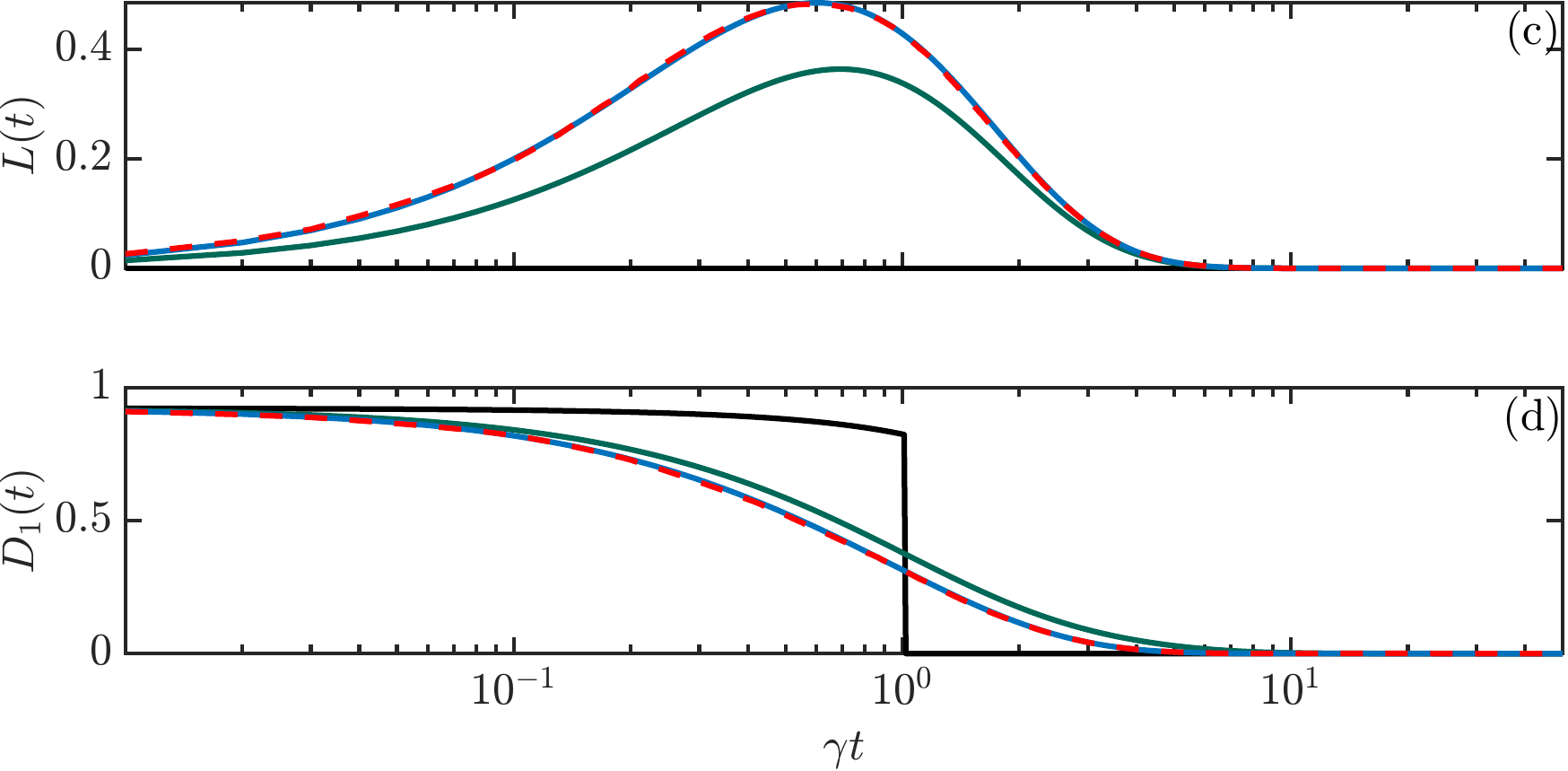}
  \end{subfigure}%
   \hspace{1em}
  \caption{Entries of density matrices and distance measures as functions of time, associated with the steering of a single qubit to the ideal state $\rho_\oplus = \ket{\uparrow}\bra{\uparrow}$ with the perturbed detector-system Hamiltonian $H_\text{ds}(t) = \sqrt{\gamma/\delta t}(\sigma^+\otimes \sigma^-+\hc)+I_d\otimes \sqrt{\Tilde{\gamma}}\xi(t)\sigma^z$. (a) Coherences of several states are generically labeled as $[\rho_s(t)]_{12}$ even if they refer to different density matrices. (b) Fidelity [Eq.~\eqref{eq:me1}]. (c) Impurity [Eq.~\eqref{eq:2.7}]. (d) Trace distance [Eq.~\eqref{eq:me2}]. Note that the impurity of the quantum trajectory is always zero as Eq.~\eqref{eq:esto1} respects purity when $B =0$. Black curves correspond to a single quantum trajectory solving Eq.~\eqref{eq:esto1} with $G = \sigma^z$. In (a), the solid and dotted black lines, respectively, refer to the real and the negative imaginary part of the coherences of the quantum trajectory. Green curves correspond to the ideal steered state solving LE~\eqref{eq09}. Blue and red curves are the exact and approximate solutions to the fully averaged LE~\eqref{eq:c13} (where $B=0$), respectively. The average was taken over $10^4$ trajectories.  For all the plots, $\gamma = \Tilde{\gamma} = 0.1$, $\delta t = 0.1$, and the initial Bloch vector $\bm{r}(0) = (1,0,-1)/\sqrt{2}$.
  } 
  \label{fig:356}
\end{figure*}

 \subsection{Errors in the measurement direction (direction of the detector projection)}\label{ssec:errros_in_the_measurement_direction}

We close the study of the dynamic errors by considering the situation in which the basis of the local observable measured on the detectors changes at each steering step with a given probability. We know \emph{a priori} that the blind measurement scheme will not be affected by this occurrence, for the partial trace does not depend on the chosen basis. Nevertheless, from the individual quantum trajectories, the situation described above can lead to a set of different SMEs unraveling the error-free LE containing a single dissipation channel whose stationary state is the ideal target state. 

Let $\{\mathcal{B}_{d,i}\}_{i \in \mathcal{I}}$ be a family of ONBs spanning $\mathcal{H}_d$ with $\mathcal{I} = \{1,\ldots, n\}$. For convenience, we set $$\mathcal{B}_{d,1} \coloneqq \{\ket{0} = \ket{\uparrow}, \, \ket{1} = \ket{\downarrow} \}$$ as the canonical basis, and for $i \neq 1$, $$\mathcal{B}_{d,i} = \{\ket*{\psi_0^{(i)}},\ket*{\psi_1^{(i)}} \} \neq \mathcal{B}_{d,1}.$$ After the detector and system interact, there is a probability $p(i)$ of measuring the detector in the ONB $\mathcal{B}_{d,i}$. Regardless of the measurement basis, we will assume that the remaining steering parameters coincide with those of the ideal protocol setting (see Sec.~\ref{sec:ideal_protocol_and_setup}). 

The general discrete SME describing the repeated interaction of the steered system with a set of detectors measured in random ONBs is
\begin{equation}\label{eq:pb2}
    \omega_{k+1} = \sum_{i \in \mathcal{I}}\sum_{\alpha \in \{0,1\} }\frac{\mathcal{M}_\alpha^{(i)}\omega_k}{p(i,\alpha\vert \omega_k)}\mathds{1}^{k+1}_{i,\alpha},
\end{equation}
where the time $t$ has been discretized, i.e., $t_k \coloneqq k \delta t$, $k \in \mathbb{Z}^+$, and so $\omega_k \coloneqq \omega(t_k)$. Further,   $\mathds{1}^{k+1}_{i,\alpha}$ is the indicator describing the random implementation of the measurement basis $\mathcal{B}_{d,i}$ and the (also random) resulting detector state $\ket*{\psi_\alpha^{(i)}}$ with $\alpha \in \{0,1\}$. The expectation of this indicator, given the prior state $\omega_k$, reads 
\begin{equation}\label{eq:pb3}
    \mathbb{E}[\mathds{1}^{k+1}_{i,\alpha}] = p(i,\alpha\vert \omega_k) = \Tr  \mathcal{M}_\alpha^{(i)}\omega_k,
\end{equation}
which is the joint probability of obtaining the detector outcome $\alpha$ with the $i$-th measurement direction. 

  \begin{figure*}[t]
\centering
 \begin{subfigure}{1\textwidth}
    \includegraphics[width=\linewidth]{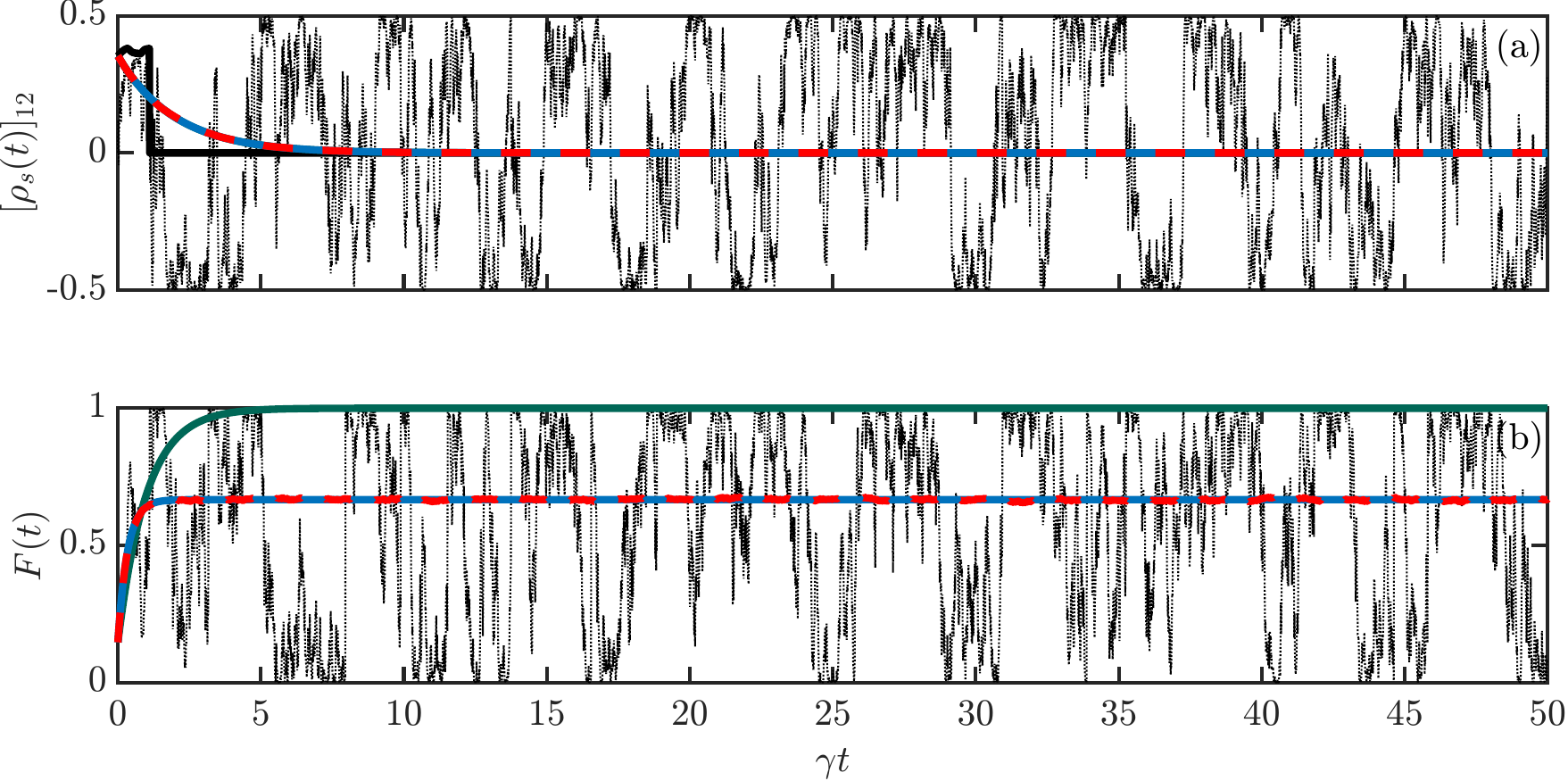}
  \end{subfigure}%
   \vskip\baselineskip
   \begin{subfigure}{1\textwidth}
    \includegraphics[width=\linewidth]{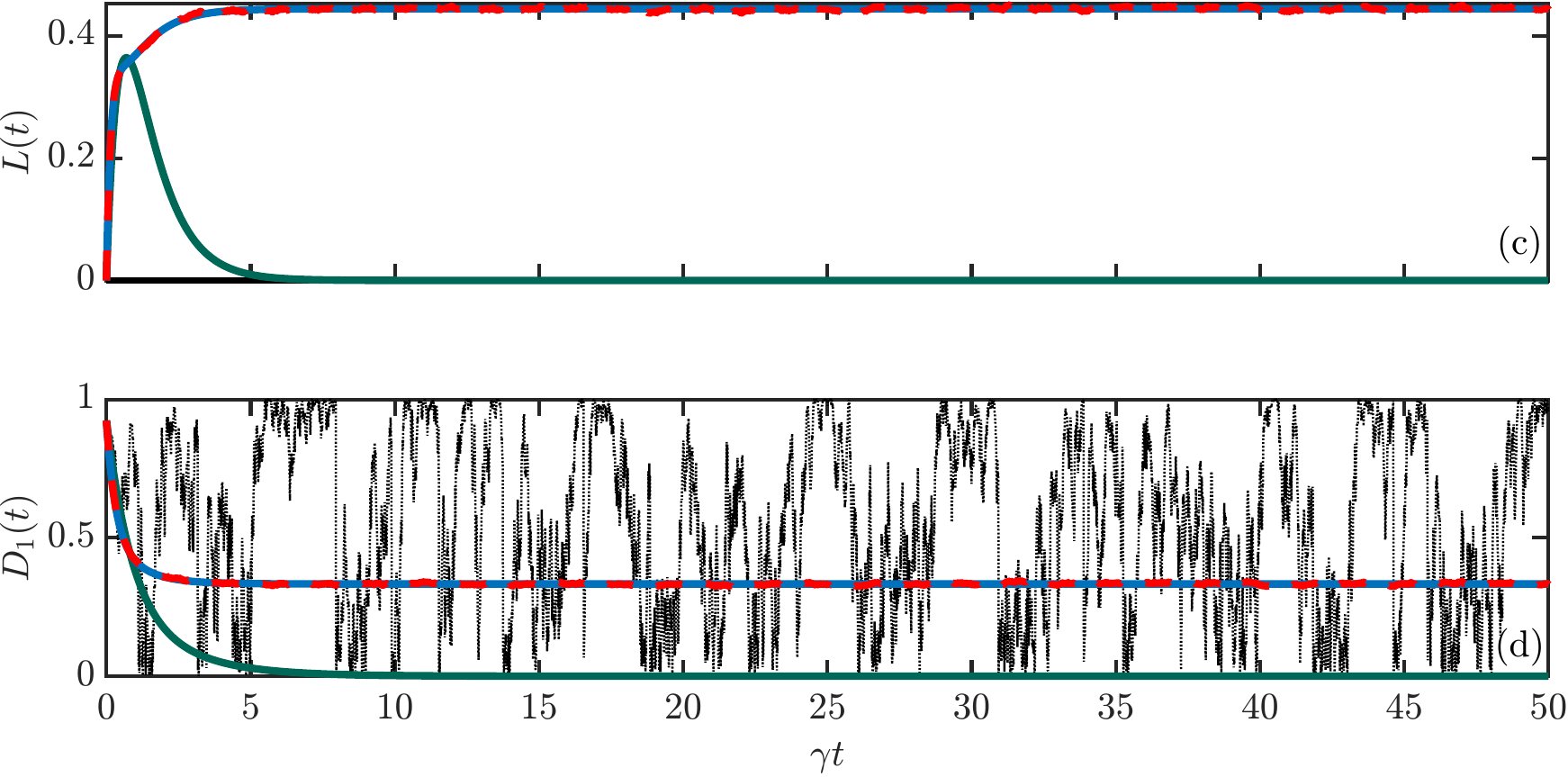}
  \end{subfigure}%
  \caption{Entries of density matrices and distance measures as functions of time, associated with the steering of a single qubit to the ideal state $\rho_\oplus = \ket{\uparrow}\bra{\uparrow}$ with the perturbed detector-system Hamiltonian $H_\text{ds}(t) = \sqrt{\gamma/\delta t}(\sigma^+\otimes \sigma^-+\hc)+I_d\otimes \sqrt{\Tilde{\gamma}}\xi(t)\sigma^x$. (a) The coherences of several states are generically labeled as $[\rho_s(t)]_{12}$ even if they correspond to different density matrices. (b) Fidelity [Eq.~\eqref{eq:me1}]. (c) Impurity [Eq.~\eqref{eq:2.7}]. Note that the impurity of the quantum trajectory is always zero as Eq.~\eqref{eq:esto1} respects purity when $B =0$. (d) Trace distance [Eq.~\eqref{eq:me2}]. Black curves correspond to a single quantum trajectory solving Eq.~\eqref{eq:esto1} with $G = \sigma^x$. In (a), the solid  and dotted black lines, respectively, correspond to the real and the negative imaginary part of the coherences of the quantum trajectory. Green curves correspond to the ideal steered state solving LE~\eqref{eq09}. Note that in (a), the green curve is not visible since it coincides with the red and blue curves. Blue and red curves are the exact and approximate solutions of the fully averaged LE~\eqref{eq:c13} (where $B=0$), respectively. The average was taken over $10^4$ trajectories. 
  For all the plots, $\gamma = \Tilde{\gamma} = 0.1$, $\delta t = 0.1$, and the initial Bloch vector $\bm{r}(0) = (1,0,-1)/\sqrt{2}$.} 
  \label{fig:357}
\end{figure*}

The unnormalized state
\begin{align}\label{eq:pb5}
\mathcal{M}_\alpha^{(i)}\omega_k &\coloneqq p(i)M_\alpha^{(i)}(\delta t)\omega_k M_\alpha^{(i)}(\delta t)^\dagger
\end{align}
describes the backaction on the prior state, $\omega_k$, once the state of the detector is reduced from $\ket{0}$ to  $\ket*{\psi_\alpha^{(i)}}$. The measurement operator is
\begin{equation}\label{eq:pb7}
M_\alpha^{(i)}(\delta t) \coloneqq \bra*{\psi_\alpha^{(i)}}\exp[-i U(\delta t)]\ket{0},    
\end{equation}
with $U(\delta t) = \sqrt{\gamma \delta t}h_0$, and $h_0 = \ket*{\Phi_d^\perp}\bra{\Phi_d}\otimes A +\hc$
If an unbiased average over \emph{all} possible results is taken, i.e., a blind measurement, the updated state becomes
\begin{subequations}
\begin{align}
    \rho_{k+1} &\coloneqq \mathbb{E}[\omega_{k+1}]\\
    &= \sum_{i,\alpha}p(i)M_\alpha^{(i)}(\delta t)\mathbb{E}[\omega_k] M_\alpha^{(i)}(\delta t)^\dagger\\
    &= \Tr_d[ U(\delta t)\rho_d\otimes \rho_k U^\dagger(\delta t)], 
\end{align}
\end{subequations}
which is the ideal, error-free, updated state. Recall that in the WM limit, the above equation leads to ideal LE~\eqref{eq09}.

\subsubsection{Example: two measurement directions}
Let us now consider a particular example in the quantum trajectories scheme. Let us suppose that there are two possible measurement directions: the ideal one denoted by  $\mathcal{B}_{d,1}$, and an erroneous one with the basis $\mathcal{B}_{d,2}$ with the respective ket vectors
\begin{equation}\label{eq:basis}
    \ket*{\psi_0^{(2)}} \coloneqq \frac{1}{\sqrt{2}}\left(\ket{0} + \ket{1} \right), \quad
    \ket*{\psi_1^{(2)}} \coloneqq \frac{1}{\sqrt{2}}\left(\ket{0} - \ket{1} \right).
\end{equation}

\begin{figure} 
\centering
  \begin{subfigure}{0.22\textwidth}
    \includegraphics[width=\linewidth]{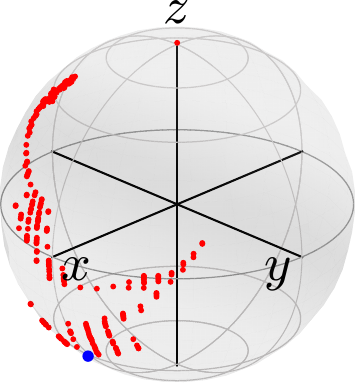}
    \caption{}\label{fig111:a}
  \end{subfigure}%
   \hspace{1em}
  \begin{subfigure}{0.22\textwidth}
    \includegraphics[width=\linewidth]{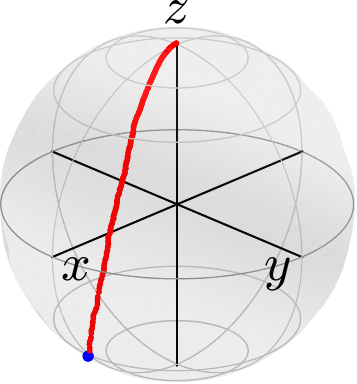}
    \caption{}\label{fig111:b}
\end{subfigure}
  \hspace{1em} 
\begin{subfigure}{0.22\textwidth}
    \includegraphics[width=\linewidth]{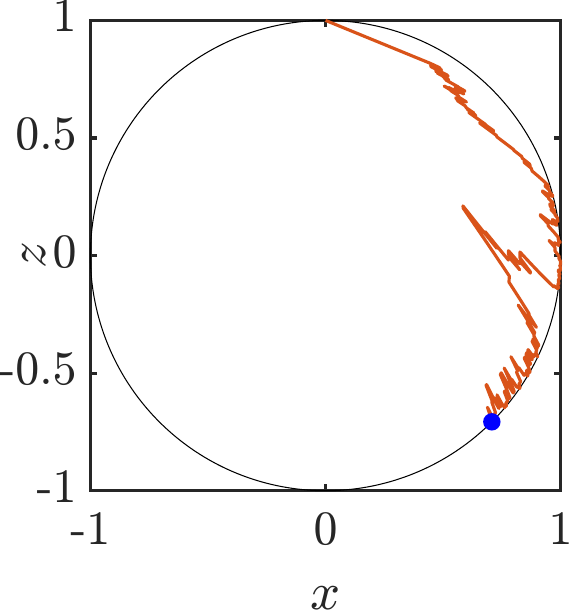}
    \caption{}\label{fig111:c}
  \end{subfigure}%
   \hspace{1em}
  \begin{subfigure}{0.22\textwidth}
    \includegraphics[width=\linewidth]{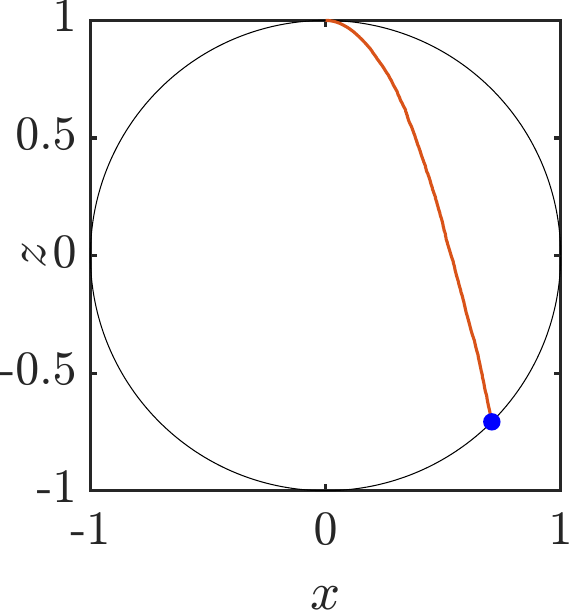}
    \caption{}\label{fig111:d}
\end{subfigure}
  \caption{Dynamics of the Bloch vector of a steered qubit toward the north pole with two measurement directions: the correct one, and an erroneous one given by Eq.~\eqref{eq:basis}. Both measurement directions are used with the same probability. A particular quantum trajectory solving Eq.~\eqref{eq:pb9} is shown in (a) and (c). An average over $10^3$ quantum trajectories of the form presented in (a) is shown in (b) and (d). The initial state is $\bm{r}(0) = (1,0,-1)/\sqrt{2}$ (represented by the blue dot). The decay is $\gamma = 0.1$, and the time-step used in the stochastic Schrödinger equation is $\delta t = 0.1$.  }
  \label{fig:111}
  \end{figure}

In the limit $\delta t = \dd t \rightarrow 0$, the discrete SME~\eqref{eq:pb2} becomes the continuous-time, hybrid SME (see Appendix~\ref{sec:jump_diffusive_type})
\begin{multline}\label{eq:pb6}
    \dd \omega_s (t) = \mathcal{D}(L)\omega_s(t) \dd t 
    \\+ \left[ \frac{L\omega_s(t)L^\dagger}{\expval{L^\dagger L}_t} - \omega_s(t) \right]\,\left[\dd N(t) - \expval*{L^\dagger L}_t\dd t \right]\chi_1(t) \\+ \left[ L\omega_s(t) + \omega_s(t)L^\dagger - \expval*{L+L^\dagger}_t\omega_s(t) \right]\dd X(t)\chi_2(t),
\end{multline}
where we have set $L\coloneqq i\sqrt{\gamma}A$ for convenience. The indicator $\chi_i(t)$ describes the random choice of the local measurement basis $\mathcal{B}_{d,i}$ appearing with probability $\mathbb{E}[\chi_i(t)] = p(i)$.  
The hybrid SME~\eqref{eq:pb6} is a combination of two standard SMEs: If $\chi_1(t) = 1$ between $t$ and $t+\dd t $, that is, if we measure in the correct direction, we have the SME of the jump-type we have derived before [cf. Eq.~\eqref{eq:07}], where the detector outcome is registered simultaneously. This detection is captured by $\dd N(t)$, where the mean value of counts up to time $t$ for a given trajectory $\omega_s(t)$ is 
\begin{equation}\label{eq:pb8}
    \mathbb{E}[ \chi_1(t)\dd N(t)] = p(1)\Tr[L^\dagger L\omega_s(t)]\dd t.
\end{equation}

On the other hand, if the measurement is performed in the wrong direction, i.e., $\chi_2(t) = 1$, we have a diffusive-type SME characterized by the Wiener increment 
$\dd W(t)$ with zero mean and variance $\dd t$. 
Note, however, that this diffusive SME does not coincide with the diffusive part of the SME~\eqref{eq:esto1} (after setting $A = B =0$). The latter equation contains a unitary, diffusive part on top of a deterministic It\^o correction arising as a simple dissipator. In contrast to this case, the diffusive-type SME~\eqref{eq:pb6} describes the continuous measurements on the detectors made in the basis $\mathcal{B}_{d,2}$, where neither a jump nor a unitary evolution occurs. 

As both the jump and diffusive part in Eq.~\eqref{eq:pb6} respect the purity of states, this equation induces the stochastic Schrödinger equation
\begin{widetext}
\begin{multline}\label{eq:pb9}
    \dd \ket{\psi_s(t)} = -\frac{1}{2}\left(L^\dagger L - \expval*{L^\dagger L}_t \right)\ket{\psi_s(t)}\chi_1(t)\dd t + \left( \frac{L}{\sqrt{\expval*{L^\dagger L}_t}} - I_s \right)\ket{\psi_s(t)}\chi_1(t)\dd N(t) \\
    + \left(- \frac{1}{2}L^\dagger L + \frac{1}{2}\expval*{L + L^\dagger}_t - \frac{1}{8}\expval*{L + L^\dagger}^2_t \right)\ket{\psi_s(t)}\chi_2(t)\dd t  + \left( L - \frac{1}{2}\expval*{L + L^\dagger}_t I_s \right)\ket{\psi_s(t)}\dd W(t)\chi_2(t).
\end{multline}
\end{widetext}

Figures~\ref{fig:111} and \ref{fig:112} show a representative quantum trajectory and the average over $10^3$ quantum trajectories obtained from Eq.~\eqref{eq:pb9} with an initial state (in the Bloch representation) $\bm{r}(0) = (1,0,-1)/\sqrt{2}$, decay $\gamma = 0.1$, time-step $\delta t = 0.1,$ and probabilities $p(1) = p(2) = 1/2$, i.e., the probability of measuring in the correct and erroneous basis is the same. \begin{figure*}[t!]
     \centering
    \includegraphics[width=\linewidth]{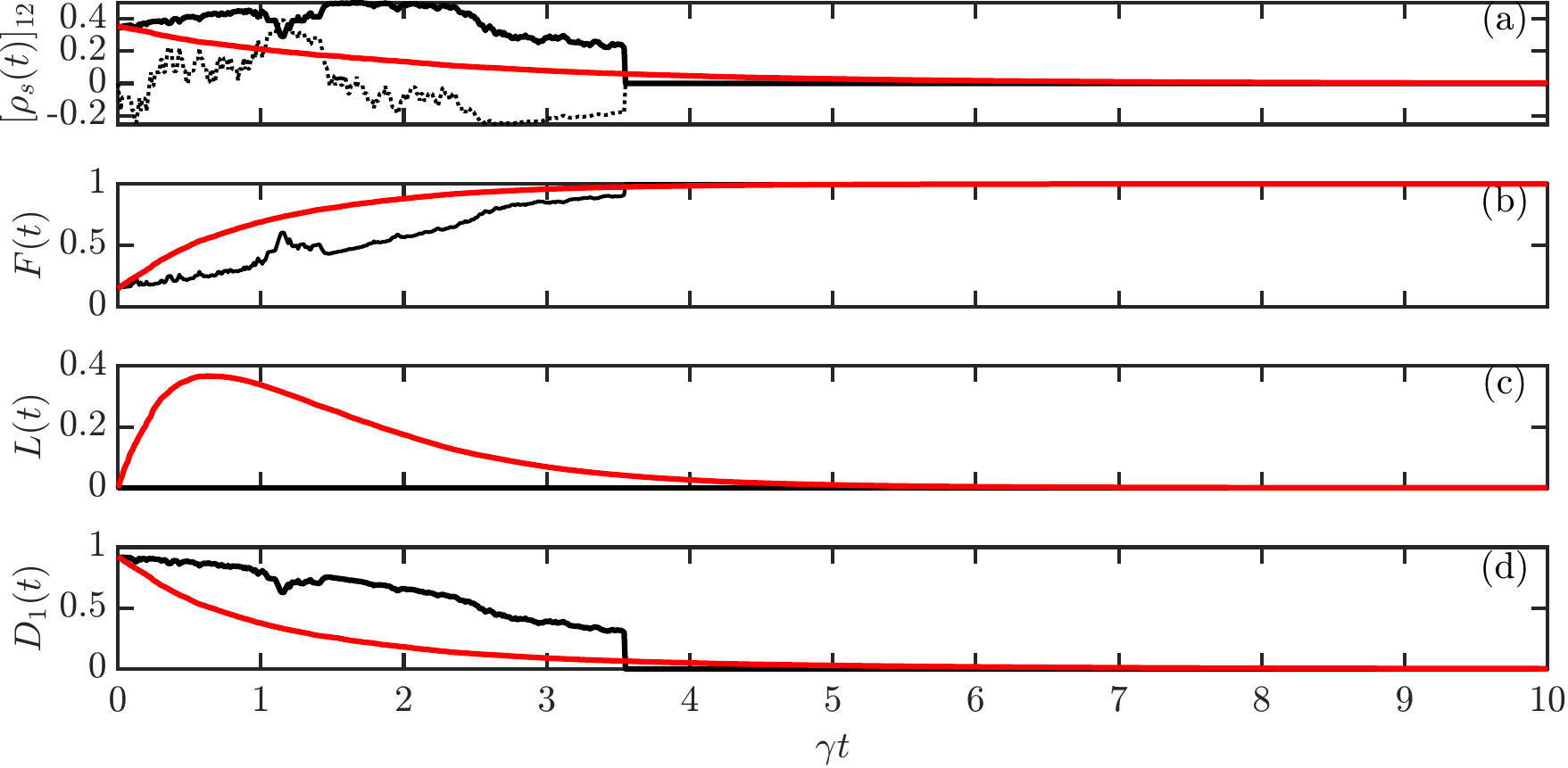}
        \caption{
       Entries and quantifiers of the density matrix as functions of time of a single qubit steered toward the north pole of the Bloch sphere, where the detector is measured in two different directions: the correct one---related to the ideal steering---and an erroneous one given by Eq.~\eqref{eq:basis}. Both measurement directions are used with the same probability as functions of time associated with the steering of a single qubit toward the north pole. (a) Coherences of the density matrices. (b) Fidelity. (c) Impurity. (d) Trace distance. Black curves correspond to a representative quantum trajectory solving Eq.~\eqref{eq:pb9}. In (a), the solid and dotted lines correspond to the real and the negative imaginary parts of the coherences. Red curves are associated with the average over $10^3$ quantum trajectories.}
    \label{fig:112}
\end{figure*}
The trajectory followed by the Bloch vector of the representative quantum trajectory shown in Figs.~\ref{fig111:a} and \ref{fig111:c}, displays the
full-fledged stochastic contributions contained in Eq.~\eqref{eq:pb9} as it approaches the north pole. Specifically, from its initial position, its evolution is randomly governed by continuous, non-unitary evolution present in Eq.~\eqref{eq:pb9} when $\dd N(t) = 0$ and $\chi_1(t)= 1$,
and the diffusive, non-unitary evolution when $\chi_2(t) = 1$. 
The steered state approaches the ideal target state---the north pole---notwithstanding, as shown also in the trace distance in Fig.~\ref{fig:112}d. This fluctuating behavior stops when a jump to the north pole is registered. The state stops evolving as this is a stationary state of Eq.~\eqref{eq:pb9}. 

Even though the quantum trajectories of Eq.~\eqref{eq:pb9} have a somewhat ``erratic'' behavior before a jump occurs (if it occurs), their average is precisely the one obtained from the ideal protocol where the measurement basis is always the correct one (see Fig.~\ref{fig:111}b and Fig.~\ref{fig:111}d and compare with Figs.~\ref{fig40:e} and \ref{fig40:j}). Therefore, as we stated at the beginning of this section, the protocol is fully robust to this type of error by design.

\section{ Discussion}\label{sec:conclusions}

We have studied the robustness of a measurement-induced quantum steering protocol to errors applied to a qubit. This protocol was introduced in Ref.~\cite{roy2020measurement} and is based on the repeated interaction of a chain of detectors with the steered system. After an interaction occurs, each detector is immediately measured, and the outcomes are not selected (``blind measurement''). The state of the steered system then reaches the predetermined target state with a given fidelity in a finite time.

The protocol could enable the preparation of any quantum system with finite degrees of freedom in a pure state, should it prove to be experimentally feasible. However, a realistic implementation of the protocol requires considering any possible undesired alteration of the protocol's steps and parameters, i.e., errors. We have sorted the errors into two categories depending on how they appear relative to each steering step: \emph{static}, if they are either constant or appear with a given probability at each steering step, or \emph{dynamic} otherwise. 

To simplify our analysis, we have studied one error at a time. We have considered two types of static errors (due to a wrongly chosen detector-system coupling parameter and erroneously prepared detector states) and four types of dynamic errors (due to fluctuating steering directions, fluctuating detector-system interaction strength, errors in the steering Hamiltonian, and erroneously chosen measurement direction, and errors in the steering Hamiltonian). We have set the error-free protocol as a reference and introduced various quantifiers such as fidelity, trace distance, and linear entropy (``impurity'') to characterize the protocol's robustness.

In our study of static errors, we have demonstrated how a wrongly chosen detector-system coupling might effectively implement either a projective measurement or a Pauli $\sigma^z$-gate on the steered state. Both occurrences impede the implementation of Lindbladian dynamics and might lead to the complete failure of the protocol when steering many-body systems.

We have also shown that erroneously prepared detectors can induce Lindbladian dynamics with an extra dissipation channel that steers toward a state orthogonal to the ideal one. While this type of dissipator is induced by the population $(\tilde\rho_d)_{22}$ of the detector state, its coherences give rise to a unitary channel. In the small-error approximation, that is, when the strength of the two former channels is small, the leading terms in the quantifiers (i.e., trace distance, impurity, and fidelity) were linear in the population $(\tilde\rho_d)_{22}$. Thus, the protocol was not that robust to this error.

As a result of their fluctuating behavior, dynamic errors produce more complex dynamics.
We have found three novel stochastic master equations describing different types of averaging hierarchies when individual detectors, interacting with the system, could steer it toward states different from the ideal one. These stochastic master equations differ from the two most common ones: One describes the detectors continuously monitoring the system, resulting in a sudden change to different pure states. In contrast, the other stochastic master equation describes how the detectors can induce a diffusive and non-unitary evolution on the steered states  \cite{barchielli1991measurements,attal2010stochastic,breuer2002theory,pellegrini2009diffusion}. 

Furthermore, we have demonstrated that, in contrast to these two types of equations, when both the random steering direction and detector results are considered, a weighted sum of stochastic master equations of the first type mentioned above is obtained, where the weights are stochastic indicators. Now, when the steering directions are averaged out---which would require massive post-selection in an experimental execution of the protocol---we found that whenever a click was registered (no matter from which detector), the steered state jumped to a mixed state. The system dynamics described by this equation coincides with the non-unitary dynamics found when several detectors monitor a quantum system, and there is no jump to any pure state. We also demonstrated how this stochastic master equation could be obtained from a different model of random, repeated interactions, where the detector-system Hamiltonian has several delta-correlated white noises.

We have further shown that when all the detector outcomes are averaged out, the system evolution is governed by the stochastic weighted sum of simple dissipators, each steering toward one of the available directions. We have provided a particular example with two erroneous steering directions parametrized by their probability and a polar angle in the Bloch sphere. In the stationary state regime, the quantifiers showed that the protocol is quite robust to this error, as its leading power in the polar angle is of order four.

In addition, we have investigated how an environment may interfere with the detector and steered system by using a perturbation Hamiltonian with multiplicative white noise. With this error, we have analytically studied the dynamics of the system and showed how the additional dissipative channels appear in the corresponding Lindbladian. We have developed two approaches to arrive at the resulting Lindblad equation. The first approach involves directly averaging the detector-system dynamics over realizations of the white noise. Subsequently, a blind measurement is performed. 

In our second approach, we have devised a novel stochastic master equation to simultaneously describe the influences of the detector measurement and the environment perturbation. This master equation is diffusive because it has a unitary fluctuating generator, unlike the master equation of the diffusive type mentioned. The deterministic part of this stochastic master equation includes the detector's contribution (the finite backaction), an additional backaction caused by the environment, and a dissipator, which is the It\^o correction of the fluctuating unitary generator. As part of this equation, an inhomogeneous Poissonian process represents the jump part. However, unlike the usual jump terms used in standard stochastic master equations, this one describes a jump to a mixed state instead of a pure state. Because of the novelty of this equation, it would be necessary to extend and strengthen its mathematical foundations in the same manner as it was done for the jump, diffusive, and more standard diffusive-type stochastic master equations \cite{attal_repeated_2006, attal2010stochastic,barchielli1991measurements,pellegrini2009diffusion, pellegrini_markov_2010}. With this error, we analyzed the dynamics analytically for a particular form of the perturbation Hamiltonian. We showed how the additional dissipative channels appear in the system dynamics. We have shown that when the ideal target state is an eigenstate of the constant operator of the perturbation, the protocol displays complete robustness. This is no longer the case when the former condition is not fulfilled.

Our findings indicate that errors due to fluctuating detector-system strength could be quenched or time-dependent. However, we have found that the resultant Lindblad equation remains the same in both cases, and only the dissipation rate is affected. Further, while fluctuating measurement directions do not change the averaged dynamics, they alter the nature of the stochastic differential equation describing the system dynamics. Specifically, we have demonstrated that the system could follow a jump-type or diffusive-type behavior after every steering step.

Our work opens up several future directions. In particular, our analytical treatment of the errors in steering directions, where we derive three stochastic master equations, is essentially valid for more complicated dynamics having multiple stochasticities. This approach can be applied to study measurement-induced entanglement transitions \cite{zerba,turkeshi2021measurement,lunt2020measurement,doggen2022generalized,ippoliti2021entanglement,skinner2019measurement,li2019measurement} where multiple stochasticities in the dynamics can stem from different errors or multiple measurement observables.

Even though we have focused on a single qubit steering, our approach can be systematically applied to quantum systems possessing a larger number of degrees of freedom, where, in addition, multiple errors can occur. Importantly, when considering systems with two degrees of freedom (or more), the role of static and dynamic errors in modifying  or undermining measurement-engineered entanglement is an outstanding challenge. While the present work only addressed the question of how different errors affect the steering protocol, the present study can be taken as a starting point for developing stabilizer codes \cite{nielsen2002quantum,lidar2013quantum,continuous_error_correction_experiment}, with the prospects of implementing error correction schemes.

Our results can be adapted and readily used in a number of experimental platforms where errors and noise are present in various measurement-based protocols. These include the observation of topological transitions in single qubits implementing weak measurements \cite{topological_transition_weak_measurements_experiment,topological_transitions_optical_experiment}, monitoring superconducting qubits via weak-measurements where the quantum trajectories are registered  \cite{quantum_trajectory_tracking_nn,qubit_quantum_traj_exp,continuous_control_experiment,quantum_jumps_Devoret}, measuring incompatible operators in superconducting qubits via weak measurements using ancillary quantum systems \cite{weak_measurements_experiment}, among other applications and platforms. The relevant types of errors and their parameters can vary from setup to setup, as can be inferred from the above references.

Note that the list of errors studied in this work is far from exhaustive. For example, it would be interesting to investigate other perturbation Hamiltonians with additional (even non-multiplicative) noise sources. More formally, it would be intriguing to mathematically substantiate our novel stochastic equations (comprising white noise), as was done with the jump and diffusive stochastic master equations \cite{pellegrini_poisson_2009,attal_repeated_2006}.

\section*{Acknowledgments} 
We thank Samuel Morales, Artem Sapozhnikov, Ronen Eldan, Markus Fröb, and Ofer Zeitouni for useful discussions. The work was supported by the Deutsche
    Forschungsgemeinschaft (DFG): Project No. 277101999
    -- TRR 183 (Project B02) and Grants No. EG 96/13-1 and No. GO 1405/6-1, the Helmholtz International Fellow
    Award, and the Israel Binational Science Foundation-- National Science Foundation through award
    DMR-2037654. 
   EMG, IVG, and YG are grateful to Departamento de F{\'i}sica, FCFM, University of Chile (Santiago) for hospitality during the final stage of this work.
   EMG. and IVG also acknowledge the support by the European Commission under the EU Horizon 2020 MSCA-RISE-2019 program (Project 873028 HYDROTRONICS).
 
\begin{widetext}

\appendix
\section{The derivation of the Lindblad equation associated with errors in the detector states}\label{sec:A1}

In this appendix, we derive the LE~\eqref{eq:14} from both the static and quenched versions of the errors in the detector state.  

\subsection{Static error}
Let us first consider the erroneous detector state 
\begin{equation}\label{eq:a1}
      \tilde\rho_d = \begin{pmatrix}
    a & \abs{b}\exp(i\phi) \\ \abs{b}\exp(-i\phi) & 1-a
    \end{pmatrix}
\end{equation}
written in the ONB $\mathcal{B}_d = \{ \ket{\Phi_d},\ket*{\Phi_d^\perp}\}$. Before the detector interacts with the steered system, we have the product state
\begin{equation}\label{eq:a2}
   \tilde\rho_d\otimes \rho_s(t) = \begin{pmatrix}
    a\rho_s(t) & b\rho_s(t) \\ b^*\rho_s(t) & (1-a)\rho_s(t)
    \end{pmatrix}.
\end{equation}
The detector-system Hamiltonian $H_\text{ds} = J\left( \ket*{\Phi_d^\perp}\bra{\Phi_d}\otimes A + \hc \right)$ can also be written in the ONB $\mathcal{B}_d$ as
\begin{equation}\label{eq:a3}
    H_\text{ds} = Jh_0 \equiv  \gamma\begin{pmatrix}
    0 & A^\dagger \\ A & 0
    \end{pmatrix}.
\end{equation}
Next, we replace Eqs. \eqref{eq:a2}-\eqref{eq:a3} into the second and third terms of the series expansion of the evolved detector-system state
\begin{equation}\label{eq:a3.1}
    \rho_{\text{ds}}(t+ \delta t ) = \exp[-iJ \delta t \, \ad (h_0)]\tilde\rho_d \otimes \rho_s(t) = \left( I_{ds} - iJ\delta t \ad (h_0)  - \frac{J^2\delta t^2}{2}\ad^2 (h_0)\right) \tilde\rho_d\otimes \rho_s(t)  + \mathcal{O}(J^3\delta t^3),   
\end{equation}
which are
\begin{equation}\label{eq:a4}
\ad( h_0)\tilde\rho_d\otimes\rho_s(t)
= \begin{pmatrix}
b^*A^\dagger \rho_s - b\rho_s(t) A &  (1-a)A^\dagger \rho_s(t) -a\rho_s(t)A^\dagger \\
 aA \rho_s(t) -c \rho_s(t)A  &  bA\rho_s(t) -b^*\rho_s(t)A^\dagger 
\end{pmatrix},
\end{equation}
and
\begin{equation}\label{eq:a5}
\ad^2 (h_0)\tilde\rho_d\otimes\rho_s(t)
=\begin{pmatrix}
\{aA^\dagger A, \rho_s \}- 2(1-a)A^\dagger\rho_s A  & bA^\dagger A \rho_s -2 b^* A^\dagger\rho_sA^\dagger + b\rho_s AA^\dagger   \\
-2bA \rho_s A + b^*AA^\dagger \rho_s+ b^*\rho_sA^\dagger A & \{(1-a)AA^\dagger, \rho_s\} - 2aA\rho_s A^\dagger 
\end{pmatrix}.
\end{equation}
Next, we take the partial trace with respect to the detectors to get
\begin{equation}\label{eq:a6}
    \rho_s(t + \delta t) = \rho_s(t) + \left[ -iJ\abs{b}\delta t\,\ad\textbf{(}\exp(i\phi)A + \hc \textbf{)} + a J^2\delta t^2\mathcal{D}(A) + (1-a)J^2\delta t^2\mathcal{D}(A^\dagger) \right]\rho_s(t)  + \mathcal{O}(J^3\delta t^3).
\end{equation}
The next step is to evaluate the limit $$\lim_{\delta t \to 0}\frac{\rho_s(t+ \delta t) - \rho_s(t)}{\delta t}$$ while guaranteeing that the products $\kappa = J\abs{b}$ and $\gamma = J^2\delta t$ are kept constant---This is the WM limit [cf. Eq.~\eqref{eq:02-1}]. The resulting equation is then 
\begin{equation}\label{eq:a6.1}
\partial_t\rho_s(t) = \left[ -i\kappa \, \ad(\tilde{h}) + \gamma_+ \mathcal{D}(A) + \gamma_-\mathcal{D}(A^\dagger) \right]\rho_s(t),  
\end{equation}
which is Eq.~\eqref{eq:14}, where $\gamma_+ \coloneqq a\gamma$ and $\gamma_- \coloneqq (1-a)\gamma$ are the decays, and $\tilde{h} = \exp(i\phi) A + \hc$

\subsection{Quenched error}
Let us now assume that at each steering step, the detector state is $\ket{\Phi_d^i} = \cos(\theta_i/2)\ket{\Phi_d}+e^{i\varphi_i}\sin(\theta_i/2)\ket*{\Phi_d^\perp}$ and randomly chosen from an ensemble with probability $p(i)$ such that  $\sum_ip(i)\ket{\Phi_d^i}\bra{\Phi_d^i} = \tilde \rho_d$. Given this quenched version of the error, the discrete-time stochastic master equation governing the dynamics of the steered state from time $t_k = k\delta t$ to $t_{k+1} =(k+1)\delta t$ is
\begin{equation}\label{eq:rvb2.2}
     \omega_{k+1} = \sum_{i,\alpha}\frac{\mathcal{M}^{(i)}_{\alpha}\omega_k}{p(i,\alpha\vert \omega_k)} \mathds{1}_{i,\alpha}^{k+1},
\end{equation}
where $\omega(t_k)= \omega_k$,  $\mathbb{E}[\mathds{1}_{i,\alpha}^{k+1}] = p(i,\alpha\vert \omega_k)$ is the probability of the indicator function of the outcomes $(i,\alpha)$ given the previous state $\omega_k$, and $\mathcal{M}^{(i)}_{\alpha}\omega_k = p(i)\bra{\alpha}U(\delta t)\ket{\Phi_d^i}\omega_k \bra{\Phi_d^i} U^\dagger(\delta t)\ket{\alpha}$ is the updated, unnormalized state. We have mapped $\ket{\Phi_d}\mapsto \ket{0}$ and $\ket*{\Phi_d^\perp} \mapsto \ket{1}$. [See Appendix~\ref{sec:appendix_jump_several_directions_static} for a more detailed description of the notation used in Eq.~\eqref{eq:rvb2.2}.]

If we average with respect to the detector outcomes and the detectors in \eqref{eq:rvb2.2}, i.e., $\mathbb{E}[\omega_k] = \rho_k$ and perform a series expansion in $\delta t$, we can have the formal derivative
\begin{align}\label{eq:rvb2.3}
   \partial_t\rho_s(t) &=  \lim_{\delta t\to 0}\dfrac{\rho_s(t+\delta t)-\rho_s(t)}{\delta t} \notag \\
   &= \sum_ip(i)\Bigl( -iJ\sin \frac{\theta_i}{2}\cos\frac{\theta_i}{2}[Ae^{i\varphi_i}+\text{h.c.},\rho_s(t)] +\cos^2 \frac{\theta_i}{2}J^2\delta t\mathcal{D}(A)\rho_s(t)\notag \\ &\,\,\,\,+\sin^2\frac{\theta_i}{2}J^2\delta t\mathcal{D}(A^\dagger)\rho_s(t)
 \Bigr) + \mathcal{O}(J^3\delta t^2).
\end{align}
Therefore, we recover \eqref{eq:a6.1} in the weak-measurement if
\begin{equation}
    \sum_ip(i)\sin\frac{\theta_i}{2}\cos\frac{\theta_i}{2}e^{i\varphi_i}= \kappa e^{i\phi}\sqrt{\frac{\delta t}{\gamma}} \quad \text{and} \quad \sum_ip(i)\cos^2\frac{\theta_i}{2} = a.
\end{equation}

 \section{Stationary ellipsoid}\label{sec:stationary_ellipsoid}
 Here, we demonstrate the properties obeyed by the stationary ellipsoids described by Eq.~\eqref{eq:18}. 

 Let us demonstrate these two properties by denoting the punctured ellipsoid by $\mathcal{C}$. To demonstrate that $\{(0,0,0)\} \notin \mathcal{C}$, we represent the Bloch vector $\bm{r}_\infty$ in spherical coordinates:
    \begin{equation}
     \bm{r}_\infty = (r\sin\theta\cos\varphi, r\sin\theta\sin\varphi, r\cos\theta).
    \end{equation}
For $\kappa > 0$, the polar angle is given by 
    \begin{equation}\label{eq:19}
        \theta = \arccos{\left[\text{sgn}(2a-1)\frac{\gamma}{\sqrt{\gamma^2 + 16\kappa^2}}\right]}.
    \end{equation}
For a fixed value of $\gamma$, the polar angle tends to $\theta \rightarrow \pi/2$  as  $\kappa \rightarrow \infty$, and therefore
      \begin{equation}\label{eq:19.1}
   r = \norm{\bm{r}_{\infty}} = \frac{\abs{2a-1}\gamma \sqrt{\gamma^2 + 16\kappa^2}}{\gamma^2 + 8\kappa^2}  \rightarrow 0 \,
\end{equation}
as $\kappa\rightarrow \infty$.
Thus, the endpoint of the ellipsoid's minor axes coinciding with the origin of the Bloch sphere does not belong to $\mathcal{C}$.   

The other statement is directly checked from Eq. \eqref{eq:18} by setting $x_\infty = y_\infty = 0$. This gives that the other endpoint of the minor axis is located at $z_\infty = 2a-1$, which is never equal to $\pm 1$, as $a \notin \{ 0, 1\}$. The physically attainable point $\bm{r}_\infty = (0,0,2a-1)$ is obtained when $\kappa \rightarrow 0$.

\section{Stochastic differential equations for erroneous steering directions}\label{sec:appendix_jump_several_directions_static}
Here, we will derive the stochastic master equations (SMEs) shown in Sect. \ref{ssec:errors_steering_direction} corresponding to different averaging hierarchies with several steering directions: Eq. \eqref{eq:sde1}, where the detector readouts and the random steering directions are present; Eq. \eqref{eq:ssde5}, where the steering directions are averaged out; and Eq. \eqref{eq:sde11}, where the detector readouts are averaged out. 

We denote the set of steering directions as $\mathcal{R} = \{(\theta_i,\varphi_i; p(i)) \}_{i\in \mathcal{I}}$, where $\mathcal{I} = \{1,2,\ldots, n \}$ is an index set indicating the number of steering directions, and $p(i)$ denotes the probability of steering toward the $i$-th direction parameterized by the angles $(\theta_i,\varphi_i)$ with $\theta_i \in [0,\pi]$ and $\varphi_i \in [0,2\pi)$.  Note that the ideal target state may be contained in $\mathcal{R}$, yet it is not required. We label each steering step by an integer number, e.g., $k \in \mathbb{Z}^+$ for a given time $t_k = k\delta t$, and we relabel the states of the system as $\omega_s(t_k) = \omega_k$.

\subsection{On the Kraus operators}
Before attempting to derive the SMEs, we must construct and adequately understand the Kraus operators inducing the operations on the system density matrix. 

The measurement operators associated to a detector readout  $\alpha$ \emph{given} that the $i$-th direction appeared (see below) are
\begin{equation}\label{eq:gen1}
    \Omega_{i,\alpha} \coloneqq M_\alpha^{(i)}(\delta t) = \bra{\alpha}\exp(-iJ\delta t h_0^{(i)})\ket{0} \, \alpha \in \{0,1 \},
\end{equation}
where $h^{(i)}_0 = \ket*{\Phi_d^\perp}\bra{\Phi_d}\otimes A_i + \hc$ is the dimensionless operator associated with the Hmailtonian of the $i$-th direction (or state), and $A_i = R(\theta_i, \varphi_i)UR^\dagger (\theta_i,\varphi_i)$ is the rotated operator operator annihilating the $i$-th target state  [cf. Eq.~\eqref{eq:7.3}]. 

Let us define the following operators
\begin{equation}\label{eq:gen2}
    W_{i,\alpha} \coloneqq \sqrt{p(i)}\Omega_{i,\alpha} \quad \text{with} \quad i \in \{1,\ldots, n \} \quad \text{and} \quad \alpha \in \{0,1\}.
\end{equation}
For a given prior state $\omega_k$, if upon a measurement of the local observable $S_d^{(i)} = \ket{\Phi_d^i}\bra{\Phi_d^i} - \ket*{\Phi^{i,\perp}_d}\bra*{\Phi^{i,\perp}_d}$ it is revealed that the $i$-th direction and the outcome $\alpha$ where measured
\emph{at the same time}, then the unnormalized posterior state is given by 
\begin{equation}\label{eq:gen3}
    \Tilde{\omega}_{k+1} = \mathcal{M}^{(i)}_{\alpha}\omega_k \coloneqq W_{i,\alpha}\omega_kW_{i,\alpha}^\dagger = p(i) \Omega_{i,\alpha}\omega_k \Omega_{i,\alpha}^\dagger.
\end{equation}
The joint probability of having the $i$-th direction \emph{and} the detector outcome $\alpha$ given the prior state $\omega_k$ is then 
\begin{equation}\label{eq:gen4}
  p(i,\alpha\vert \omega_k) \coloneqq \Tr\Tilde{\omega}_{k+1} = p(i)\Tr(\Omega_{i,\alpha}^\dagger \Omega_{i,\alpha}\omega_k).
\end{equation}
We want to point out once more that the direction and click result are \emph{simultaneously} read. 

The above joint probability equation shows that the term multiplying $p(i)$ is the conditional probability of obtaining the result $\alpha$ given the $i$-th direction and prior state $\omega_k$ appeared. We denote this probability by
\begin{equation}\label{eq:gen5}
    p(\alpha\vert i; \omega_k) \coloneqq \Tr(\Omega_{i,\alpha}^\dagger \Omega_{i,\alpha}\omega_k).
\end{equation}
At first, this choice of probability might seem strange, and one might think that it describes the joint probability Eq. \eqref{eq:gen4}. To disprove this, we give the following arguments: First, the conditional probability Eq. \eqref{eq:gen5} fulfills the well-known normalization condition (in terms of probability)
\begin{align*}
   \sum_{\alpha = 0,1}p(\alpha\vert i; \omega_k) &= \sum_{\alpha = 0,1}\Tr(\Omega_{i,\alpha}^\dagger \Omega_{i,\alpha}\omega_k) 
    = \sum_{\alpha = 0,1}\Tr[ M_\alpha{^{(i)}}(\delta t)^\dagger M_\alpha{^{(i)}}(\delta t)\omega_k] 
    = \Tr[\omega_k] = 1.
\end{align*}
Therefore, summing over $\alpha$ in Eq. \eqref{eq:gen4} gives
\begin{equation}\label{eq:gen7}
    \sum_{\alpha}p(i,\alpha\vert \omega_k) = p(i),
\end{equation}
which is precisely the marginal probability of the $i$-th steering direction. Moreover, the right-hand side  of  Eq. \eqref{eq:gen4} is just the formula for the conditional probability
\begin{equation}\label{eq:gen7.1}
    p(i,\alpha\vert \omega_k) = p(i)p(\alpha\vert i;\omega_k).
\end{equation}
Note that the probability $p(i\vert \omega_k) \equiv  p(i)$ is independent of the prior state.

A trivial consequence of Eq. \eqref{eq:gen7} is that summing over the directions gives unity, so
\begin{equation}\label{eq:gen8}
    \sum_{\alpha,i}p(\alpha, i\vert \omega_k) = 1.
\end{equation}
Therefore, the joint probabilities defined in Eq. \eqref{eq:gen4} are well defined. In addition, the set $\{A_{i,\alpha}\}$ satisfies the Kraus condition
\begin{equation}\label{eq:gen10}
    \sum_{m,\alpha}W_{i,\alpha}^\dagger W_{i,\alpha} = I_s.
\end{equation}
Our last argument in favor of the definition of the operators $\{A_{i,\alpha}\}$ is that they naturally incorporate the classical probability $p(m)$. 

\subsection{The three averaging hierarchies}
Having defined the set of Kraus operators $\{W_{i,\alpha} \}$, the updated state upon jointly obtaining the click $\alpha$ and the $i$-th direction is 
\begin{equation}\label{eq:gen11}
    \omega_{k+1} = \frac{\mathcal{M}_\alpha^{(i)}\omega_k}{p(i,\alpha\vert \omega_k)}.
\end{equation}
This is our starting point to find the SMEs describing different averaging hierarchies, starting with Eq.~\eqref{eq:sde1}.

\subsubsection{Full stochasticity}
With the aid of Eq. \eqref{eq:gen11}, we can write the discrete SME
\begin{equation}\label{eq:gen12}
    \omega_{k+1} = \sum_{i,\alpha}\frac{\mathcal{M}^{(i)}_{\alpha}\omega_k}{p(i,\alpha\vert \omega_k)} \mathds{1}_{i,\alpha}^{k+1},
\end{equation}
where $\{\mathds{1}_{i,\alpha}^{k+1}\}$ is a set of indicator  functions from which a \emph{single} one appears between $t_k$ and $t_{k+1}$ with expectation
\begin{equation}\label{eq:gen13}
    \mathbb{E}[\mathds{1}_{i,\alpha}^{k+1}] = p(i,\alpha\vert \omega_k).
\end{equation}
In other words, the indicator functions act as the stochastic variables describing the direction and click outcomes. Hence, 
Eq. \eqref{eq:gen12} describes the stochastic, discrete evolution of a steered state once a given direction and a detector outcome are determined at the same time.

Pointing toward the WM limit [cf. Eq.~\eqref{eq:02-1}], we perform the now usual rescaling of the detector-system coupling strength $J = \sqrt{\gamma/\delta t}$, and thus we have up to the first order in $\delta t$
\begin{align}\label{eq:gen13.1}
    \frac{\mathcal{M}^{(i)}_{0}\omega_k}{p(i,\alpha=0\vert \omega_k)} &= \omega_k - \frac{\gamma \delta t}{2}\{A_i^\dagger A_i - \expval*{A_i^\dagger A_i}_k,\omega_k \} + \mathcal{O}(\delta t^2)\\
    \frac{\mathcal{M}^{(i)}_{1}\omega_k}{{p(i,\alpha=1\vert \omega_k)}} &= \frac{A_i\omega_k A_i^\dagger}{\expval*{A_i^\dagger A_i}_k} +  \mathcal{O}(\delta t^2),
\end{align}
where $\expval*{A_i^\dagger A_i}_k \coloneqq \Tr(A_i^\dagger A_i\omega_k)$.
Replacing the above expansions in Eq. \eqref{eq:gen12} gives
\begin{equation}\label{eq:gen13.2}
\omega_{k+1} = \sum_i \left( \omega_k - \frac{\gamma \delta t}{2}\{A_i^\dagger A_i - \expval*{A_i^\dagger A_i}_t,\omega_k \}\right)\mathds{1}_{i,\alpha = 0}^{k+1} + \sum_i \frac{A_i\omega_k A_i^\dagger}{\expval*{A_i^\dagger A_i}_t}\mathds{1}_{i,\alpha = 1}^{k+1} +  \mathcal{O}(\delta t^2),
\end{equation}
where the expectation of the indicators become
\begin{align}
    \mathbb{E}[\mathds{1}_{i,\alpha = 0}^{k+1}] &= (1 - \gamma \delta t \expval*{A^\dagger_iA_i}_k)p(i) + \mathcal{O}(\delta t^2), \label{eq:gen13.3} \\
    \mathbb{E}[\mathds{1}_{i,\alpha = 1}^{k+1}] &= \gamma \delta t \expval*{A^\dagger_i A_i}_k p(i) +  \mathcal{O}(\delta t^2). \label{eq:gen13.3-1}
\end{align}

After adding an appropriate zero operator in Eq. \eqref{eq:gen13.2} and setting $\delta t = \dd t \rightarrow 0$, we have
\begin{equation}\label{eq:gen13.4}
    \dd \omega(t) = \sum_i\left(\gamma\expval*{A_i^\dagger A_i}_t\omega(t) - \frac{\gamma }{2}\{A^\dagger_i A_i,\omega(t) \}  \right)\chi_i(t) \dd t + \sum_i\left( \frac{A_i\omega(t)A_i^\dagger}{\expval*{A_i^\dagger A_i}_t} - \omega(t) \right)\chi_i(t)\dd N_m(t),
\end{equation}
where
\begin{equation}\label{eq:ecu1}
    \mathds{1}_{i,\alpha = 0}^{k+1} \longrightarrow \chi_i(t) \quad \text{as} \quad \delta t = \dd t \rightarrow 0
\end{equation}
and
\begin{equation}\label{eq:gen13.5}
    \mathbb{E}[\chi_i(t)] = p(i).
\end{equation}
Similarly, 
\begin{equation}\label{eq:ecu2}
\mathds{1}_{i,\alpha = 1}^{k+1} \longrightarrow \chi_i(t)\dd N_i(t) \quad \text{as} \quad \delta t = \dd t \rightarrow 0
\end{equation}
with
\begin{equation}\label{eq:poisson_full_1}
    \mathbb{E}[\chi_i(t)\dd N_i(t)] = \gamma\expval*{A_i^\dagger A_i}_tp(i) \dd t = \gamma \Tr[A_i^\dagger A_i\omega(t)]p(i)\dd t.
\end{equation}
Arranging the SME in Eq. \eqref{eq:gen13.4} gives Eq. \eqref{eq:sde1}.

Clearly, the stochastic variable $\chi_i(t)$ expresses that the $i$-th steering direction is obtained between $\omega(t)$ and $\omega(t+ \delta t)$, and $\dd N_i(t)$ is the Poissonian increment registering a jump (or the lack of it) with strength $\gamma \expval*{A_i^\dagger A_i}_t\dd t$ given that the latter direction appeared.

Taking the full average over all random variables in Eq.  \eqref{eq:gen13.4} gives
\begin{equation}\label{eq:ecu3}
    \dd \mathbb{E}[\omega(t)] = \sum_i  \gamma p(i) \mathcal{D}(A_i)\mathbb{E}[\omega(t)]\dd t,
\end{equation}
which corresponds to the LE 
\begin{equation}\label{eq:finalle}
    \partial_t \rho(t) = \sum_{i}\gamma p(i) \mathcal{D}(A_i)\rho(t)
\end{equation}
describing the fully averaged dynamics of the steered density matrix. Here, $\rho(t) \coloneqq \mathbb{E}[\omega(t)]$ is the trajectory-averaged density matrix.

\subsubsection{Average over the directions}
We will now demonstrate Eq.~\eqref{eq:ssde5}. For this, let $\pi_k = \mathbb{E}_i[\omega_k]$ be a prior state where $\mathbb{E}_i$ denotes the classical average over the steering directions.
Suppose the click result $\alpha$ was obtained, and we are \emph{not} interested in which steering direction was obtained. Therefore, we must pre-multiply a posterior state \eqref{eq:gen11} by the conditional probability of obtaining the $i$-th direction \emph{given} the click result $\alpha$ and the prior state $\pi_k$ were obtained, $p(i\vert \alpha; \pi_k)$. This gives the updated state
\begin{equation}\label{eq:gen14}
    \pi_{k + 1} = \sum_{i}p(i\vert \alpha; \pi_k)\frac{\mathcal{M}^{(i)}_{\alpha}\pi_k}{p(i,\alpha\vert \pi_k)}.
\end{equation}
By implementing Bayes rule
\begin{equation}\label{eq:gen15}
    p(i\vert \alpha;\pi_k)p(\alpha\vert \pi_k)= p(\alpha\vert i;\pi_k)p(i\vert \pi_k)
\end{equation}
the updated state becomes
\begin{equation}\label{eq:gen16}
    \pi_{k + 1} = \frac{1}{p(\alpha\vert \pi_k)} \sum_{i}\mathcal{M}^{(i)}_{\alpha}\pi_k.
\end{equation}

Similarly to the fully stochastic case [cf. Eq. \eqref{eq:gen12}]
the discrete SME describing the evolution of the posterior state is
\begin{equation}\label{eq:gen17}
    \pi_{k + 1} =\sum_{\alpha} \frac{1}{p(\alpha\vert \pi_k)}\left[\sum_{i}p(i){\Phi_{i,\alpha}(\pi_k)}\right]\mathds{1}_\alpha^{k+1},
\end{equation}
where
\begin{equation}\label{eq:gen18}
   \mathbb{E}_\alpha[\mathds{1}_\alpha^{k+1}] = p(\alpha\vert \pi_k)
\end{equation}
is the expectation value of the new indicator function $\mathds{1}_\alpha^{k+1}$. This function is indeed the conditional expectation value of the indicator function $\mathds{1}_{i,\alpha}^{k+1}$ given that the click result $\alpha$ was obtained. We denote this by [cf. Eqs.~\eqref{eq:gen13.3}-\eqref{eq:gen13.3-1}]
\begin{equation}\label{eq:gen19}
    \mathbb{E}_{i}[\mathds{1}_{i,\alpha}^{k+1}] \coloneqq p(i\vert \alpha;\pi_k)\mathds{1}_\alpha^{k+1}.
\end{equation}
More specifically, 
\begin{align}\label{eq:eqn20}
    \mathbb{E}_{i}[\mathds{1}_{i,\alpha=0}^{k+1}] &= \frac{p(\alpha=0 \vert i;\pi_k)p(i)}{\sum_n p(\alpha = 0\vert n;\pi_k)p(n)}\mathds{1}_{\alpha = 0}^{k+1} = \left(1 - \gamma \delta t\expval{A_i^\dagger A_i - \sum_np(n)A_n^\dagger A_n }_k \right)p(i)\mathds{1}_{\alpha = 0}^{k+1} + \mathcal{O}(\delta t^2), \\
    \mathbb{E}_{i}[\mathds{1}_{i,\alpha=1}^{k+1}] &= \frac{p(\alpha= 1 \vert i;\pi_k)p(i)}{\sum_n p(\alpha = 1\vert n;\pi_k)p(n)}\mathds{1}_{\alpha = 1}^{k+1} =  \frac{\expval*{A_i^\dagger A_i}_kp(i)}{\expval{\sum_n{p(n)A_n^\dagger A_n}}_k}\mathds{1}_{\alpha = 1}^{k+1}+ \mathcal{O}(\delta t^2).
\end{align}
Alternatively, we can use Eq. \eqref{eq:gen17} to derive continuous-time SME. To do so, we need the following expressions for the marginal click probabilities
\begin{align}\label{eq:eqn22.1}
    p(\alpha = 0 \vert \pi_k) &= \sum_np(\alpha = 0\vert n; \pi_k)p(n) = 1 - \gamma \delta t\expval{\sum_np(n) A_n^\dagger A_n}_k + \mathcal{O}(\delta t^2),\\
    p(\alpha = 1 \vert \pi_k) &= \sum_np(\alpha = 1\vert n; \pi_k)p(n) = \gamma \delta t\expval{\sum_np(n) A_n^\dagger A_n}_k + \mathcal{O}(\delta t^2).
\end{align}

Expanding the discrete SME gives us 
\begin{multline}\label{eq:eqn22}
    \pi_{k+1} = \sum_i\left(1 + \gamma \delta t\expval{\sum_np(n)A_n^\dagger A_n}_k\right)\left(\pi_k - \frac{\gamma \delta t}{2}\{A_i^\dagger A_i, \pi_k \}\right)p(i) \mathds{1}_{\alpha=0}^{k+1} 
    + \frac{\sum_i p(i) A_i\pi_kA_i^\dagger }{\expval{\sum_np(n)A_n^\dagger A_n}_k}\mathds{1}_{\alpha=1}^{k+1} + \mathcal{O}(\delta t^2).
\end{multline}
Now, since the probability of obtaining a click is proportional to $\delta t$ in leading order, we can set $\mathds{1}_{\alpha=0}^{k+1} \approx 1$ for all $k$. Moreover, in the WM limit, we have [cf. Eq.~\eqref{eq:poisson_full_1}] $ \mathds{1}_{\alpha=1}^{k+1}\rightarrow \dd N(t)$ with 
\begin{equation}\label{eq:eqn23}
    \mathbb{E}_{\alpha}[\dd N(t)] =  \gamma \expval{\sum_n p(n)A_n^\dagger A_n}_t \dd t.
\end{equation}
The SME is then
\begin{equation}\label{eq:eqn24}
\dd \pi(t) = \left(\expval{\sum_n \gamma p(n) A_n^\dagger A_n}_t\pi(t) - \sum_i \frac{\gamma}{2}\{p(i)A_i^\dagger A_i,\pi(t) \} \right)\dd t + \left( \frac{\sum_i p(i)A_i \pi(t)A_i^\dagger}{\expval{\sum_n p(n) A_n^\dagger A_n}_t} -\pi(t)\right)\dd N(t).
\end{equation}
This coincides with Eq. \eqref{eq:ssde5} after some rearrangement.
After taking the average over clicks, we have [cf. Eq. \eqref{eq:ecu3}]
\begin{equation}\label{eq:ecu5}
    \mathbb{E}_\alpha[\dd \pi(t)] = \dd \mathbb{E}_\alpha[\pi(t)] = \sum_i \gamma p(i) \mathcal{D}(A_i)\mathbb{E}_\alpha[\pi(t)]\dd t.
\end{equation}
The relation between the two types of density matrices is $\rho(t) = \mathbb{E}_\alpha[\pi(t)] =\left( \mathbb{E}_\alpha \circ \mathbb{E}_i\right)[\omega(t)]$.

\subsubsection{Stationary state of the deterministic map}\label{ssec:stationary_state_of_the_deterministic_map}
We demonstrate Eq.~\eqref{eq:sde6}, which is the stable fixed point of the deterministic part of Eq.~\eqref{eq:eqn24}.

When no jump occurs---either at all, i.e., $\dd N(t) \equiv 0$, or between $t$ and $t+\dd t$---and the steering directions are averaged, the evolution of the state is given by 
\begin{equation}\label{eq:ecu6}
    \partial_t \pi^{\text{det}}(t) = \expval{\sum_n \gamma p(n) A^\dagger_nA_n}_t\pi^{\text{det}}(t) - \sum_i \frac{\gamma}{2}\{p(i)A_i^\dagger A_i, \pi^{\text{det}}(t)\},
\end{equation}
where $\pi^{\text{det}}$ denotes a density matrix that evolves deterministically.
We aim to find its stationary state. To accomplish this, we make use of the ordered matrix ONB $F = (F_0,F_1,F_2,F_3) = (F_0, \bm{F}) \coloneqq (I_s, \bm{\sigma})/\sqrt{2}$ with respect to the inner product 
\begin{equation}\label{eq:ecu7}
    (F_\mu,F_\nu) \coloneqq \Tr(F_\mu^\dagger F_\nu) = \delta_{\mu \nu}.
\end{equation}
Greek indices run from 0 to 3, and Latin from 1 to 3. With respect to the above ONB, the density matrix will be denoted as $\pi = \sum_\mu{x_\mu}F_\mu$.

Given the above notation, we will use the identity \cite{gamel2016entangled}
\begin{equation}\label{eq:ecu9}
    F_\mu F_\nu = \frac{1}{\sqrt{2}}\sum_\gamma(\theta_{\mu\nu \gamma } + i\varepsilon_{\mu\nu\gamma})F_\gamma,
\end{equation}
where 
\begin{align}\label{eq:ecu10}
    \theta_{\mu\nu\gamma} \coloneqq \begin{cases}
    1 & \text{one index is 0, the other two equal}; \\
    0 &\text{otherwise},
    \end{cases}
\end{align}
is a fully symmetric tensor satisfying 
\begin{equation}\label{eq:ecu10.1}
    \theta_{\mu \nu a} = \delta_{\mu 0}\delta_{\nu a} + \delta_{\nu 0}\delta_{\mu a}, \quad \theta_{\mu \nu 0} = \delta_{\mu \nu} 
\end{equation}
and 
\begin{align}\label{eq:ecu11}
    \varepsilon_{\mu\nu\gamma} \coloneqq \begin{cases}
    1 &\mu\nu\gamma \in \{123,231,312 \}; \\
    -1 &\mu\nu\gamma \in \{132,213,321 \};\\
    0 &\text{repeated indices, or any index is 0},
    \end{cases}
\end{align}
is an extended Levi-Civita symbol.

Using the Bloch representation and setting Eq. \eqref{eq:ecu6} to zero, we get
\begin{align}\label{eq:ecu12}
  \partial_t\pi^{\text{det}}(t) =   \gamma\Tr[\sum_{a \in \mathcal{I}} p(a) A_a^\dagger A_a \sum_\mu x_\mu F_\mu]\sum_\lambda x_\lambda F_\lambda - \sum_{a \in \mathcal{I}}  \gamma p(a) \{A_a^\dagger A_a, \sum_\nu r_\nu F_\nu \} = 0.
\end{align}
We will send the negative part to the right-hand side (RHS), and we will treat it independently from the left-hand side (LHS) for the sake of order: 
\begin{subequations}
\begin{align}
    \text{LHS} &= \gamma \Tr[\sum_{a\in \mathcal{I}} p(a) A_a^\dagger A_a \sum_\mu x_\mu F_\mu]\sum_\lambda x_\lambda F_\lambda \label{eq:ecu12.1}\\
    &= \gamma \sum_\lambda \sigma_\lambda x_\lambda \sum_{a \in \mathcal{I},\mu}p(a)x_\mu \Tr[A_a^\dagger A_a \sigma_\mu]. \label{eq:ecu13}
\end{align}
\end{subequations}
Setting above the four-vector
\begin{equation}\label{eq:ecu13.1}
    B_\mu \coloneqq \sum_{a \in \mathcal{I}} p(a)\Tr[A_a^\dagger A_a F_\mu],
\end{equation}
gives
\begin{equation}\label{eq:ecu14}
    \text{LHS} = \gamma \sum_{\lambda\mu} x_\lambda x_\mu B_\mu,
\end{equation}
where LHS denotes the left-hand side.
Now, informally writing the RHS as $\text{RHS} = \sum_\mu \Tr[\text{RHS} \, F_\mu]F_\mu$, we have
\begin{subequations}
\begin{align}
    \text{RHS} &= \frac{\gamma}{2} \sum_\lambda F_\lambda \Tr[\sum_{a\in \mathcal{I} }p(a) \{A_a^\dagger A_a, \sum_\nu x_\nu F_\nu \}F_\lambda] \label{eq:ecu15}\\
    &= \frac{\gamma}{2} \sum_{a\in \mathcal{I}}\sum_{\lambda \nu}F_\lambda p(a) x_\nu \Tr[ A_a^\dagger A_a F_\nu F_\lambda + A^\dagger_a A_a F_\lambda F_\nu] \label{eq:ecu16} \\
    &= \frac{\gamma}{2\sqrt{2}}\sum_{a\in \mathcal{I}}\sum_{\lambda \nu}F_\lambda p(a) x_\nu\left( \theta_{\nu \lambda \delta} + i\epsilon_{\nu \lambda \delta} + \theta_{\lambda \nu \delta} + i\epsilon_{\lambda \nu \delta} \right) \Tr[A_a^\dagger A_a F_\delta]
\end{align}
\end{subequations}
Setting the LHS equal to the RHS and using the fact that the set $\{F_\mu \}_\mu$ is an orthonormal basis together with the identity Eq. \eqref{eq:ecu10.1}, we have for all $\mu$
\begin{equation}\label{eq:ecu19}
    x_\lambda\sum_{\mu}x_\mu B_\mu =  \sum_{\nu \delta }x_\nu \theta_{\nu \lambda \delta }B_\delta.
\end{equation}
Setting above $\lambda = 0$ and using the second identity of Eq. \eqref{eq:ecu10.1} gives the tautology $\sum_{\mu}x_\mu B_\mu = \sum_{\nu}x_\nu B_\nu$. 

Before continuing with Eq. \eqref{eq:ecu19}, if we use the first identity of Eq. \eqref{eq:ecu10.1} together with the fact that $A\cong \sigma^+$ leads to $\Tr[A_a^\dagger A_a] = \Tr[\sigma^-\sigma^+ R_aR_a^\dagger] = 1$, we get 
\begin{equation}\label{eq:ecu19.1}
    B_0 = \sum_{a \in \mathcal{I}}p(a) \Tr[A_a^\dagger A_a F_0] = \frac{1}{\sqrt{2}}\sum_{a \in \mathcal{I}}p(a)\Tr[A_a^\dagger A_a] = \frac{1}{\sqrt{2}},
\end{equation}
result that coincides with $x_0 = \Tr[\pi F_0].$ Above, recall that $R_a = R(\theta_a,\varphi_a) = \exp(-\frac{i}{2}\theta_a \sigma^z)\exp(-\frac{i}{2}\varphi_a \sigma^y)$ is the rotation operator corresponding to the $a$-th steering direction and $A_a = R_aAR_a^\dagger$.

Let us set $\lambda = a$ in Eq. \eqref{eq:ecu19}, to which we get
\begin{subequations}
\label{eqecu19.2}
\begin{align}
    x_a\sum_{\mu}x_\mu B_\mu &= \sum_{\nu \delta }\theta_{\nu \delta 0}B_\delta \\
    &= \sum_{\nu \delta}(\delta_{\nu 0}\delta_{a \delta} + \delta_{\delta 0}\delta_{a \nu})x_\nu B_\delta\\
    x_a \left(x_0 + \sum_{b}x_bB_b \right) &= x_0B_a + x_aB_0 \Rightarrow 
    x_a\sum_bx_b B_b = B_a.
\end{align}
\end{subequations}
This equation implies that the 3-vector $\bm{x} \coloneqq (x_1,x_2,x_3)$, which is proportional to the Bloch vector $\bm{r} = \Tr[\pi \bm{\sigma}] = \sqrt{2}\Tr[\pi \bm{F}]$ is proportional to the vector
\begin{equation}\label{eq:ecu19.3}
    \bm{B} = \Tr[\sum_{a\in \mathcal{I}}A_a^\dagger A_a \bm{F}] = \frac{1}{\sqrt{2}}\Tr[\sum_{a\in \mathcal{I}}A_a^\dagger A_a \bm{\sigma}].
\end{equation}
Hence, we rewrite Eq. \eqref{eqecu19.2} as 
\begin{equation}\label{eq:ecu19.4}
    (\bm{r}\cdot \bm{B})\bm{r} = \bm{B},
\end{equation}
where $\bm{r}\cdot \bm{B}$ denotes the euclidean inner product between $\bm{r}$ and $\bm{B}$. Hence, the Bloch vector is either parallel or anti-parallel to $\bm{A}$. We proceed to solve this issue.

To solve Eq. \eqref{eq:ecu19.4}, we first conclude that the stationary state of Eq. \eqref{eq:ecu6} is pure because this equation induces the non-linear Schrödinger equation
\begin{equation}\label{eq:ecu23.1}
    \dv{t} \ket{\psi(t)} =  -\frac{1}{2}\sum_{a \in \mathcal{I}} p(a) \gamma\left(A_a^\dagger A_a - \norm{A_a\psi(t)}^2 \right)\ket{\psi(t)}
\end{equation}
whenever the statistical operator $\pi(t)$ evolves deterministically starting from a pure state. In this case, $\pi(t) = \ket{\psi(t)}\bra{\psi(t)}$. Nonetheless, if the initial state is not pure, the relation between the pure states and the statistical operator is given by $\pi(t) = \mathbb{E}[\ket{\psi(t)}\bra{\psi(t)}]$ and Eq. \eqref{eq:ecu23.1} still holds. This observation about the state's purity allows us to set  $\norm{\bm{r}} = 1$ into Eq. \eqref{eq:ecu19.4} such that
\begin{equation}\label{eq:ecu24}
    r_a = \frac{B_a}{\norm*{\bm{B}}} \sec\alpha
\end{equation}
with $\sec \alpha \in \{1,-1\}$. We can determine the value of $\sec\alpha$ in two ways: we either calculate the divergence of the vector field defined in Eq. \eqref{eq:ecu6} or we only go back to the ideal protocol and set $A_a = A \cong \sigma^+$ and $p(a) = 1/n$ for all $a \in \mathcal{I}$, i.e., we set the target state to be $\bm{r} = (0,0,1)$. Hence, Eq. \eqref{eq:ecu24} becomes
\begin{equation}\label{eq:ecu25}
    r_a = \frac{ \Tr[\sum_{b \in \mathcal{I}}p(b)A_b^\dagger A_b\sigma_a]}{\norm{\Tr[\sum_{b \in \mathcal{I}}p(b)A_b^\dagger A_b\bm{\sigma}]}}\sec \alpha  =  \frac{\Tr[\ket{\downarrow}\bra{\downarrow}\sigma_a]}{\norm{\Tr[\ket{\downarrow}\bra{\downarrow}\bm{\sigma}]}} \sec \alpha =  \delta_{3,a}\sec \alpha .
\end{equation}
Thus, $\sec \alpha = -1$ and the Bloch vector associated to the \emph{stable} stationary state of Eq. \eqref{eq:ecu6} is given by
\begin{equation}
    \bm{r} = -\frac{\Tr[\sum_{b \in \mathcal{I}}p(b)A_b^\dagger A_b\bm{\sigma}]}{\norm{\Tr[\sum_{b \in \mathcal{I}}p(b)A_b^\dagger A_b\bm{\sigma}]}}.
\end{equation}

\subsubsection{Average over the clicks}
Let us now demonstrate Eq.~\eqref{eq:sde11}. By 
following similar steps as in the two previous discrete SMEs, suppose the $i$-th direction appeared and that we are not interested in the detector readout. The updated state is then [cf. Eq.~\eqref{eq:gen16}]
\begin{equation}\label{eq:gen25}
    \sigma_{k + 1} = \sum_{\alpha}p(\alpha\vert i; \sigma_k)\frac{\mathcal{M}_\alpha^{(i)}\sigma_k}{p(i,\alpha\vert \sigma_k)},
\end{equation}
where the quantum average $\mathbb{E}_\alpha$ over detector readouts is taken over the density matrix $\omega_k$.

The associated discrete SME describing all the possible steering outcomes is [cf. Eq.~\eqref{eq:gen12} and \eqref{eq:gen17}]
\begin{equation}\label{eq:gen26}
    \sigma_{k + 1} =\sum_{i} \frac{1}{p(i\vert \sigma_k)}\left( \sum_{\alpha}{\Phi_{i,\alpha}(\sigma_k)}\right)\mathds{1}_i^{k+1} = \sum_i \gamma \mathcal{D}(A_i)\sigma_k\delta t \mathds{1}_{i}^{k+1} + \mathcal{O}(\delta t^2).
\end{equation}
with 
\begin{equation}\label{eq:gen27}
    \mathbb{E}_i[\mathds{1}_i^{k+1}] = p(i\vert\sigma_k).
\end{equation}
The connection between the indicators is given by [Eq.~\eqref{eq:gen13.3}-\eqref{eq:gen13.3-1} and \eqref{eq:gen19}]
\begin{equation}\label{eq:gen28}
    \mathbb{E}_\alpha[\mathds{1}_{i,\alpha}^{k+1}] =p(\alpha\vert i;\sigma_k) \mathds{1}_{i}^{k+1}, 
\end{equation}
and so the SME is
\begin{equation}
\dd \sigma(t) = \sum_i \gamma \mathcal{D}(A_i)\sigma(t)\chi_i(t)\dd t
\end{equation}
with
\begin{equation}
    \mathbb{E}_i[\chi_i(t)] = p(i).
\end{equation}
Once more, we obtain the following LE by performing the average with respect to the directions,
\begin{equation}
    \dd \mathbb{E}_i[\sigma(t)] = \dd\rho(t) = \sum_i p(i)\gamma \mathcal{D}(A_i)\rho(t)\dd t,
\end{equation}
where $\mathbb{E}_i[\sigma(t)] = \left( \mathbb{E}_i\circ \mathbb{E}_\alpha\right)[\omega(t)] = \rho(t)$.

\section{Equivalence between two formalisms}\label{sec:several_white_noises}
In this section, we show how the SME~\eqref{eq:ssde5} can be obtained from a different model of repeated, random interactions.

Consider the detector-system interaction Hamiltonian in $\mathcal{H}_d\otimes \mathcal{H}_s$
\begin{equation}\label{eq:u1}
    H_\text{ds}(t\vert \{\xi_i \}) = \sum_{i=1}^N \sqrt{\gamma p(i)} \xi_i(t)h_0^{(i)},
\end{equation}
where
\begin{align}\label{eq:u2}
    h_0^{(i)} = \ket*{\Phi_d^\perp}\bra{\Phi_d}\otimes A_i + \text{h.c.},
\end{align}
is an operator corresponding to the steering toward the \emph{rotated} target state $\ket*{\Psi_\oplus^{(i)}} = R(\theta_i, \varphi_i)\ket{\Psi_\oplus}$, i.e., $A_i\ket*{\Psi_\oplus^{(i)}} = 0$, and
\begin{equation}\label{eq:u3}
    \mathbb{E}[\xi_i(t)] = 0, \quad \mathbb{E}[\xi_i(t)\xi_j(s)] = \delta_{ij}\delta(t-s), \quad \forall i,j = 1,2,3,
\end{equation}
are delta-correlated white noises.
When using the unitary operator generated by Eq.~\eqref{eq:u1}, we will use the It\^o formalism with the table
\begin{equation}\label{eq:u4}
\dd X_i \dd X_j(t) = \delta_{ij}\dd t, \quad (\dd t)^2 = 0, \quad \dd X_i(t) \dd t = 0, \quad \forall i,j = 1,2,3,
\end{equation}
where $\dd X_i = \xi_i(t)\dd t$ is a Wiener increment.

Let us then consider the following unitary operator using the above rules
\begin{equation}\label{eq:u5}
    U(\dd t) = \exp(-i\sum_{i=1}^n \sqrt{\gamma p(i)}h_0^{(i)}\dd X_i) = I -i\sum_{i=1}^n \sqrt{\gamma p(i)}h_0^{(i)}\dd X_i - \frac{1}{2}\sum_{i=1}^n\gamma p(i) \left(h_0^{(i)}\right)^2\dd t.
\end{equation}
As the initial density operator is $\rho_d\otimes \rho(t)$, the associated measurement operators are
\begin{align}
    M_0(\dd t) &= I - \sum_{i=1}^n \frac{\gamma p(i)}{2}A_i^\dagger A_i \dd t \quad \text{and} \quad 
    M_1(\dd t) = - i\sum_{i=1}^n\sqrt{\gamma p(i)}A_i \dd X_i.  \label{eq:u7}
\end{align}
Thus, the two possible operations over $\rho(t)$ upon a detector readout are given by
\begin{align}\label{eq:u8}
    \rho_{0}(t+\dd t ) &= \frac{M_0(\dd t) \rho(t) M_0^\dagger (\dd t )}{\expval{M_0^\dagger(\dd t) M_0(\dd t)}_t}  = \rho(t) - \sum_{i=1}^n \frac{\gamma p(i)}{2}\{A_i^\dagger A_i - \expval*{A_i^\dagger A_i}_t, \rho(t) \}\dd t, \\
    \rho_{1}(t+\dd t ) &= \frac{M_1(\dd t) \rho(t) M_1^\dagger (\dd t )}{\expval{M_1^\dagger(\dd t) M_1(\dd t)}_t} = \frac{\sum_{i=1}^n p(i) A_i\rho(t) A_i^\dagger}{\expval*{\sum_{j=1}^np(j)A_j^\dagger A_j}_t}.
\end{align}
The probabilities of obtaining these two states are
\begin{align}\label{eq:u9}
    P(\alpha = 0) &= 1 - \gamma\expval{\sum_{i=1}^n  p(i) A_i^\dagger A_i}_t\dd t, \quad \text{and} \quad 
    P(\alpha = 1) = \gamma\expval{\sum_{i=1}^n p(i) A_i^\dagger A_i}_t\dd t.
    \end{align}
Note that the outcome $\rho_1(t+\dd t)$ is Poissionian.  We are thus set to write the SME governing the dynamics of $\rho(t):$
\begin{equation}
    \dd \rho(t) = \sum_{i=1}^n \gamma \mathcal{D}(A_i)\rho(t)\dd t + \left( \frac{\sum_{i=1}^n p(i) A_i\rho(t) A_i^\dagger}{\expval*{\sum_{j=1}^np(j)A_j^\dagger A_j}_t} - \rho(t)\right)\left( \dd N(t) - \gamma \expval{\sum_{i=1}^n p(i)A_i^\dagger A_i}_t\dd t \right).
\end{equation}
This equation has the same form as Eq. \eqref{eq:ssde5}, yet it was obtained from a different dynamic error.

\section{Two continuous steering directions}\label{sec:two_continuous_steering_directions}
This section aims to show that the protocol is also robust against two continuous steering directions: the continuous distribution and the von Mises distribution  \cite{gatto2007generalized} and compare the results to those obtained in Sec.~\ref{ssec:stationary_state_of_the_fully_averaged}.

The set of steering states (or directions) is now denoted by
\begin{equation}\label{eq:cont_dist}
\mathcal{R} \coloneqq \{ (\omega(\theta); p(\theta,\lambda) \},
\end{equation}
where $p(\theta,\lambda)$ is a probability distribution over the circle characterized by the parameters $\lambda$ [cf. Eq.~\eqref{eq:dir}]. The corresponding LE must be integrated with respect to the probability distribution [cf. Eq.~\eqref{eq:field0}], i.e., 
\begin{equation}\label{eq:cont_dist1}
    \partial_t\rho_s(t) = \int \dd \theta p(\theta,\lambda) \mathcal{D}\textbf{(}A(\theta)\textbf{)}\rho_s(t).
\end{equation}

With the above definitions at hand, let us study the continuous distribution between defined in $[-\tilde{\theta},\tilde{\theta}]$
\begin{equation}\label{eq:cont_dist2}
    p(\theta,\tilde{\theta})\coloneqq \frac{1}{2\tilde{\theta}}\left[ \Theta(\theta + \tilde{\theta})- \Theta(\theta -\tilde{\theta}) \right],
\end{equation}
where $\Theta$ is the Heaviside step function Eq.~\eqref{eq:op1.20},
and the von Mises distribution \cite{gatto2007generalized}
\begin{equation}\label{eq:cont_dist3}
    p(\theta,\sigma) \coloneqq \frac{\exp(\sigma^{-2}\cos\theta  )}{2\pi I_0(\sigma^{-2})},
\end{equation}
where $\sigma^2$ is the variance, and $I_\nu$ is the modified Bessel function of order $\nu$.

The relevant entries of the stationary states corresponding to each distribution are
\begin{align}
[\tilde\rho_{\oplus}(\tilde{\theta})]_{11} &= \frac{1}{2} + \frac{4\sin\tilde{\theta}}{6\tilde{\theta} + \sin2\tilde{\theta}},\quad  [\tilde\rho_{\oplus}(\tilde{\theta})]_{12} = 0;\\
\left[\tilde\rho_{\oplus}(\sigma)\right]_{11} &= \frac{1}{2} + \frac{I_1(\sigma^{-2})}{I_0(\sigma^{-2}) - \theta^2I_1(\sigma^{-2})}, \quad \left[\tilde\rho_{\oplus}(\sigma)\right]_{12} = 0.  
\end{align}
The fidelities concerning the above two states with the ideal target state $\rho_\oplus = \ket{\uparrow}\bra{\uparrow}$ are
\begin{align}
    F_{\infty}(\tilde{\theta}) &= \frac{1}{2} + \frac{4\sin\tilde{\theta}}{6\tilde{\theta} + \sin2\tilde{\theta}} \approx 1- \frac{\tilde{\theta}^4}{80}, \label{eq:ui1}\\
    F_{\infty}(\sigma) &= \frac{1}{2} + \left( \frac{2I_0(\sigma^{-2})}{I_1(\sigma^{-2})} - \sigma^2 \right)^{-1} \approx 1-  \frac{3}{16}\sigma^4.\label{eq:ui2} 
\end{align}
In the small error approximation, i.e., for $\sigma, \tilde{\theta} \ll 1$, we see once more that the leading terms are of order four as in Eq.~\eqref{eq:td2}, $F_\infty(p=1/2,\theta) \approx 1 - \theta^4/16$.  Notably, the coefficients multiplying the small parameters are much smaller for the two continuous distributions than for the discrete, symmetric distribution. Additionally, the continuously distributed steering directions present the fidelity closest to unity, which indicates a higher degree of robustness.

\section{Commutation of the expectation value with the time-ordering operator and partial trace.}\label{sec:commutation_expectation}
As the time ordering operator and the expectation are linear operators acting on different mathematical objects, they automatically commute---This is what we used in Sec.~\ref{ssec:steering_hamiltonian}. However, proving this is involved, as we now show.
\begin{theorem}\label{th:teorema1}
For a multiplicative white noise, the time-ordering operator $\overrightarrow{\mathcal{T}}$ commutes with the expectation value with respect to the white noise $\mathbb{E}$.
\end{theorem}
\begin{proof}
In what follows, every integral must be understood in the Stratonovich form \cite{kuo2006stochastic}. Let $f(s_i) \coloneqq \ad\textbf{(}\hat{H}(s_i)\textbf{)}$ for $i \in \mathcal{I}_n\coloneqq \{1,\ldots, n \}$ with $n \in 2\mathbb{Z}$. $\hat{H}(s)$ is some time-dependent Hamiltonian. Let $\mathcal{P}(\mathcal{I}_n)$ be the set of all partitions on $\mathcal{I}_n$. The expectation value of the product of an \textit{even} number of white noise variables is
\begin{equation}\label{eq:op1.1}
\mathbb{E}\left[\prod_{i=1}^n\xi(s_i)\right] = \sum_{\pi \in \mathcal{P}(\mathcal{I}_n)}\prod_{\{i,j \} \in \pi}\delta(s_i- s_j) \quad n \in 2\mathbb{Z}^+.
    \end{equation}
The product of an odd number is zero. Moreover, we will make use of the Heaviside function adapted for the Stratonovich integral
\begin{equation}\label{eq:op1.20}
    \Theta(t) \coloneqq \begin{cases}
    1 & t > 0, \\
    1/2 & t = 0, \\
    0 & t < 0.
    \end{cases}
\end{equation}
Let us treat the adjoint version of the unitary operator
$\hat{U}(t\vert \xi) \coloneqq \overrightarrow{\mathcal{T}}\exp[-i\int_0^t \hat{H}(s)\xi(s) \dd s]$
given by
$\mathcal{U}(t\vert \xi) \coloneqq \overrightarrow{\mathcal{T}}\exp[-i\int_0^t f(s)\xi(s) \dd s]$, and take its expectation:
\begin{subequations}
\begin{align}
  \hat{\mathcal{U}}(t\vert \xi) &\coloneqq   \mathbb{E}\left[  \overrightarrow{\mathcal{T}}\exp(-i\int_0^t f(s)\xi(s)\dd s)   \right]  \label{eq:op1.2}\\
  &= \mathbb{E}\left[ \sum_{n=0}^\infty (-i)^n \int_0^t \dd s_1 f(s_1)\xi(s_1)\dotsi \int_0^{s_{n-1}}\dd s_n f(s_n)\xi(s_n) \right] \label{eq:op1.3} \\
  &= \sum_{n=0}^\infty (-i)^n \int_{\mathbb{R}^n}\prod_{i=1}^n\dd s_i \Theta(s_i)f(s_i)\Theta(t-s_1)\prod_{l=1}^ {n-1}\Theta(s_l - s_{l+1}) \mathbb{E}\left[\prod_{i=1}^n\xi(s_i)\right] \label{eq:op1.4}\\
  &= \sum_{k = 0}^\infty(-1)^k \sum_{\pi \in \mathcal{P}(\mathcal{I}_n)} \int_{\mathbb{R}^{2k}}\prod_{i=1}^{2k}\dd s_i \Theta(s_i)f(s_i)\Theta(t-s_1)\prod_{l=1}^{2k-1}\Theta(s_l - s_{l+1})\prod_{\{i,j \} \in \pi }\delta(s_i - s_j) \label{eq:op1.5} \\
  &= \sum_{ k = 0}^\infty (-1)^k\int_{\mathbb{R}^k}\prod_{i = 1}^k\dd s_i \Theta(s_i)f^2(s_i) \Theta(t - s_1)\Theta^k(0)  \label{eq:op1.6}\\
  &= \sum_{k=0}^\infty \frac{(-1)^k}{2^k}\int_0^t \dd s_1 f^2(s_1) \dotsi \int_{0}^{s_{k-1}}\dd s_kf^2(s_k) \label{eq:op1.7}\\
  &= \overrightarrow{\mathcal{T}}\exp(-\frac{1}{2}\int_0^t \dd s f^2(s))\label{eq:op1.70} \\
  &= \overrightarrow{\mathcal{T}}\exp[-\frac{1}{2}\int_0^t \ad^2\textbf{(}\hat{H}_1(s)\textbf{)}\dd s].
  \end{align}
  \end{subequations}
From Eq.~\eqref{eq:op1.3} to \eqref{eq:op1.4}, we used the Heaviside function to rewrite the integrals and thus be able to integrate the Dirac delta; from Eq.~\eqref{eq:op1.4} to \eqref{eq:op1.5} we used the fact that the expectation value over the product of the white noise variables is only non-zero for an even number of them [see Eq.~\eqref{eq:op1.1}]; from Eq. \eqref{eq:op1.5} to \eqref{eq:op1.7}, we used the fact that only partitions having pairs of the form $ \{1, 2 \}, \ldots \{i, i+1 \}, \ldots, \{k-1, k \}$ make the integrand different than zero because Heaviside functions of the form $\Theta(s_i- s_j)\Theta(s_j-s_i)$  are zero for $i\neq j$.

Let us now take the following turn by evaluating the following superoperator
\begin{subequations}
\begin{align}
   \hat{\tilde{\mathcal{U}}}(t\vert \xi) &\coloneqq   \overrightarrow{\mathcal{T}}\mathbb{E}\left[\exp(-i\int_0^t f(s)\xi(s)\dd s)\right] \label{eq:op1.8}\\
    &= \overrightarrow{\mathcal{T}} \exp( -\frac{1}{2}\mathbb{E}\left[ \int_0^tf(s)\xi(s)\dd s \, \int_0^t f(u)\xi(u)\dd u \right]) \label{eq:op1.9} \\
    &= \overrightarrow{\mathcal{T}} \exp( -\frac{1}{2}\int_0^tf^2(s)\dd s ). \label{eq:op1.10}
\end{align}
\end{subequations}
From Eq. \eqref{eq:op1.8} to Eq. \eqref{eq:op1.9}, we applied the well-known identity for Gaussian stochastic integrals \cite{kuo2006stochastic}. From the above equation, we conclude that $\hat{\mathcal{U}} = \hat{\tilde{\mathcal{U}}}$, (cf. Eq. \eqref{eq:op1.70}), which in turn, implies that $\overrightarrow{\mathcal{T}}\,  \mathbb{E} = \mathbb{E} \,\overrightarrow{\mathcal{T}}$.
\end{proof}

\section{The partial trace commutes with the expectation value}\label{sec:lastproof}
Let $\rho_{\text{ds}}(0) = \rho_d\otimes \rho_s(0)$. Let us consider $\hat{\zeta}_{\text{ds}}(t\vert\xi) \coloneqq \hat{U}(t \vert \xi)\rho_{ds}(0)\hat{U}^\dagger(t\vert \xi),$ where 
\begin{equation}
     \hat{U}(t\vert \xi) \coloneqq \overrightarrow{\mathcal{T}}\exp[-i\int_0^t\hat{h}_1(s) \dd s]
\end{equation}
is the time-evolution operator in the interaction picture with respect to the perturbed detector--system Hamiltonian $H(t) = Jh_0 + \xi(t)h_1$, and  $\xi(t)$ is the delta-correlated white noise. Let $\hat{\rho}_s(t\vert\xi)\coloneqq \Tr_d\hat{\zeta}_{\text{ds}}(t\vert\xi)$, $\bar{\hat{\rho}}_{\text{ds}}(t) \coloneqq \mathbb{E}[\hat{\zeta}_{\text{ds}}(t\vert\xi)]$, and $\bar{\hat{\rho}}_s(t) \coloneqq \mathbb{E}[\hat{\rho}_s(t\vert\xi)]$. We will show that $\bar{\hat{\rho}}_s(t) = \Tr_d\bar{\hat{\rho}}(t)$ (i.e,   $\Tr_d\circ \mathbb{E} = \mathbb{E}\circ \Tr_d$). In other words, we will show that  we can either take the average over realizations of the stochastic noise up to time $t$ and then perform a blind measurement or that we can perform a blind measurement at time $t$ over the \textit{same} stochastic trajectory and \textit{then} take the average over the realizations of the stochastic variable. 
\begin{theorem}\label{th:teoremazo}
For a multiplicative white noise, the partial trace commutes with the expectation value with respect to the white noise $\mathbb{E}$.
\end{theorem}
     \begin{proof}
     Let us denote
     \begin{equation}\label{eq:rco}
         f^{(n_1)}(s_1)f^{(n_{2})}(s_2)\dotsi f^{(n_{i})}(s_i) \coloneqq \Tr_d\left\lbrace \ad^{(n_{1})}\textbf{(}\hat{h}(s_1)\textbf{)}\ad^{(n_{2})}\textbf{(}\hat{h}(s_{2})\textbf{)}\dotsi \ad^{(n_i)}\textbf{(}\hat{h}(s_i)\textbf{)}\,\rho_{\text{ds}}(0) \right\rbrace,
     \end{equation}
     where $n_i \in \mathbb{N}$ for all $i\in \mathcal{I}_n \coloneqq \{1,\ldots, n \}$. Consider
     \begin{subequations}
     \begin{align}
         \bar{\hat{\rho}}_s(t) &\coloneqq  \mathbb{E}\left[ \Tr_d \hat{\zeta}_{ds}(t\vert\xi) \right]  \\
         &= \mathbb{E}\left[ \Tr_d\sum_{n=0}^\infty (-i)^n \int_0^t \dd s_1 \ad\textbf{(}\hat{h}(s_1)\textbf{)}\xi(s_1)\dotsi \int_0^{s_{n-1}}\dd s_n \ad\textbf{(}\hat{h}(s_n)\textbf{)}\xi(s_n) \rho_{\text{ds}}(0)\right]\\
         &= \mathbb{E}\left[ \sum_{n=0}^\infty (-i)^n \int_0^t \dd s_1 f(s_1)\xi(s_1)\dotsi \int_0^{s_{n-1}}\dd s_n f(s_n)\xi(s_n) \right].
     \end{align}
     \end{subequations}
     Following the same steps as in the proof of Theorem \ref{th:teorema1} and recalling Eq. (\ref{eq:rco}), we arrive at 
     \begin{subequations}
     \begin{align}
          \bar{\hat{\rho}}_s(t) &=  \sum_{k=0}^\infty \frac{(-1)^k}{2^k}\int_0^t \dd s_1 \dotsi \int_{0}^{s_{k-1}}\dd s_kf^2(s_1) \dotsi f^2(s_k)\rho_s(0)\\
          &= \sum_{k=0}^\infty \frac{(-1)^k}{2^k}\int_0^t \dd s_1  \dotsi \int_{0}^{s_{k-1}}\dd s_k \Tr_d\left[\ad^2\textbf{(}\hat{h}(s_1)\textbf{)}\dotsi \ad^2\textbf{(}\hat{h}(s_k)\textbf{)}\rho_{\text{ds}}(0)\right] \\
          &= \Tr_d\left[\overrightarrow{\mathcal{T}}\exp(-\frac{1}{2}\int_0^t \ad^2\textbf{(}\hat{h}(s)\textbf{)}\dd s)\rho(0)\right]  \\
          &= \Tr_d \bar{\hat{\rho}}_{\text{ds}}(t).
     \end{align}
     \end{subequations}
     The last equation implies that $\mathbb{E} \circ \Tr_d = \Tr_d \circ \mathbb{E}$.
     \end{proof}
     
\section{Lindblad equation: errors in the steering Hamiltonian}\label{sec:LE_errors_hamitlonian}
In this appendix, we obtain the LE~\eqref{eq:c13}, which describes the full averaged dynamics of a steered system obtained from the direct averaging unraveling. 

Upon averaging over the white noise, the LE describing the evolution of the detector-system density matrix is 
\begin{equation}\label{eq:apc11}
\partial_t\rho_{\text{ds}}(t) = \left( -iJ\ad(h_0) - \frac{\tilde{\gamma}}{2}\ad^2( \tilde{h}_{\text{ds}})  \right)\rho_{\text{ds}}(t) = \mathcal{L}\rho_{\text{ds}}(t).
\end{equation}
Hence, the formal solution of this equation is $\rho_{\text{ds}}(t) = \exp(\mathcal{L}t)\rho_{\text{ds}}(0)$, where $\rho_{\text{ds}}(0) = \rho_d\otimes \rho_s(t)$, and the state of the steered system is obtained by just performing the partial trace, i.e., $\rho_s(t) = \Tr_d\left[\exp(\mathcal{L}t)\rho_d\otimes \rho_s(0) \right]$. After another interaction with a new detector, the updated state is given by 
\begin{subequations}
\label{eqyy10}
\begin{align}\label{eq:yy10}
{\rho}_{\text{ds}}(t+ \delta t) &= \exp(\mathcal{L}\delta t)[\rho_d\otimes {\rho}_s(t)] \\
  &= \rho_d\otimes {\rho}_s(t)  + \delta t\left[ -iJ\ad(h_0) - \frac{\tilde{\gamma}}{2}\ad^2(\tilde{h}_{\text{ds}}) +
     \frac{\delta t}{2}\left( -iJ\ad(h_0) - \frac{\tilde{\gamma}}{2}\ad^2(\tilde{h}_{\text{ds}}) \right)^2   \right]\rho_d\otimes {\rho}_s(t) \notag \\  &\,\,\,\,\,\,+ \mathcal{O}(\delta t^3).
\end{align}
\end{subequations}
Similarly as before (see Appendix. \ref{sec:A1}), we  choose the following decomposition of the Hamiltonian operators and the density matrix of the detectors:
\begin{equation}\label{eq:yy11}
    h_0 = \begin{pmatrix}
    0 & A^\dagger \\
    A & 0
    \end{pmatrix}, \quad
    \tilde{h}_{\text{ds}} = \begin{pmatrix}
    G& B^\dagger \\
    B & C
    \end{pmatrix},  \quad \text{and} \quad
        \rho_d = \begin{pmatrix}
    1 & 0 \\
    0 & 0
    \end{pmatrix},
\end{equation}\label{eq:yy111}
with $G = G^\dagger$ and $C = C^\dagger$.
With this decomposition, all the relevant terms in Eq.~\eqref{eqyy10} are
     \begin{align}
        \ad( h_0)\rho_d\otimes {\rho}_s &= \begin{pmatrix} 0 & -{\rho}_sA^\dagger \\ A\rho_s & 0  \end{pmatrix}, \label{eq:adjunto1} \\
        \ad(\Tilde{h}_{\text{ds}})\rho_d\otimes \rho_s &= \begin{pmatrix} \ad(G) \rho_s & -\rho_sB^\dagger \\ B\rho_s & 0  \end{pmatrix},
           \\
           \ad^2(h_0)\rho_d\otimes \rho_s &= \begin{pmatrix} 
         \{G^\dagger G, \rho_s \} & 0 \\ 
         0 & -2G\rho_sG^\dagger
         \end{pmatrix},\\
         \ad^2(\Tilde{h}_{\text{ds}}) \rho_d\otimes \rho_s &= \begin{pmatrix} \ad^2(G)(\rho_s) + \{B^\dagger B, \rho_s \} & -\ad(G)\rho_sB^\dagger - G\rho_sB^\dagger +\rho_sB^\dagger C \\ B\,\ad(G)\rho_s - B\rho_sG + CB\rho_s & -2B\rho_sB^\dagger \end{pmatrix}.
    \end{align}
      We replace the above terms into Eq.~\eqref{eq:yy10} and take the partial trace over the detectors:
    \begin{align}\label{eq:yy12}
        \rho_s( t+ \delta t) &= \rho_s(t) - \frac{\delta t \tilde{\gamma}}{2}\left[ \ad^2(G)\rho_s(t) + \{B^\dagger B, \rho_s(t) \} - 2B \rho_s(t) B^\dagger\right] \notag \\    &\,\,\,\,\,\,+ \frac{J^2\delta t^2}{2}\left[ -2A\rho_s(t) A^\dagger + \{A^\dagger A, \rho_s(t)\} \right] \notag  \notag \\ &\,\,\,\,\,\,+  \frac{\delta t^2 \tilde{\gamma}}{2}\Tr_d\left[\frac{\tilde{\gamma}}{2}\ad^4(\tilde{h}_{\text{ds}})\rho_d\otimes \rho_s(t) - i\frac{J}{2}\{\ad(h_0), \ad^2(\tilde{h}_{\text{ds}}) \}\rho_d\otimes \rho_s(t)\right]  +\mathcal{O}(\delta t^3).
    \end{align}
The term containing the trace in the above equation will vanish in the WM limit, which is our next step. 
 Subtracting  $\rho_s(t)$ above, dividing by $\delta t$, and taking the limit $\delta t \rightarrow 0$ while keeping  $\gamma \coloneqq J^2\delta t$ fixed, gives Eq.~\eqref{eq:c13}.

\section{Stochastic master equation of the jump-diffusive-type I}\label{sec:jump-diffussive-sme}
In this appendix, we derive the SME in Eq.~\eqref{eq:esto1}. 

Let $\xi(t)$ be a white noise with normalization conditions 
\begin{equation}\label{eq:estoc1}
    \mathbb{E}[\xi(t)] = 0, \quad \mathbb{E}[\xi(t)\xi(s)] = \delta(t-s).
\end{equation}
Let $\delta X \coloneqq \int_0^{\delta t}\xi(s)\dd s = \xi(\delta t)\delta t$ be a Wiener increment on $\mathbb{R}$ with zero mean and variance $\delta t$. Following Ref. \cite{onorati2017mixing},
\begin{equation}\label{eq:estoc2}
    U(t) = \lim_{\delta t \to 0}\prod_{l = t/\delta t}^1\exp[Y(l\delta t) - Y\textbf{(}(l-1)\delta t\textbf{)}]U(0)
\end{equation}
is a stochastic process describing Brownian motion on the Lie group $U(4)$ and $Y(t)$ corresponds to Brownian motion on the Lie algebra $\mathfrak{u}(4)$. Above, $U(0) = I_{\text{ds}}$.  

The increments of the process $Y(t)$ can be related to the so-called \emph{Hamiltonian increments}
\begin{equation}\label{eq:estoc3}
    H(\delta t) \coloneqq i\Theta(\delta t),
\end{equation}
where 
\begin{equation}\label{eq:estoc4}
    \Theta(\delta t) \coloneqq \frac{1}{\delta t}\left[Y(l\delta t) - Y\textbf{(}(l-1)\delta t\textbf{)} \right]
\end{equation}
are the increments in $\mathfrak{u}(4)$. In our case, since the perturbed detector--system Hamiltonian is given by $H(t) = \sqrt{\frac{\gamma}{\delta t}}h_0 \,\,\,+ \xi(t)\sqrt{\tilde{\gamma}}\tilde{h}_{ds}$, we have 
\begin{equation}\label{eq:estoc5}
    \Theta(\delta t) = -\sqrt{\frac{\gamma}{\delta t}}h_0 - i\xi(\delta t) \sqrt{\tilde{\gamma}}\tilde{h}_{\text{ds}},
\end{equation}
which in turn gives that
\begin{equation}\label{eq:estoc6}
    U(l\delta t) - U\textbf{(}(l-1)\delta t\textbf{)} = -i\sqrt{\gamma \delta t}h_0 - i\sqrt{\tilde{\gamma}}\delta X \tilde{h}_{\text{ds}}.
\end{equation}
Replacing the above increment in Eq.~\eqref{eq:estoc2} allows us to write 
\begin{equation}\label{eq:estoc7}
    U(\delta t) = \exp[\Theta(\delta t)\delta t] = \exp(-i\sqrt{\gamma \delta t}h_0 - i\sqrt{\Tilde{\gamma}}\delta X \tilde{h}_{\text{ds}}),
\end{equation}
from which we will obtain the measurement operators $M_\alpha(\delta t) = \bra{\alpha}U(\delta t) \ket{0}$ for $\alpha \in \{0,1\}$.

Before performing a series expansion of $U(\delta t)$, we note that even though the product $\sqrt{\delta t}\delta X$ is of order $\delta t$, it has zero mean and variance $(\delta t)^2$, so we can neglect this product (in the It\^o sense) when compared with terms such as $\delta t$ and $\delta X$. With this observation, we have up to order $\delta t$
\begin{equation}\label{eq:estoc8}
    U(\delta t) = I_{ds} - i\sqrt{\gamma \delta t}h_0 - i\sqrt{\tilde{\gamma}}\delta X \tilde{h}_{\text{ds}} - \frac{\gamma \delta t h_0^2}{2} - \frac{\tilde{\gamma} \delta t \tilde{h}_{\text{ds}}^2}{2}.
\end{equation}
Using Eq.~\eqref{eq:yy11}, we get the following measurement operators
\begin{align}\label{eq:estoc9}
    M_0(\delta t) &= I_s - i\delta X \sqrt{\tilde{\gamma}} G - \frac{\gamma \delta  t}{2}A^\dagger A - \frac{\tilde{\gamma}^2\delta t}{2}(G^2 + B^\dagger B), \\
    M_1(\delta  t) &= -i\sqrt{\gamma \delta t} A - i\sqrt{\tilde{\gamma}}\delta X B.
\end{align}

For a given prior state, $\omega_s(t)$, the two possible states corresponding to a no-click and a click result, respectively, are
\begin{subequations}
\begin{align}
    \omega_{s,0}(t +\delta t) &=\frac{ M_0(\delta t) \omega_s(t) M_0^\dagger (\delta t)}{\Tr[M_0^\dagger (\delta t)M_0 (\delta t)]}\\
    &= \left\lbrace1 + \Tr\left[\gamma A^\dagger  A \omega_s(t)+ B^\dagger B \omega_s(t)\right]\delta t \right\rbrace\omega_s(t) - i\dd X_t\sqrt{\tilde{\gamma}}[G,\omega_s(t)] + \tilde{\gamma}\mathcal{D}(G)\omega_s(t) \delta t \notag \\ &\,\,\,\,\,\,- \frac{1}{2}\left\lbrace \gamma A^\dagger A + \tilde{\gamma} B^\dagger B, \omega_s(t) \right\rbrace\delta t, 
\end{align}
\end{subequations}
and
\begin{subequations}
    \begin{align}
         \omega_{s,1}(t +\delta t) &=\frac{ M_1(\delta t) \omega_s(t) M_1^\dagger (\delta t)}{\Tr[M_1^\dagger (\delta t)M_1 (\delta t)]}= \frac{\gamma A\omega_s(t)A^\dagger + \tilde{\gamma} B\omega_s(t)B^\dagger }{\Tr\left[\gamma A^\dagger  A \omega_s(t)+\tilde{\gamma} B^\dagger B \omega_s(t)\right]}.
    \end{align}
\end{subequations}

The time-dependent probability of a click result to occur is again too small compared with the one of a no-click result, i.e., $\Tr[M_1^\dagger (\delta t)M_1(\delta t)\omega_s(t)] \propto \delta t$. Thus, we can regard this event to be registered by an inhomogeneous Poissonian process $N(t)$, and we can capture it with the continuous one (the no-click result) in a single SME:
\begin{subequations}
\begin{align}
\dd \omega_s(t) &= \left[\omega_{s,0}(t +\delta t) - \omega_s(t) \right][1-\dd N(t)] +  \left[  \omega_{s,1}(t +\delta t) - \omega_s(t) \right]\dd N(t) \notag \\
&=-i\sqrt{\tilde{\gamma}}[G, \omega_s(t) ]\dd X_t + \tilde{\gamma}\mathcal{D}(G)\omega_s(t)\dd t + \gamma\mathcal{D}(A)\omega_s(t)\dd t +  \tilde{\gamma}\mathcal{D}(B)\omega_s(t)\dd t \notag \\ &\,\,\,\,\,+  \left[\frac{\gamma A\omega_s(t)A^\dagger + \tilde{\gamma} B\omega_s(t)B^\dagger }{\Tr\left[(\gamma A^\dagger  A+ \tilde{\gamma} BB^\dagger) \omega_s(t)\right]} - \omega_s(t) \right] \left( \dd N(t) -  \Tr\left[(\gamma A^\dagger A + \tilde{\gamma} B^\dagger B) \omega_s(t)\right]\dd t\right),
\end{align}
\end{subequations}
where we have set $\delta t = \dd t$ and $\dd N(t) = N(t + \dd t) - N(t)$, and similarly for $\dd \omega_s(t)$. The above equation is precisely Eq.~\eqref{eq:esto1}.

After taking the mean value over the clicks, i.e.,  $\mathbb{E}[\dd N(t)] = \Tr[(\gamma A^\dagger A + \tilde{\gamma} B^\dagger B)\omega_s(t) ]\dd t$, and then taking the average over trajectories and setting $\rho_s(t) \coloneqq \mathbb{E}[\omega_s(t)]$, we get Eq.~\eqref{eq:c13}.

\section{Stochastic master equation of the jump-diffusive-type II}\label{sec:jump_diffusive_type}

In this section, we demonstrate Eq.~\eqref{eq:pb6}.

We start with the general discrete stochastic master equation
\begin{equation}\label{eq:apo1}
    \omega_{k+1} = \sum_{i \in \mathcal{I}}\sum_{\alpha \in \{0,1\} }\frac{\mathcal{M}_\alpha^{(i)}\omega_k}{p(i,\alpha\vert \omega_k)}\mathds{1}^{k+1}_{i,\alpha},
\end{equation}
where $I =\{1,2\}$. Index $i = 1$ indicates a measurement in the correct basis $\mathcal{B}_1 = \{\ket*{\psi_0^{(1)}} = \ket{0}= (1,0), \ket*{\psi_1^{(1)}} =  (0,1)\}$, appearing with probability $p(1)$. Index $i=2$ corresponds to the basis $\mathcal{B}_2$ with kets
\begin{equation}
    \ket*{\psi_0^{(2)}} \coloneqq \frac{1}{\sqrt{2}}\left(\ket{0} + \ket{1} \right),\quad 
    \ket*{\psi_1^{(2)}} \coloneqq \frac{1}{\sqrt{2}}\left(\ket{0} - \ket{1} \right).
\end{equation}
The operations on a prior state $\omega_k$ are represented by
\begin{align}\label{eq:apo2}
\mathcal{M}_\alpha^{(i)}\omega_k &\coloneqq p(i)M_\alpha^{(i)}(\delta t)\omega_k M_\alpha^{(i)}(\delta t)^\dagger
\end{align}
with the corresponding measurement operators
\begin{equation}\label{eq:apo3}
M_\alpha^{(i)}(\delta t) \coloneqq \bra*{\psi_\alpha^{(i)}}\exp(-i \sqrt{\gamma \delta t}h_0)\ket{0}.    
\end{equation}

Let us define the new Lindblad operator $L \coloneqq i\sqrt{\gamma }A$. If the basis $\mathcal{B}_1$ appears, we will have a discrete SME of the jump-type again. Thus, we can write its contribution in Eq. \eqref{eq:apo1} as
\begin{equation}\label{eq:apo4}
    \omega_{k+1} = \left( \omega_k - \frac{\delta t}{2}\{L^\dagger L,\omega_k \} \right)\left( 1 + \delta t\expval*{L^\dagger L}_k \right)\mathds{1}_{1,0}^{k+1} + \frac{L\omega_kL^\dagger}{\expval*{L^\dagger L}_k}\mathds{1}_{1,1}^{k+1} + \sum_{\alpha}\frac{\mathcal{M}_\alpha^{(2)}\omega_k}{p(2,\alpha\vert \omega_k)}\mathds{1}^{k+1}_{2,\alpha}.
\end{equation}

Turning to the erroneous basis, the corresponding measurement operators are
\begin{equation}
    M_{\alpha = 0}^{(2)}(\delta t) = \frac{1}{\sqrt{2}}\left(I_s - \frac{\delta t}{2}L^\dagger L - \sqrt{\delta t}L \right), \quad \text{and} \quad
    M_{\alpha = 1}^{(2)}(\delta t) = \frac{1}{\sqrt{2}}\left(I_s - \frac{\delta t}{2}L^\dagger L + \sqrt{\delta t}L \right),
\end{equation}
which induce the following operations over the prior state $\omega_k$:
\begin{align}
    \mathcal{M}_{0}^{(2)}\omega_k &= \frac{p(2)}{2}\left[ \omega_k - \frac{\delta t}{2}\{L^\dagger L, \omega_k \} + \delta t L \omega_k L^\dagger - \sqrt{\delta t}\left( \omega_k L^\dagger + L\omega_k \right) \right], \\
     \mathcal{M}_{1}^{(2)}\omega_k &=\frac{p(2)}{2}\left[ \omega_k - \frac{\delta t}{2}\{L^\dagger L, \omega_k \} + \delta t L \omega_k L^\dagger + \sqrt{\delta t}\left( \omega_k L^\dagger + L\omega_k \right) \right].
\end{align}
Respectively, the probability of each operation is
\begin{align}
    \mathbb{E}[\mathds{1}_{2,0}^{k+1}] &= p(2,0\vert \omega_k) = \frac{p(2)}{2}\left( 1 - \sqrt{\delta t}\expval*{L^\dagger + L}_k \right), \\
    \mathbb{E}[\mathds{1}_{2,1}^{k+1}]&= p(2,1\vert \omega_k) = \frac{p(2)}{2}\left( 1 + \sqrt{\delta t}\expval*{L^\dagger + L}_k \right).
\end{align}
Replacing the above probabilities and the above operations into the diffusive part of Eq.~\eqref{eq:apo4} gives, after some calculation, 
\begin{multline}    \sum_{\alpha}\frac{\mathcal{M}_\alpha^{(2)}\omega_k}{p(2,\alpha\vert \omega_k)}\mathds{1}^{k+1}_{2,\alpha} = \omega_k(\mathds{1}_{2,0}^{k+1} + \mathds{1}_{2,1}^{k+1}) + \left[ \mathcal{D}(L)\omega_k - \expval*{L + L^\dagger}_k\left( L\omega_k + \omega_k L^\dagger  - \expval*{L + L^\dagger}_k\omega_k\right) \right]\left(\mathds{1}_{2,0}^{k+1} + \mathds{1}_{2,1}^{k+1} \right)\delta t \\
    + \left( L\omega_k + \omega_kL^\dagger - \expval*{L + L^\dagger}_k \omega_k  \right)\left(\mathds{1}_{2,1}^{k+1} - \mathds{1}_{2,1}^{k+1} \right)\sqrt{\delta t}.
\end{multline}
Now, by noting that
\begin{equation}
    \mathbb{E}[\mathds{1}_{2,0}^{k+1} + \mathds{1}_{2,1}^{k+1}] = p(2,0\vert \omega_k) + p(2,1\vert \omega_k) = \frac{p(2)}{2}(1 - \delta t\expval*{L + L^\dagger}) + \frac{p(2)}{2}(1 + \sqrt{\delta t}\expval*{L + L^\dagger}_k) = p(2),
\end{equation}
we can set, in the limit $\delta t = \dd t \rightarrow 0$,
\begin{equation}
    \mathds{1}_{2,0}^{k+1} + \mathds{1}_{2,1}^{k+1} \longrightarrow \chi_2(t).
\end{equation}
Similarly, as
\begin{equation}
    \mathbb{E}[\mathds{1}_{2,1}^{k+1} - \mathds{1}_{2,0}^{k+1}]\sqrt{\delta t} = p(2)\expval*{L+L^\dagger}_k, 
\end{equation}
we can set, in the limit $\delta t\rightarrow 0$,
\begin{equation}
    (\mathds{1}_{2,1}^{k+1} - \mathds{1}_{2,0}^{k+1})\sqrt{\delta t} \longrightarrow \chi_2(t)\dd Z(t),
\end{equation}
where $\dd Z(t)$ is a Wiener increment with  variance $\expval*{L + L^\dagger}_t\dd t$.

Replacing the above results corresponding to the diffusive parts in the total discrete SME, and after also setting $\mathds{1}_{2,1}^{k+1} + \mathds{1}_{2,0}^{k+1} \rightarrow \chi_1(t)$, and re-centering the wiener differential to $\dd W(t)  = \dd Z(t) - \expval*{L + L^\dagger}_t\dd t$, we get Eq.~\eqref{eq:pb6}.

\section{Multiple errors at the same time}\label{sec:multiple_errors_at_the_same_time}
In this appendix, we will show that considering several errors at a time just requires the addition of the corresponding dissipative and unitary channels of each error to the fully averaged dynamics. To see this, let us consider a generic version of the detector-system Hamiltonian
\begin{equation}
    H^{(i,j)}_{\text{ds}} = Jh^{(i,j)}_{\text{ds}} =  J\ket*{\Phi_d^{j,\perp}}\bra*{\Phi_d^j}\otimes A_i + \text{h.c.}
\end{equation}
appearing with probability $p(i,j)$ and where $\ket*{\Phi_d^j} = \cos(\theta_j/2)\ket{\Phi_d} + \sin(\theta_i/2)e^{i\varphi_j}\ket*{\Phi_d^{j,\perp}}$. (Note that this type of Hamiltonian encloses all our errors except the time-dependent ones. At the end of this section, we will argue about the inclusion of time-dependent errors.) Furthermore, let us assume that the detectors are prepared in the state 
\begin{equation}
    \tilde\rho_d = \begin{pmatrix} 
    a & b \\
    b^* & 1-a
    \end{pmatrix}.
\end{equation}
In a block-matrix form, the above dimensionless Hamiltonian reads
\begin{equation}
    h^{(i,j)}_{\text{ds}} = \begin{pmatrix}
        -\sin\frac{\theta_j}{2}\cos\frac{\theta_j}{2}(A_i+A_i^\dagger) & e^{-i\varphi_j}\left( \cos^2 \frac{\theta_i}{2}A_i^\dagger - \sin^2 \frac{\theta_j}{2}A_i \right) \\
        e^{i\varphi_j}\left( \cos^2 \frac{\theta_i}{2}A_i - \sin^2 \frac{\theta_j}{2}A_i^\dagger \right) & \sin\frac{\theta_j}{2}\cos\frac{\theta_j}{2}(A_i+A_i^\dagger)
    \end{pmatrix} \equiv \begin{pmatrix}
        \Lambda_{ij} & \Gamma_{ij} \\
        \Gamma_{ij}^\dagger & -\Lambda_{ij}
    \end{pmatrix}.
\end{equation}

The averaged dynamics of the steered system is given by
\begin{equation}
    \rho_s(t+\delta t)  = \rho_s(t) + \sum_{ij}p(i,j)\Tr_d\left[\left( -iJ\delta \ad h^{(i,j)}_{\text{ds}} - \frac{J^2\delta t^2}{2}\ad^2h^{(i,j)}_{\text{ds}} \right)\tilde\rho_d \otimes \rho_s(t) \right] + \mathcal{O}(J^3\delta t^3).
\end{equation}
Upon the replacement of $h^{(i,j)}_{\text{ds}}$ in this equation, we can find the formal derivative
\begin{multline}\label{eq:longeq}
    \partial_t \rho_s(t) = \lim_{\delta t \to 0} \, -iJ a\sum_{i,j}p(i,j)[\Lambda_{ij},\rho_s(t)] +iJ (1-a)\sum_{i,j}p(i,j)[\Lambda_{ij},\rho_s(t)] -iJ\sum_{i,j}p(i,j)[b^*\Gamma_{ij}+\hc,\rho_s(t)]  \\ + J^2\delta t\sum_{ij}p(i,j)\Biggl[ \mathcal{D}(\Lambda_{ij})+ a\mathcal{D}(\Gamma^\dagger_{ij})+(1-a)\mathcal{D}(\Gamma_{ij}) +b^*\mathcal{D}(\Gamma_{ij},\Lambda_{ij}) + b\mathcal{D}(\Lambda_{ij},\Gamma_{ij})\\ -b^*\mathcal{D}(\Lambda_{ij},\Gamma^\dagger_{ij}) - b\mathcal{D}(\Gamma^\dagger_{ij},\Lambda_{ij})   \Biggr]\rho_s(t),
\end{multline}
where $\mathcal{D}(A,B)\rho \coloneqq A\rho B^\dagger - \frac{1}{2}\{B^\dagger A,\rho \}$ is a mixed dissipative channel.

Following the sufficient condition of the WM limit as in Appendix~\ref{sec:A1}, we must require that $\mathcal{O}(\Lambda_{ij}) = \mathcal{O}(b) = \mathcal{O}(\sqrt{\delta t})$ and  $\mathcal{O}(\Gamma_{ij}) = \mathcal{O}(1).$ Hence, we will perform the following rescalings: 
\begin{equation}\label{eq:rescalings}
    \Lambda_{ij} = \lambda_{ij} \sqrt{\frac{\delta t}{\gamma}}, \quad b = \kappa \sqrt{\frac{\delta t}{\gamma}}e^{i\phi},
\end{equation}
where $\kappa, \phi \in \mathbb{R}$. After replacing Eq.~\eqref{eq:rescalings} in Eq.~\eqref{eq:longeq} and taking the limit, the resulting LE is
\begin{equation}\label{eq:generic}
    \partial_t\rho_s(t) = -i(2a-1)\sum_{ij}p(i,j)[\lambda_{ij},\rho_s(t)] -i\sum_{ij}p(i,j)[e^{i\phi}\Gamma_{ij}+\hc, \rho_s(t)] + \gamma \sum_{ij}p(i,j)\left[a \mathcal{D}(\Gamma_{ij}) + (1-a)\mathcal{D}(\Gamma^\dagger_{ij}) \right]\rho_s(t).
\end{equation}

Let us check that we can recover the previous LEs when individual errors are considered. First, let us assume that there is only one index $i$, which gives $A_i = A$, and that $\theta_j = \varphi_j = 0$ for all $j$. Hence, $\Lambda_{ij}=0$ and $\Gamma_{ij} = A$. Let also $a= 1$ in $\tilde \rho_d$. Therefore, with these conditions, we recover Eq.~\eqref{eq:a6.1}. Second, let us assume that $a= 1$ and $b=0$ in $\tilde \rho_d$, $\varphi_j = 0$ for all $j$, and $A_i = A$ for all $i$. Then, Eq.~\eqref{eq:generic} coincides with Eq.~\eqref{eq:finalle}. 

Overall, we conclude from Eq.~\eqref{eq:generic} that the coherences of $\tilde \rho_d$ (if they scale appropriately) always induce unitary dynamics, whereas its populations always induce dissipative dynamics. Let us note that an additional unitary channel is induced as long as $\tilde \rho_d$ is not maximally mixed and $\theta_i \notin \{0,\pi\}$.  

The above analysis combined both static and quenched errors, and we saw that each error contributes linearly to dissipative and unitary channels; that is, each of the channels can be obtained by considering each error at a time. Now, if we add time-dependent errors containing white noise, we concluded from Appendix~\ref{sec:jump-diffussive-sme} that they only add two extra dissipative channels. Thus, the same will occur if we consider them together with other errors. 

\end{widetext}
\bibliography{references} 

\end{document}